\documentclass[twoside]{article}

\usepackage[accepted]{aistats2022}

\usepackage{amsmath,amssymb,mathtools,dsfont}
\usepackage{graphicx}
\usepackage{color,caption,subcaption}
\usepackage{booktabs}
\usepackage{algorithm2e}
\usepackage{natbib}
\usepackage{hyperref}
\usepackage{cleveref}
\usepackage{verbatim}
\usepackage{microtype}
\usepackage{bm}
\newcommand{\eps}{\epsilon}
\newcommand{\ind}[1]{\mathds{1}\{#1\}} %

\newcommand{\rIdx}[2]{#1^{#2}} %
\newcommand{\tIdx}[2]{#1_{#2}} %
\newcommand{\loo}[2]{#1(-#2)} %

\newcommand{\singleT}{T}
\newcommand{\coupT}{\widehat{T}}
\newcommand{\targetE}{H^*}
\newcommand{\targetP}{p_{\Pi}}
\newcommand{\lotargetP}[1]{p_{\Pi_{-#1}}}
\newcommand{\burnin}{\ell}
\newcommand{\minIter}{m}
\newcommand{\nProc}{J}
\newcommand{\nEst}{V}
\newcommand{\maxK}{\widetilde{K}}
\newcommand{\GibbsCond}[1]{p_{\Pi | \loo{\Pi}{#1}}}
\newcommand{\SpMe}{\text{SplitMerge}}
\newcommand{\coupfn}[6]{\psi^{#1}_{#2}(#3,#4,#5,#6)}
\newcommand{\coupfnNoArg}{\psi}
\newcommand{\coupfnOT}{\coupfnNoArg^{\text{OT}}}
\newcommand{\coupfnIndep}{\coupfnNoArg^{\text{ind}}}
\newcommand{\nugget}{\eta}
\newcommand{\coupfnOTwN}{\coupfnNoArg^{\text{OT}}_\nugget}
\newcommand{\coup}{\gamma}
\newcommand{\onestep}[1]{\widetilde{#1}}
\newcommand{\cEst}[1]{H_{\burnin:\minIter}(\rIdx{X}{#1},\rIdx{Y}{#1})}
\newcommand{\uEst}[1]{U^{#1}}
\newcommand{\defaultNugget}{10^{-5}}
\newcommand{\timeTaken}{\xi}

\newcommand{\allSame}{\bm{1}}
\newcommand{\calP}[1]{\mathcal{P}_{#1}} %

\newcommand{\LCP}{\textrm{LCP}}
\newcommand{\CC}[2]{\textrm{CC}(#1,#2)}

\newcommand{\defined}{:=}

\newcommand{\Gaussian}{\mathcal{N}}
\newcommand{\given}{\,|\,}
\newcommand{\pr}{\mathbb{P}}
\newcommand{\E}{\mathbb{E}}

\newcommand{\distiid}{\overset{\textrm{\tiny\textrm{i.i.d.}}}{\sim}}
\newcommand{\distind}{\overset{\textrm{\tiny\textrm{indep}}}{\sim}}
\newcommand{\distNorm}[2]{\mathcal{N}(#1,#2)}
\newcommand{\distDP}[2]{\mathrm{DP}(#1,#2)}

\newcommand{\TV}[2]{\|#1-#2\|_{\mathrm{TV}}}
\newcommand{\HamDist}{\text{d}}

\newcommand{\data}{W}

\newcommand{\GibbsTime}{\beta}

\DeclareMathOperator*{\argmin}{arg\,min}

\newcommand{\BlackBox}{\rule{1.5ex}{1.5ex}}  %
\newenvironment{proof}{\par\noindent{\bf Proof\ }}{\hfill\BlackBox\\[2mm]}
\newtheorem{example}{Example} 
\newtheorem{theorem}{Theorem}
\newtheorem{lemma}{Lemma} 
\newtheorem{proposition}{Proposition} 

\newtheorem{corollary}{Corollary}
\newtheorem{definition}{Definition}

\makeatletter
\def\@opargbegintheorem#1#2#3{\trivlist
	\item[]{\bfseries #1\ #2\ (#3)} \itshape}
\makeatother

\begin{document}
	
\twocolumn[

\aistatstitle{Many Processors, Little Time: MCMC for Partitions via Optimal Transport Couplings}

\aistatsauthor{ {Tin D.} Nguyen \And {Brian L.} Trippe \And Tamara Broderick}

\aistatsaddress{MIT LIDS} ]

\begin{abstract}
	Markov chain Monte Carlo (MCMC) methods are often used in clustering since they guarantee asymptotically exact expectations in the infinite-time limit. In finite time, though, slow mixing often leads to poor performance. Modern computing environments offer massive parallelism, but naive implementations of parallel MCMC can exhibit substantial bias. In MCMC samplers of continuous random variables,  Markov chain couplings can overcome bias. But these approaches depend crucially on paired chains meetings after a small number of transitions. We show that straightforward applications of existing coupling ideas to discrete clustering variables fail to meet quickly. This failure arises from the ``label-switching problem'': semantically equivalent cluster relabelings impede fast meeting of coupled chains. We instead consider chains as exploring the space of partitions rather than partitions' (arbitrary) labelings. Using a metric on the partition space, we formulate a practical algorithm using optimal transport couplings. Our theory confirms our method is accurate and efficient. In experiments ranging from clustering of genes or seeds to graph colorings, we show the benefits of our coupling in the highly parallel, time-limited regime.

\end{abstract}

\section{INTRODUCTION} \label{sec:intro}
Markov chain Monte Carlo (MCMC) is widely used in applications for exploring distributions over clusterings, or partitions, of data. For instance, \citet{prabhakaran2016dirichlet} use MCMC to approximate a Bayesian posterior over clusters of gene expression data for ``discovery and characterization of cell types'';
\citet{chen2019improved} use MCMC to approximate the number of $k$-colorings of a graph; and \citet{DeFord2021Recombination} use MCMC to identify partisan gerrymandering via partitioning of geographical units into districts.
An appealing feature of MCMC for many applications is that it yields asymptotically exact expectations in the infinite-time limit. 
However, real-life samplers must always be run in finite time, and MCMC mixing is often prohibitively slow in practice. While this slow mixing has led some practitioners to turn to other approximations such as variational Bayes \citep{blei2006variational}, these alternative methods can yield arbitrarily poor approximations of the expectation of interest \citep{huggins2020validated}.

A different approach is to speed up MCMC, e.g.\ by taking advantage of recent computational advantages.
While wall-clock time is often at a premium, modern computing environments increasingly offer massive parallel processing. 
For example, institute-level compute clusters commonly make hundreds of processors available to their users simultaneously \citep{reuther2018interactive}.
Recent efforts to enable parallel MCMC on graphics processing units \citep{lao2020tfp} offer to expand parallelism further, with modern commodity GPUs providing over ten thousand cores.
A naive approach to exploiting parallelism is to run MCMC separately on each processor;
we illustrate this approach on a genetics dataset (\textsc{gene}) in \Cref{fig:key} with full experimental details in \Cref{sec:applications}.
One might either directly average the resulting estimates across processors (red solid line in \Cref{fig:key}) or use a robust averaging procedure (red dashed line in \Cref{fig:key}).
Massive parallelism can be used to reduce variance of the final estimate but does not mitigate the problem of bias, so the final estimate does not improve substantially as the number of processes increases.

\begin{figure}[t]
	\centering
	\includegraphics[width=0.8\linewidth]{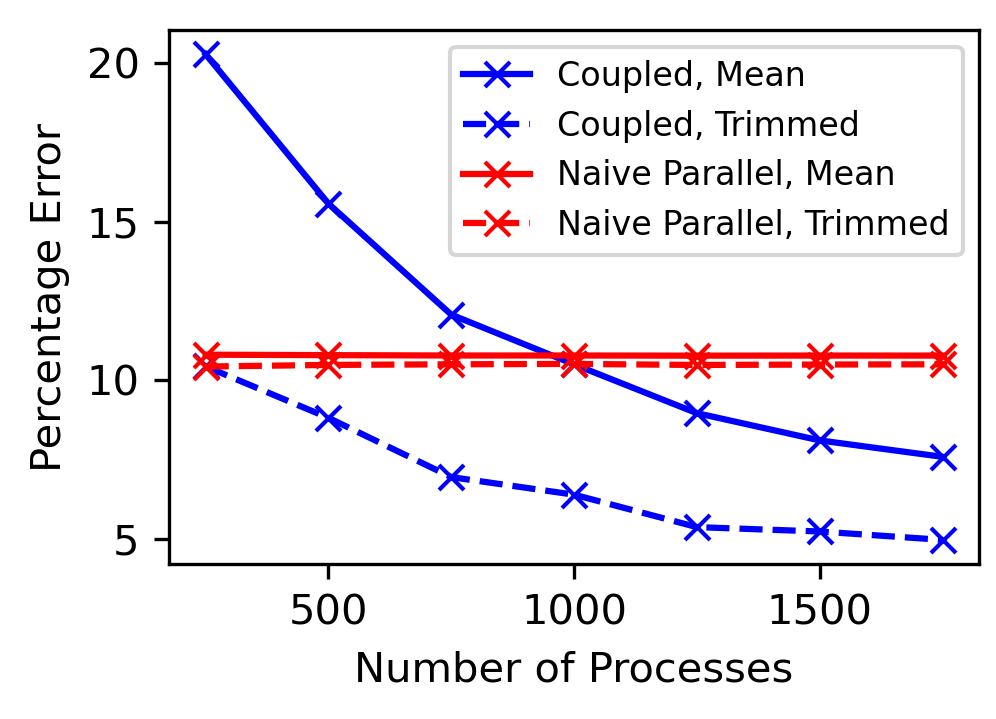}
	\caption{Lower error at high process count using our estimator (blue) versus using naive parallelism (red). %
	For details, see \Cref{s-sec:accurate-estimation}.}
	\label{fig:key}
\end{figure}

Recently, \citet{jacob2020unbiased} built on the work of \citet{glynn2014exact} to eliminate bias in MCMC with a \emph{coupling}. The basic idea is to cleverly set up dependence between two MCMC chains so that they are still practical to run and also meet exactly at a random but finite time. After meeting, these coupled chains can be used to compute an unbiased estimate of the expectation of interest. So arbitrarily large reductions in the estimate's variance due to massive parallelism translate directly into arbitrarily large reductions in total error. Since a processor's computation concludes after the chains meet, a useful coupling relies heavily on setting up coupled chains that meet quickly.

\citet{jacob2020unbiased} did not consider MCMC over partitions in particular and \citet{glynn2014exact} did not work on MCMC.
But there is existing work on couplings applied to partitions in other contexts that can be adapted into the \citet{jacob2020unbiased} framework. 
For instance, \cite{jerrum1998mathematical} uses maximal couplings on partition labelings to prove convergence rates for graph coloring, and \citet{gibbs2004convergence} uses a common random number coupling for two-state Ising models.
Though \cite{jerrum1998mathematical} was theoretical rather than practical and \citet{gibbs2004convergence} did not apply to general partition models, we can adapt the \citet{jacob2020unbiased} setup in a straightforward manner to use either coupling scheme. While this adaptation ensures asymptotically-unbiased MCMC samples, we will see (\Cref{s-sec:meeting-times}) that both schemes exhibit slow meeting times in practice. We attribute this issue to the \emph{label-switching problem}, which is well-known for plaguing MCMC over partitions \citep{jasra2005markov}. In particular, many different labelings correspond to the same partition. In the case of couplings, two chains may nearly agree on the partition but require many iterations to change label assignments, so the coupling is unnecessarily slow to meet.

Our main contribution, then, is to propose and analyze a practical coupling that uses the unbiasedness of the \citep{jacob2020unbiased} framework but operates directly in the true space of interest -- i.e., the space of partitions -- to thereby exhibit fast meeting times. In particular, we define an optimal transport (OT) coupling in the partition space (\Cref{sec:method}). For clustering models, we prove that our coupling produces unbiased estimates (\Cref{s-sec:unbiased}). We provide a big-O analysis to support the fast meeting times of our coupling (\Cref{s-sec:runtime}).We empirically demonstrate the benefits of our
coupling on a simulated analysis; on Dirichlet process mixture models applied to real genetic, agricultural, and marine life data; and on a graph coloring problem. We show that, for a fixed wall time, our coupling provides much more accurate estimates and confidence intervals than naive parallelism (\Cref{s-sec:accurate-estimation}). And we show that our coupling meets much more quickly than standard label-based couplings for partitions (\Cref{s-sec:meeting-times}). Our code is available at \url{https://github.com/tinnguyen96/partition-coupling}.

\textbf{Related work.}
Couplings of Markov chains have a long history in MCMC. 
But they either have primarily been a theoretical tool,
do not provide guarantees of consistency in the limit of many processes, or are not generally applicable to Markov chains over partitions (\Cref{apd:related_work}).
Likewise, much previous work has sought to utilize parallelism in MCMC. But this work has focused on splitting large datasets into small subsets and running MCMC separately on each subset. But here our distribution of interest is over partitions of the data; combining partitions learned separately on multiple processors seems to face much the same difficulties as the original problem (\Cref{apd:related_work}).
\citet{xu2021couplings} have also used OT techniques within the \citet{jacob2020unbiased} framework, but their focus was continuous-valued random variables. For partitions, OT techniques might most straightforwardly be applied to the label space -- and we expect would fare poorly, like the other label-space couplings in \Cref{s-sec:meeting-times}. Our key insight is to work directly in the space of partitions.

\section{SETUP} \label{sec:bck}
Before describing our method, we first review random partitions, set up Markov chain Monte Carlo for partitions -- with an emphasis on Gibbs sampling, and review the 
\citet{jacob2020unbiased} coupling framework.

\subsection{Random Partitions} \label{s-sec:random-cluster}
For a natural number $N$, a \emph{partition} of $[N]  \defined \{1,2,\ldots,N\}$ is a collection of $K \le N$ non-empty disjoint sets $\{\rIdx{A}{1},\rIdx{A}{2},\ldots,\rIdx{A}{K}\}$, whose union is $[N].$ In a clustering problem, we can think of $\rIdx{A}{k}$ as containing the data indices in a particular cluster. Let $\calP{N}$ denote the set of all partitions of $[N]$. Let $\pi$ denote an element of $\calP{N}$, and let $\Pi$ be a random partition (i.e.\ a $\calP{N}$-valued random variable) with probability mass function (p.m.f.)\ $\targetP$.
We report a summary that takes the form of an expectation: $\targetE \defined \int h(\Pi) \targetP(\Pi) d\Pi$.

As an example, consider a Bayesian cluster analysis for $N$ data points $\{\data_{n}\}_{n=1}^{N}$, with $\data_{n} \in \mathbb{R}^{D}$. A common generative procedure uses a Dirichlet process mixture model (DPMM) and conjugate Gaussian cluster likelihoods -- with hyperparameters $\alpha > 0$, $\mu_0 \in \mathbb{R}^{D}$, and $\Sigma_0,\Sigma_1$ positive definite $D\times D$ matrices.
First draw $\Pi = \pi$ with probability $\alpha^{|\pi |} \prod_{A \in \pi} (|A| -1)!  / \left[ \alpha (\alpha + 1) \cdots (\alpha + N-1) \right]$.
Then draw cluster centers $\mu_A \distiid \distNorm{\mu_0}{\Sigma_0}$ for $A \in \Pi$
and observed data $\data_{j} \given  \mu_A \distiid \distNorm{\mu_A}{\Sigma_1}$ for $j \in A$.
The distribution of interest is the Bayesian posterior over $\Pi$:
$p_\Pi(\pi) := \Pr(\Pi=\pi\given \data)$.
A summary $\targetE$ of interest might be the posterior mean of the number of clusters for $N$ data points or
of the proportion of data in the largest cluster; see \Cref{apd:functions} for more discussion.

An assignment of data points to partitions is often encoded in a vector of labels.
E.g., one might represent $\pi = \{\{1,2\}, \{3\}\}$ with the vector $z = [1, 1, 2]$; $z$ indicates that data points $1$ and $2$ are in the same cluster (arbitrarily labeled $1$ here) while point $3$ 
is in a different cluster (arbitrarily labeled $2$).
The partition can be recovered from the labeling, but the labels themselves are ancillary to the partition and, as we will see, can introduce unnecessary hurdles for fast MCMC mixing.

\subsection{Markov Chain Monte Carlo}
In the DPMM example and many others, the exact computation of the summary $\targetE$ is intractable,
so Markov chain Monte Carlo provides an approximation.
In particular, let $X_t$ (for any $t$) denote a random partition; suppose we have access to a Markov chain $\{\tIdx{X}{t}\}_{t=0}^\infty$ with starting value $X_0$ drawn according to some initial distribution
and evolving according to a transition kernel $\tIdx{X}{t} \sim T(\tIdx{X}{t-1}, \cdot)$ stationary with respect to $\targetP$.
Then we approximate $\targetE$ with the empirical average of samples: $T^{-1} \sum_{t=1}^{T} h(X_t)$.

We focus on Gibbs samplers in what follows -- since they are a convenient and popular choice for partitions \citep{maceachern1994estimating,neal2000markov,valpine2017programming}.
We also extend our methods to more sophisticated samplers, such as split-merge samplers \citep{jain2004split}, that use Gibbs samplers as a sub-routine; see \Cref{s-sec:split-merge}.
To form a Gibbs sampler on the partition itself rather than the labeling, we first introduce some notation.
Namely, let $\loo{\pi}{n}$ and $\loo{\Pi}{n}$ denote $\pi$ and $\Pi$, respectively, with data point $n$ removed.
For example, if $\pi = \left\{\{1, 3\}, \{ 2\}\right\}$, then $\loo{\pi}{1} = \left\{ \{3 \},\{ 2\} \right\}$.

With this notation, we can write the \emph{leave-out} conditional distributions of the Gibbs sampler
as $\GibbsCond{n}$. In particular, take a random partition $X$. Suppose $\loo{X}{n}$ has $K -1$ elements. Then 
the $n$th data point can either be added to an existing element or form a new element in the partition. Each of these $K$ options forms a new partition; call the new partitions $\{\rIdx{\pi}{k}\}_{k=1}^{K}$. It follows that there exist $a_k \ge 0$ such that
\begin{equation} \label{eq:GibbsCond-as-Dirac-sum}
	\GibbsCond{n}(\cdot \given \loo{X}{n}) = \sum_{k=1}^{K}  \tIdx{a}{k} \delta_{\rIdx{\pi}{k}}(\cdot),
\end{equation}
where $\delta_{\rIdx{\pi}{k}}$ denotes a Dirac measure on $\rIdx{\pi}{k}.$
When $\targetP$ is available up to a proportionality constant, it is tractable to compute or sample from $\GibbsCond{n}$.

\Cref{alg:single-sweep} shows one sweep of the resulting Gibbs sampler. 
For any $X$, the transition kernel $T(X, \cdot)$ for this sampler's Markov chain is the distribution of the output, $\onestep{X}$, of \Cref{alg:single-sweep}.
\begin{algorithm}[h!]
	\caption{Single Gibbs Sweep}
	\label{alg:single-sweep}
	\LinesNumbered
	\KwIn{Target $\targetP$. Current partition $X.$}
	$\onestep{X} \leftarrow X$  \\
	\For{$n\leftarrow 1$ \KwTo $N$}{
		$\onestep{X} \sim \GibbsCond{n}( \cdot \given \loo{\onestep{X}}{n} )$ \\
	}
	\text{Return} $\onestep{X}$
\end{algorithm}

\subsection{An Unbiased Estimator} \label{s-sec:unbiased-review}
\citep{jacob2020unbiased} show how to construct an unbiased estimator of $\targetE$ for some Markov chain $\{X_t\}$ when an additional Markov chain $\{Y_t\}$ with two properties is available. First, $Y_t \given Y_{t-1}$ must also evolve using the same transition $T(\cdot, \cdot)$ as $\{X_t\}$, so that $\{Y_t\}$ is equal in distribution to $\{\tIdx{X}{t}\}$. Second, there must exist a random \emph{meeting time} $\tau < \infty$ with sub-geometric tails such that the two chains meet exactly at time $\tau$ ($\tIdx{X}{\tau} = \tIdx{Y}{\tau-1}$) and remain faithful afterwards (for all $t\ge\tau$, $\tIdx{X}{t} = \tIdx{Y}{t-1}$). When these properties hold, the following provides an unbiased estimate of $\targetE$:
\begin{equation} \label{eqn:unbiased_estimate}
	\begin{aligned}
	& H_{\burnin:\minIter}(X,Y) \coloneqq
	\underbrace{\frac{1}{\minIter-\burnin+1} \sum_{t=\burnin}^\minIter h(\tIdx{X}{t})}_{\text{Usual MCMC average}} + {} \\
	& \quad \underbrace{\sum_{t=\burnin+1}^{\tau - 1} \text{min}\left(1, \frac{t-\burnin}{\minIter-\burnin+1}\right)\left\{h(\tIdx{X}{t}) - h(\tIdx{Y}{t-1}) \right\},}_{\text{Bias correction}} 
	\end{aligned}
\end{equation}
where $\burnin$ is the burn-in length and $\minIter$ sets a minimum number of iterations \citep[Equation 2]{jacob2020unbiased}.
$\burnin$ and $\minIter$ are hyperparameters that impact the runtime and variance of $H_{\burnin:\minIter}$;
for instance, smaller $\minIter$ is typically associated with smaller runtimes but larger variance.
\citet[Section 6]{jacob2020unbiased} recommend setting $\burnin$ to be a large quantile of the meeting time and $\minIter$ as a multiple of $\burnin$.
We follow these recommendations in our work.

One interpretation of \Cref{eqn:unbiased_estimate} is as the usual MCMC estimate plus a bias correction.
Since $H_{\burnin:\minIter}$ is unbiased, a direct average of many copies of $H_{\burnin:m}$ computed in parallel can be made to have arbitrarily small error (for estimating $\targetE$).
It remains to apply the idea from \Cref{eqn:unbiased_estimate} to partition-valued chains.

\subsection{Couplings}
\label{sec:couplings}
To create two chains of partitions that evolve together, we will need a joint distribution over partitions from both chains that respects the marginals of each chain. To that end, we define a coupling.
\begin{definition} \label{def:coupling}
	A \emph{coupling} $\coup$ of two discrete distributions, $\sum_{k=1}^{K}  \tIdx{a}{k} \delta_{\rIdx{\pi}{k}}(\cdot)$ and $\sum_{k^{\prime}=1}^{K^{\prime}}  \tIdx{b}{k^{\prime}} \delta_{\rIdx{\nu}{k^{\prime}}}(\cdot)$, is a distribution on the product space,
	\begin{equation} \label{eq:coupfn}
		\coup(\cdot) = \sum_{k} \sum_{k^{\prime}} \rIdx{u}{k,k^{\prime}} \delta_{(\rIdx{\pi}{k}, \rIdx{\nu}{k^{\prime}})}(\cdot),
	\end{equation}
	that satisfies the marginal constraints \\
	$
		\sum_{k} \rIdx{u}{k,k^{\prime}} = b_{k^{\prime}}, \hspace{10pt} \sum_{k^{\prime}} \rIdx{u}{k,k^{\prime}}  = a_k, \hspace{10pt} 0 \leq \rIdx{u}{k,k^{\prime}} \leq 1.
	$
\end{definition}

\section{OUR METHOD}\label{sec:method}
We have just described how to achieve unbiased estimates when two chains with a particular relationship are available. It remains to show that we can construct these chains so that they meet quickly in practice. First, we describe a general setup for a coupling of two Gibbs samplers over partitions in \Cref{s-sec:gibbs_partition_coupling}. Our method is a special case where we choose a coupling function that encourages the two chains to meet quickly (\Cref{s-sec:algo}). We extend our coupling to split-merge samplers in \Cref{s-sec:split-merge}. We employ a variance reduction procedure to further improve our estimates (\Cref{s-sec:var-redux}).

\subsection{Coupling For Gibbs On Partitions}\label{s-sec:gibbs_partition_coupling}

Let $X, Y$ be two partitions of $[N].$ By \Cref{eq:GibbsCond-as-Dirac-sum}, we can write $\GibbsCond{n}(\cdot \given \loo{X}{n}) = \sum_{k=1}^{K}  \tIdx{a}{k} \delta_{\rIdx{\pi}{k}}(\cdot)$ for some $K$ and tuples $(\tIdx{a}{k}, \rIdx{\pi}{k})$. And we can write $\GibbsCond{n}(\cdot \given \loo{Y}{n}) = \sum_{k^{\prime}=1}^{K^{\prime}}  \tIdx{b}{k^{\prime}} \delta_{\rIdx{\nu}{k^{\prime}}}(\cdot)$ for some $K^{\prime}$ and tuples $(\tIdx{b}{k^{\prime}}, \rIdx{\nu}{k^{\prime}})$. We say that a coupling function is any function that returns a coupling for these distributions.
\begin{definition} \label{def:coupfn}
	A \emph{coupling function} $\coupfnNoArg$ takes as input a target $\targetP$, a leave-out index $n$, and partitions $X,Y.$ It returns a coupling $\coup=\coupfn{}{}{\targetP}{n}{X}{Y}$ of $\GibbsCond{n}(\cdot \given \loo{X}{n})$ and $\GibbsCond{n}(\cdot \given \loo{Y}{n})$.
\end{definition}

Given a coupling function $\coupfnNoArg$, \Cref{alg:coupled-sweep} gives the coupled transition from the current pair of partitions $(X,Y)$ to another pair $(\onestep{X}, \onestep{Y}).$ Repeating this algorithm guarantees the first required property from the \citet{jacob2020unbiased} construction in \Cref{s-sec:unbiased-review}: co-evolution of the two chains with correct marginal distributions. It remains to show that we can construct an appropriate coupling function and that the chains meet (quickly).
\begin{algorithm}[h]
	\caption{Coupled Gibbs Sweep}
	\label{alg:coupled-sweep}
	\LinesNumbered
	\KwIn{Target $\targetP$. Coupling function $\coupfnNoArg.$ Current partitions $X$ and $Y$.}
	$\onestep{X} \leftarrow X, \onestep{Y} \leftarrow Y$  \\
	\For{$n\leftarrow 1$ \KwTo $N$}{
		 $\coup \leftarrow \coupfn{}{}{\targetP}{n}{\onestep{X}}{\onestep{Y}}$ \\
		$(\onestep{X}, \onestep{Y})  \sim \coup$\\
	}
	\text{Return} $\onestep{X}, \onestep{Y}$
\end{algorithm}

\subsection{An Optimal Transport Coupling} \label{s-sec:algo}

We next detail our choice of coupling function; namely, we start from an optimal transport (OT) coupling and add a nugget term for regularity. For a distance $\HamDist$ between partitions, the OT coupling function $\coupfnNoArg^{\text{OT}} = \coupfn{\text{OT}}{}{\targetP}{n}{X}{Y}$ minimizes the expected distance between partitions 
after one coupled Gibbs step given partitions $X, Y$ and leave-out index $n$.
Using the notation of \Cref{sec:couplings,s-sec:gibbs_partition_coupling}, 
we define
\begin{equation} \label{eq:OT-coupling}
	\coupfnNoArg^{\text{OT}} \coloneqq \argmin_{\text{couplings } \coup} \sum_{k=1}^{K} \sum_{k^\prime=1}^{K^\prime} \rIdx{u}{k,k^\prime} \HamDist(\rIdx{\pi}{k}, \rIdx{\nu}{k^\prime}).
\end{equation}
To complete the specification of $\coupfnNoArg^{\text{OT}}$, we choose a metric $\HamDist$ on partitions that was introduced by \citet{mirkin1970measurement} and \citet{rand1971objective}:
\begin{equation} \label{eqn:distance_between_partitions}
\HamDist(\pi, \nu) 
= \sum_{A\in \pi}|A|^2 
+  \sum_{B \in \nu} |B|^2
-2\sum_{A\in \pi, B\in\nu} |A\cap B|^2.
\end{equation}
Observe that $\HamDist(\pi, \nu)$ is zero when $\pi = \nu$. More generally, we can construct a graph from a partition by treating the indices in $[N]$ as vertex labels and assigning any two indices in the same partition element to share an edge; then $\HamDist/2$ is equal to the Hamming distance between the adjacency matrices implied by $\pi$ and $\nu$ \citep[Theorems 2--3]{mirkin1970measurement}. The principal trait of $\HamDist$ for our purposes is that $\HamDist$ steadily increases as $\pi$ and $\nu$ become more dissimilar. In \Cref{apd:metric}, we discuss other potential metrics and show that an alternative with similar qualitative behavior yields essentially equivalent empirical results.

In practice, any standard optimal transport\footnote{We note that the optimization problem defining \Cref{eq:OT-coupling} is an \emph{exact} transport problem, not an entropically-regularized transport problem \citep{cuturi2013sinkhorn}. Hence the marginal distributions defined by $\coupfnOT$ automatically match the inputs $\GibbsCond{n}(\cdot \given \loo{X}{n})$ and $\GibbsCond{n}(\cdot \given \loo{Y}{n})$, without need of post-processing.} solver can be used in $\coupfnNoArg^{\text{OT}}$, and we discuss our particular choice in more detail in \Cref{s-sec:runtime}.
To prove unbiasedness of a coupling (\Cref{thm:verification}), it is convenient to ensure that every joint setting of $(X,Y)$ is reachable from every other joint setting in the sampler. As we discuss after \Cref{thm:verification} and in \Cref{apd:verify-proof}, adding a small nugget term to the coupling function accomplishes this goal.
To that end, define the independent coupling $\coupfnIndep$ to have atom size $\rIdx{u}{k,k^\prime} = a_k b_{k^\prime}$ at $(\pi^k, \nu^{k^{\prime}})$. Let $\nugget \in (0,1)$.
Then our final coupling function $\coupfnOTwN = \coupfnOTwN(\targetP,n,X,Y)$ equals
\begin{equation}\label{eq:OT-w-nugget}
	\begin{cases} \coupfnOT(X,X) &\text{ if } X = Y \\
		(1-\nugget)\coupfnOT(X,Y) + \nugget \coupfnIndep(X,Y) &\text{ else, }
	\end{cases} 
\end{equation}
where we elide the dependence on $\targetP,n$ for readability.
In practice, we set $\nugget$ to $\defaultNugget$, so the behavior of $\coupfnOTwN$ is dominated by $\coupfnOT$.

As a check, notice that when two chains first meet, the behavior of $\coupfnOTwN$
reverts to that of $\coupfnOT$. Since there is a coupling with 
expected distance zero, that coupling is chosen as the minimizer in $\coupfnOT$. Therefore,
the two chains remain faithful going forward.

\subsection{Extension To Other Samplers} \label{s-sec:split-merge}

With $\coupfnOTwN$, we can also couple samplers that use Gibbs sampling as a sub-routine; to illustrate, we next describe a coupling for a split-merge sampler \citep{jain2004split}. 
Split-merge samplers pair a basic Gibbs sweep with a Metropolis-Hastings (MH) move designed to facilitate larger-scale changes across the clustering. In particular, the MH move starts from partition $X$ by selecting a pair of distinct data indices $(i,j)$ uniformly at random. If $i$ and $j$ belong to the same cluster, the sampler proposes to split this cluster. Otherwise, the sampler proposes to merge together the two clusters containing $i$ and $j$. The proposal is accepted or rejected in the MH move.
For our purposes, we summarize the full move, including proposal and acceptance but conditional on the choice of $i$ and $j$, as $\onestep{X} \sim \SpMe(i,j,X)$.
One iteration of the split-merge sampler is identical to \Cref{alg:single-sweep}, except that between lines 1 and 2 of \Cref{alg:single-sweep}, we sample $(i,j)$ and perform $\SpMe(i,j,\onestep{X})$. 

\Cref{alg:coupled-splitMerge_sweep} shows our coupling of a split-merge sampler.
We use the same pair of indices $(i,j)$ in the split-merge moves across both the $X$ and $Y$ chains.
We use $\coupfnOTwN$ to couple at the level of the Gibbs sweeps.
\begin{algorithm}[h]
	\caption{Coupled Gibbs Sweep with Split--Merge Move}
	\label{alg:coupled-splitMerge_sweep}
	\LinesNumbered
	\KwIn{Target probability mass function (p.m.f.) $\targetP$. Current partitions $X$ and $Y$.}
	\KwOut{$\onestep{X}, \onestep{Y}$}
	$\onestep{X} \leftarrow X, \onestep{Y} \leftarrow Y$  \\
	$(i,j) \leftarrow \text{Uniformly random pair of data indices}$ \\
	$\onestep{X} \sim \SpMe(i,j,\onestep{X})$ \\
	$\onestep{Y} \sim \SpMe(i,j,\onestep{Y})$ \\
	\For{$n\leftarrow 1$ \KwTo $N$}{
		$\coup \leftarrow \coupfn{\text{OT}}{\nugget}{\targetP}{n}{\onestep{X}}{\onestep{Y}}$ \\
		$(\onestep{X}, \onestep{Y})  \sim \coup$\\
	}
	\text{Return} $\onestep{X}, \onestep{Y}$
\end{algorithm}

Gibbs samplers and split-merge samplers offer differing strengths and weaknesses.
For instance, the MH move may take long to finish; \Cref{alg:single-sweep} might run for more iterations in the same time, potentially producing better estimates sooner.
The MH move is also more complex and thus potentially more prone to errors in implementation.
In what follows, we consider both samplers; we compare our coupling to naive parallelism for Gibbs sampling in \Cref{sec:applications}, and we make the analogous comparison for split-merge samplers in \Cref{apd:extensions}.

\subsection{Variance Reduction Via Trimming} \label{s-sec:var-redux}

We have described how to generate a single estimate of $\targetE$
from \Cref{eqn:unbiased_estimate}; in practice, on the $j$th processor, we run chains $X^j$ and $Y^j$
to compute $\cEst{j}$. It remains to decide how to aggregate the observations $\{\cEst{j}\}_{j=1}^{\nProc}$ across
$\nProc$ processors.

A natural option is to report the sample mean, $\frac{1}{\nProc} \sum_{j=1}^{\nProc} \cEst{j}$.
If each individual estimate is unbiased, the squared error of the sample mean decreases to zero
at rate $1/\nProc$. And standard confidence intervals have asymptotically
correct coverage.

For finite $J$, though, there may be outliers that drive
the sample mean far from $\targetE$. To counteract the effect of outliers and achieve a lower squared error,
we also report a classical robust estimator: the trimmed mean \citep{tukey1963less}.
Recall that for $\alpha \in (0,0.5)$, the \emph{$\alpha$-trimmed mean} is the average of the observations between (inclusive) the $100\alpha$ quantile and the $100(1-\alpha)$ quantile of the observed data.
The trimmed mean is asymptotically normally distributed \citep{bickel1965robust,stigler1973asymptotic}
and provides sub-Gaussian confidence intervals \citep{lugosi2019mean}.
See \Cref{apd:trimming} for more discussion on the trimmed mean.

\section{THEORETICAL RESULTS}\label{sec:theory}
To verify that our coupling is useful, we need to check that it efficiently returns accurate estimates.
We first check that the coupled estimate $H_{\burnin:\minIter}(X,Y)$ at a single processor is unbiased -- so that 
aggregated estimates across processors can exhibit arbitrarily small squared loss.
Second, we check that there is no undue computational cost of coupling relative to a single
chain.

\subsection{Unbiasedness} \label{s-sec:unbiased}
\citet[Assumptions 1--3]{jacob2020unbiased} give sufficient conditions for unbiasedness of \Cref{eqn:unbiased_estimate}. We next use these to establish sufficient conditions that $H_{\burnin:\minIter}(X,Y)$ is unbiased when targeting a DPMM posterior.

\begin{theorem}[Sufficient Conditions for Unbiased Estimation] \label{thm:verification}
	Let $\targetP$ be the DPMM posterior in \Cref{s-sec:random-cluster}.  
	Assume the following two conditions on $\coupfnNoArg$. \\
	${} \quad$ (1) There exists $\eps > 0$ such that for all $n \in [N]$ and for all $X, Y \in  \calP{N}$ such that $X \neq Y$, the output $\coup$ of the coupling function $\coupfnNoArg$ satisfies
		\begin{equation}
			\label{eq:assume-coup-fn}
			\forall k \in [K] \textrm{ and } k^{\prime} \in [K^{\prime}], \quad \rIdx{u}{k,k^\prime} \geq \eps. 
		\end{equation}
	${} \quad$ (2) If $X = Y$, then the output coupling $\coup$ of $\coupfnNoArg$ satisfies $\coup(\onestep{X} = \onestep{Y}) = 1;$ i.e.\ the coupling is faithful. \\
	Then, the estimator in \Cref{eqn:unbiased_estimate} constructed from \Cref{alg:coupled-sweep} is an unbiased estimator for $\targetE$. Furthermore, \Cref{eqn:unbiased_estimate} has a finite variance and a finite expected computing time.
\end{theorem}

We prove \Cref{thm:verification} in \Cref{proof:verification}. Our proof exploits the discreteness of the sample space to ensure chains meet. Condition (1) roughly ensures that any joint state in the product space is reachable from any other joint state under the Gibbs sweep; we use it to establish that the meeting time $\tau$ has sub-geometric tails. Condition (2) implies that the Markov chains are faithful once they meet.

\begin{corollary}\label{corollary:opt_coupling_is_unbiased}
Let $\targetP$ be the DPMM posterior.
The \Cref{eqn:unbiased_estimate} estimator using \Cref{alg:coupled-sweep} with coupling function $\coupfnOTwN(\targetP,n,X,Y)$ is unbiased for $\targetE$.
\end{corollary}
\begin{proof}
It suffices to check \Cref{thm:verification}'s conditions.
We show $\coupfnOTwN$ is faithful at the end of \Cref{s-sec:algo}.
For a partition, the associated leave-out distributions place positive mass on all $K$ accessible atoms, so
marginal transition probabilities are lower bounded by some $\omega>0.$
The nugget guarantees each $\rIdx{u}{k,k^\prime} \geq \eta \omega^2>0.$
\end{proof}
Note that the introduction of the nugget allows us to verify the first condition of \Cref{thm:verification} is met without relying on properties specific to the optimal transport coupling.
We conjecture that one could analogously show unbiased estimates may be obtained using couplings of Markov chains defined in the label space by introducing a similar nugget to transitions on this alternative state space.
Crucially, though, we will see in \Cref{s-sec:meeting-times} that our coupling in the partition space exhibits much faster meeting times in practice than these couplings in the label space.

\subsection{Time Complexity} \label{s-sec:runtime}
The accuracy improvements of our method can be achieved only if the compute expense of coupling is not too high relative to single-chain Gibbs. In \Cref{s-sec:accurate-estimation}, we show empirically that our method outperforms naive parallel samplers run for the same wall time. Here we use theory to describe why we expect this behavior.

There are two key computations that must happen in any coupling Gibbs step within a sweep: \\
(1) computing the atom sizes $\tIdx{a}{k},\tIdx{b}{k'}$ and atom locations $\rIdx{\pi}{k}, \rIdx{\nu}{k'}$ in the sense of \Cref{def:coupling} and \Cref{def:coupfn}; \\
(2) computing the pairwise distances $\HamDist(\rIdx{\pi}{k}, \rIdx{\nu}{k^\prime})$; and solving the optimal transport problem (\Cref{eq:OT-coupling}).

Let $\GibbsTime(N,K)$ represent the time it takes to compute the Gibbs conditional $\GibbsCond{n}$ for a partition of size $K$, and let $\maxK$ represent the size of the largest partition visited in any chain, across all processors, while the algorithm runs. Then part (1) takes $O(\beta(N,\maxK))$ time to run.
For single chains, computing atom sizes and locations dominates the compute time; the computation required is of the same order, but is done for \emph{one} chain, rather than two, on each processor.
We show in \Cref{prop:OT-solution} in \Cref{apd:gibbs_sweep_proof} that part (2) can be computed in $O(\maxK^3 \log \maxK)$ time.
\Cref{prop:OT-solution} follows from efficient use of data structures; naive implementations are more computationally costly.
Note that the total running time for a full Gibbs sweep (\Cref{alg:single-sweep} or \Cref{alg:coupled-sweep}) will be $N$ times the single-step cost.

The extra cost of a coupling Gibbs step will be small relative to the cost of a single-chain Gibbs step, then, if $O(\maxK^3 \log \maxK)$ is small relative to $O(\beta(N,\maxK))$.\footnote{We show in \Cref{apd:gibbs_sweep_proof} that, while there are also initial setup costs before running any Gibbs sweep, these costs do not impact the amortized complexity.} 
As an illustrative example, consider again the DPMM application from \Cref{s-sec:random-cluster}. We start with a comparison that we suspect captures typical operating procedure, but we also consider a worst-case comparison. \\
\textbf{Standard comparison:} The direct cost of a standard Gibbs step is $\GibbsTime(N,K) = O(ND + KD^3)$ (see \Cref{prop:dense-runtime} in \Cref{apd:gibbs_sweep_proof}). By Equation 3.24 in \citet{pitman2006combinatorial}, the number of clusters in a DPMM grows a.s.\ as $O(\log N)$ as $N \to \infty$.\footnote{Two caveats: (1) If a Markov chain is run long enough, it will eventually visit all possible cluster configurations. But if we run in finite time, it will not have time to explore every collection of clusters.
So we assume $O(\log N)$ is a reasonable approximation of finite time.
(2) Also note that the $\log N$ growth is for data generated from a DPMM whereas in real life we cannot expect data are perfectly simulated from the model.} If we take $\maxK = O(\log N)$, $O(\maxK^3 \log \maxK)$ will generally be smaller than $\GibbsTime(N,K) = O(ND + KD^3)$ for sufficiently large $N$. \\
\textbf{Worst-case comparison:} 
The complexity of a DPMM Gibbs step can be reduced to $\GibbsTime(N,K) = O(KD + D^3)$
through careful use of data structures and conditional conjugacy (see \Cref{prop:dense-runtime} in \Cref{apd:gibbs_sweep_proof}).
Still, the coupling cost $O(\maxK^3 \log \maxK)$ is not much larger than the cost of this step whenever $\maxK$ is not much larger than $D$.

For our experiments, we run the standard rather than optimized Gibbs step due to its simplicity and use in existing work \citep[e.g.][]{valpine2017programming}.
In e.g.\ our gene expression experiment with $D=50$, we expect this choice has little impact on our results.
Our \Cref{prop:OT-solution} establishing $O(\maxK^3 \log \maxK)$ for the optimal transport solver applies to Orlin's algorithm \citep{orlin1993faster}. However, convenient public implementations are not available. So instead we use the simpler network simplex algorithm \citep{kelly1991minimum} as implemented by \citet{flamary2021pot}. Although \citet[Section 3.6]{kelly1991minimum} upper bound the worst-case complexity of the network simplex as $O(\maxK^5)$, the algorithm's average-case performance may be as good as $O(\maxK^2)$ \citep[Figure 6]{bonneel2011displacement}.

\section{EMPIRICAL RESULTS}\label{sec:applications}
We now demonstrate empirically that our OT coupling (1) gives more accurate estimates and confidence intervals for the same wall time and processor budget as naive parallelism and (2) meets much faster than label-based couplings.

\subsection{Models, Datasets, And Implementation} \label{s-sec:targets}
We run samplers for both clustering and graph coloring problems, which we describe next. 
We detail our construction of ground truth, sampler initialization, and algorithm hyperparameters ($\burnin$ and $\minIter$) in \Cref{apd:MC_setting}.

\textbf{Motivating examples and target models.}
For clustering, we use single-cell RNA sequencing data \citep{prabhakaran2016dirichlet},
X-ray data of agricultural seed kernels \citep{charytanowicz2010complete,dua2019uci}, physical measurements of abalone \citep{nash1994population,dua2019uci},
and synthetic data from a Gaussian mixture model. In each case, our target model is the Bayesian posterior over partitions from the DPMM.
For graph colorings, sampling from the uniform distribution on $k$-colorings of graphs is a key sub-routine in fully polynomial randomized
approximation algorithms. And it suffices to sample from the partition distribution induced by the uniform distribution on $k$-colorings, which serves as our target model; see \Cref{apd:targets} for details.

\textbf{Summaries of interest.} 
Our first summary is the mean proportion of data points in the largest cluster; we write $\LCP$ for ``largest component proportion.'' See, e.g., \citet{liverani2015premium} for its use in Bayesian analysis.
Our second summary is the co-clustering probability; we write $\CC{a}{b}$ for the probability that data points indexed by $a$ and $b$ belong to the same cluster. See, e.g., \citet{DeFord2021Recombination} for its use in redistricting.
In \Cref{apd:ppd}, we also report a more complex summary: the posterior predictive distribution, which is a quantity of interest in density estimation \citep{gorur2010dirichlet,escobar1995bayesian}.

\textbf{Dataset details.}
Our \textsc{synthetic} dataset has {300} observations and {2} covariates.
Our \textsc{gene} dataset originates from \citet{zeisel2015cell} and was previously used by \citet{prabhakaran2016dirichlet} in a DPMM-based analysis. We use a subset with {200} observations and {50} covariates to allow us to quickly iterate on experiments.
We use the unlabeled version of the \textsc{seed} dataset from \citet{charytanowicz2010complete,dua2019uci} with {210} observations and {7} covariates.
For the \textsc{abalone} dataset from \citet{nash1994population,dua2019uci}, we remove the labels and binary features, which yields {4177} observations and {7} covariates.
For graph data (\textsc{k-regular}), we use a 4-regular graph with {6} vertices; we target the partition distribution induced by the uniform distribution on 4-colorings.

\subsection{Improved Accuracy With Coupling} \label{s-sec:accurate-estimation}
In \Cref{fig:composite}, we first show that our coupling estimates and confidence intervals offer improved accuracy over
naive parallelism. To the best of our knowledge, no previous coupling paper as of this writing has compared
coupling estimates or confidence intervals to those that arise from naively parallel chains.

\begin{figure*}[t]
	\centering
	\includegraphics[width=\linewidth]{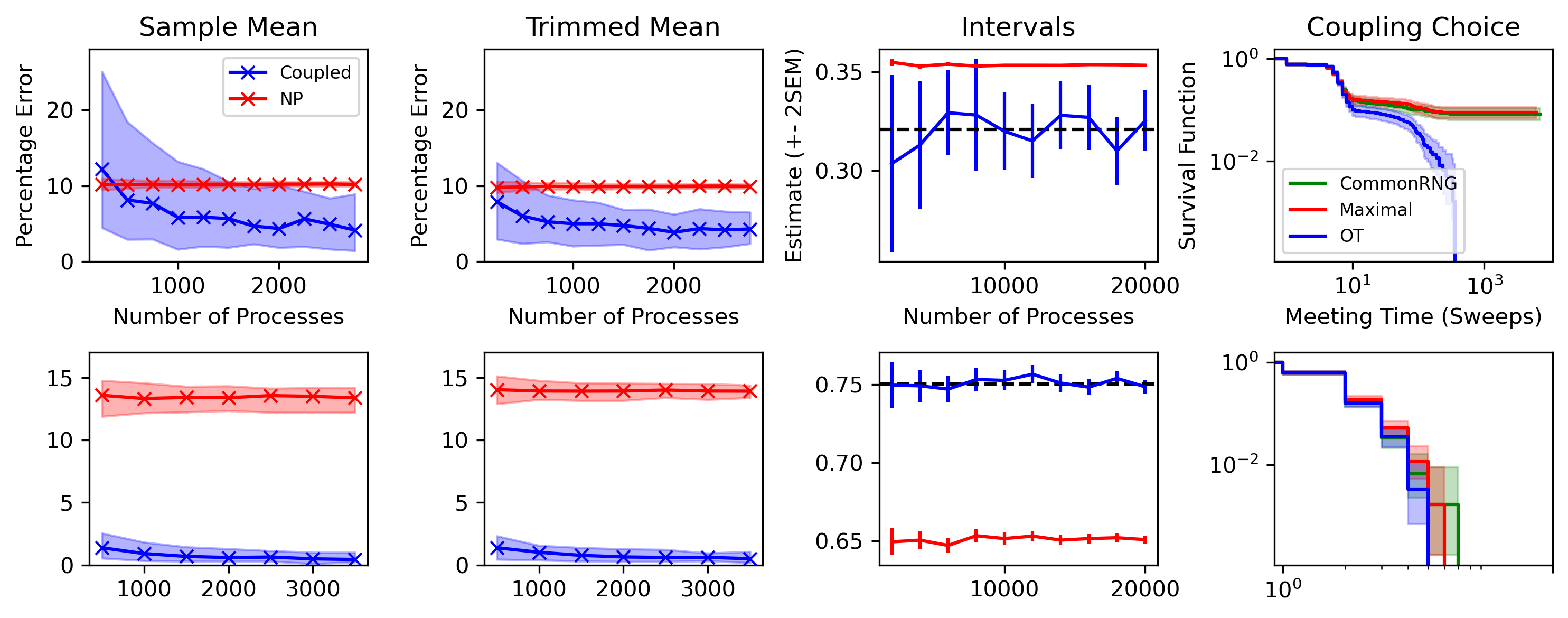}
	\caption{Top row and bottom row give results for \textsc{gene} and \textsc{k-regular}, respectively. The first two columns show that coupled chains provide better point estimates than naive parallelism. The third column shows that confidence intervals based on coupled chains are better than those from naive parallelism. The fourth column shows that OT coupling meets in less time than label-based couplings.}
	\label{fig:composite}
\end{figure*}

\textbf{Processor setup.} We give both coupling and naively parallel
approaches the same number of processors $\nProc$. We ensure equal wall time across processors as we describe next; this setup 
represents a computing system where, e.g., the user pays for total wall time, in which case we ensure equal cost between approaches.
For the coupling on the $j$th processor, we run until the chains meet and record the total time $\rIdx{\timeTaken}{j}$.
In the naively parallel case, then, we run a single chain on the $j$th processor for time $\rIdx{\timeTaken}{j}$.
In either case, each processor returns an estimate of $\targetE$. We can aggregate these estimates with a sample mean or 
trimmed estimator. Let $H_{c,J}$ represent the coupled estimate after aggregation across $J$ processors and $H_{u,J}$ represent
the naive parallel (uncoupled) estimate after aggregation across $J$ processors.
To understand the variability of these estimates, we replicate them $I$ times: $\{H_{c,J}^{(i)}\}_{i=1}^{I}$ and $\{H_{u,J}^{(i)}\}_{i=1}^{I}$.
In particular, we simulate running on 180{,}000 processors, so for each $J$, we let $I = 180{,}000 / J$; see \Cref{apd:MC_setting} for details. 
For the $i$th replicate, we compute squared error $e_{c,i} := (H_{c,J}^{(i)} - \targetE)^2$; similarly in the uncoupled case.

\textbf{Better point estimates.}
The upper left panel of \Cref{fig:composite} shows the behavior of LCP estimates for \textsc{gene}. The horizontal axis gives the number of processes $\nProc$.
The vertical value of any solid line is found by taking the square root of the median (across $I$ replicates) of the squared error and then dividing by the (positive) ground truth. Blue shows the performance of the aggregated standard-mean coupling estimate; red shows the naive parallel estimate.
The blue regions show the 20\% to 80\% quantile range. We can see that, at higher numbers of processors, the coupling estimates consistently yield a lower percentage error than the naive parallel estimates for a shared wall time.
The difference is even more pronounced for the trimmed estimates (first row, second column of \Cref{fig:composite}); here we see that, even at smaller numbers of processors, the coupling estimates consistently outperform the naive parallel estimates for a shared wall time. We see the same patterns for estimating CC(2,4) in \textsc{k-regular} (second row, first two columns of \Cref{fig:composite}) and also for \textsc{synthetic}, \textsc{seed}, and \textsc{abalone} in \Cref{sub-fig:seed-losses,sub-fig:synthetic-losses,sub-fig:abalone-losses} in \Cref{apd:all_figures}. 
We see similar patterns in the root mean squared error across replicates in \Cref{fig:key} (which pertains to \textsc{gene}) and the left panel of \Cref{sub-fig:kRegular-estimation,sub-fig:seed-estimation,sub-fig:synthetic-estimation,sub-fig:abalone-estimation} for the remaining datasets.

\begin{figure}[t]
	\centering
	\includegraphics[width=\linewidth]{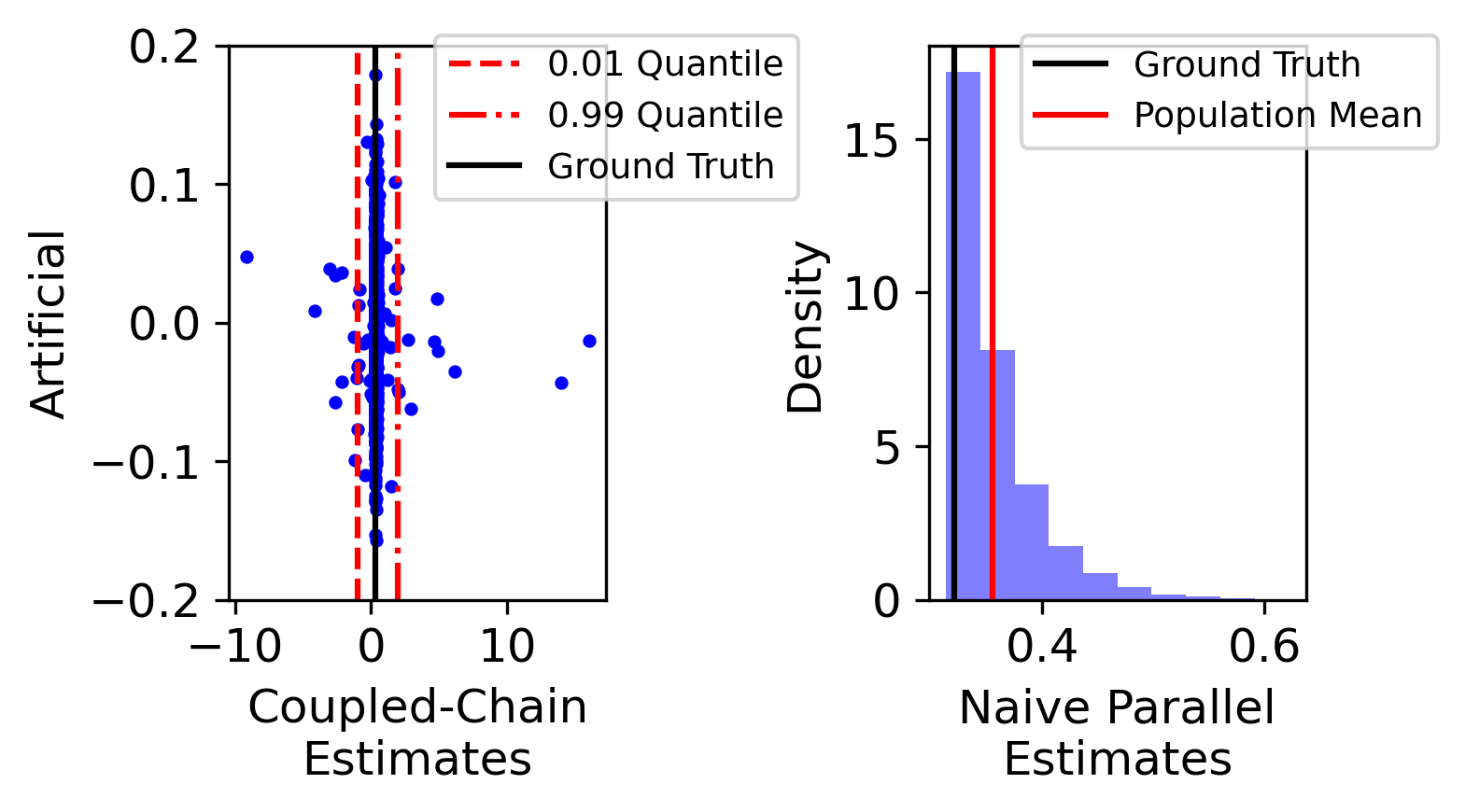}
	\caption{Coupled-chain estimates have large outliers. Meanwhile, naive parallelism estimates have substantial bias that does not go away with replication.}
	\label{fig:estimates}
\end{figure}

\Cref{fig:estimates} illustrates that the problem with naive parallelism is the bias of the individual chains, whereas only variance is eliminated by parallelism. In particular, the histogram on the right depicts the $\nProc$ estimates returned across each uncoupled chain at each processor $j$. We see that the population mean across these estimates is substantially different from the ground truth. This observation also clarifies why trimming does not benefit the naive parallel estimator: trimming can eliminate outliers but not systematic bias across processors.

By contrast, we plot the $\nProc$ coupling estimates returned across each processor $j$ as horizontal coordinates of points in the left panel of \Cref{fig:estimates}. Vertical coordinates are random noise to aid in visualization. By plotting the 1\% and 99\% quantiles of the $\nProc$ estimators, we can see that trimming will eliminate a few outliers. But the vast majority of estimates concentrate near the ground truth.

\textbf{Better confidence intervals.} The third column of \Cref{fig:composite} shows that the confidence intervals returned by coupling are also substantially improved relative to naive parallelism. The setup here is slightly different from that of the first two columns. For the first two columns, we instantiated many replicates of individual users and thereby checked that coupling generally can be counted upon to beat naive parallelism. But, in practice, an actual user would run just a single replicate. Here, we evaluate the quality of a confidence interval that an actual user would construct. We use only the individual estimates $s_j$ that make up \emph{one} $H_{c,J}$, $s_j=H_{\burnin:\minIter}(X^j,Y^j)$ (or the equivalent for $H_{u,J}$), to form a point estimate of $\targetE$ and a notion of uncertainty.

In the third column of \Cref{fig:composite}, each solid line shows the sample-average estimate aggregated across $J$ processors: $(1/J) \times \sum_{j=1}^{J} s_j$.
The error bars show $\pm2$ standard errors of the mean (SEM), where one SEM equals $\sqrt{\textrm{Var}( \{s_j\}_{j=1}^{J})/(J-1)}$.  Since the individual coupling estimators (blue) from each processor $j$ are unbiased, we expect the error bars to be calibrated, and indeed we see appropriate coverage of the ground truth (dashed black line). By contrast, we again see systematic bias in the naive parallel estimates -- and very-overconfident intervals; indeed they are so small as to be largely invisible in the top row of the third column of \Cref{fig:composite} -- i.e., when estimating LCP in the \textsc{gene} dataset.
The ground truth is many standard errors away from the naive parallel estimates. We see the same patterns for estimating CC(2,4) for \textsc{k-regular} (second row, third column of \Cref{fig:composite}).
See the right panel of \Cref{sub-fig:seed-estimation,sub-fig:synthetic-estimation,sub-fig:abalone-estimation} in \Cref{apd:all_figures} for similar behaviors in \textsc{synthetic}, \textsc{seed}, and \textsc{abalone}.

\subsection{Faster Meeting With OT Couplings} \label{s-sec:meeting-times}
Next we show that meeting times with our OT coupling on partitions are faster than 
with label-based coupling using maximal \citep{jerrum1998mathematical} and common random number generator (common RNG) \citep{gibbs2004convergence}.
We did not directly add a comparison with label-based couplings to our plots in \Cref{s-sec:accurate-estimation} 
since, in many cases, the label-based coupling chains fail to meet altogether even with a substantially larger time budget than \Cref{s-sec:accurate-estimation} currently uses.

Instead, we now provide a direct comparison of meeting times in the fourth column of \Cref{fig:composite}.
To generate each figure, we set a fixed amount of compute time budget: {10} minutes for the top row, and {2} minutes for the bottom row. Each time budget is roughly the amount of time taken to generate the ground truth (i.e., the long, single-chain runs) for each dataset.
If during that time a coupling method makes the two chains meet, we record the meeting time $\tau$; otherwise, the meeting time for that replica is right-censored, and we record the number of data sweeps up to that point.
Using the censored data, we estimate the survival functions of the meeting times using the classic Kaplan--Maier procedure \citep{kaplan1958nonparametric}.

In the clustering examples (\Cref{fig:composite} top row, fourth column and also the left panel of \Cref{sub-fig:seed-coupling,sub-fig:synthetic-coupling,sub-fig:abalone-coupling} in \Cref{apd:all_figures}), the label-based couplings' survival functions $\Pr(\tau > t)$ do not go to zero for large times $t,$ but instead they plateau around $0.1$. In other words, the label-based coupling chains fail to meet on about 10\% of attempts. 
Meanwhile, all replicas with our OT coupling successfully meet in the allotted time. 
Since so many label-based couplings fail to meet before the time taken to generate the ground truth, these label-based couplings perform worse than essentially standard MCMC.
In addition to survival functions, we also plot the distance between coupled chains -- which decreases the fastest for our OT couplings -- in the right panel of \Cref{sub-fig:gene-coupling,sub-fig:kRegular-coupling,sub-fig:seed-coupling,sub-fig:synthetic-coupling,sub-fig:abalone-coupling} in \Cref{apd:all_figures}.
As discussed in \Cref{apd:label_switching}, we believe the improvement of our OT coupling over baselines arises from using a coupling function that incentivizes decreasing the distance between partitions rather than between labelings.

Separate from accurate estimation in little time, our comparison of survival functions in the bottom row, fourth column of \Cref{fig:composite} and in \Cref{fig:meeting-ER} from \Cref{apd:more-meeting} is potentially of independent interest.
While the bottom row of \Cref{fig:composite} gives results for \textsc{k-regular}, \Cref{fig:meeting-ER} gives results on Erd\H{o}s-R\'{e}nyi random graphs.
The tightest bounds for mixing time for Gibbs samplers on graph colorings to date \citep{chen2019improved} rely on couplings on labeled representations.
Our result suggests better bounds may be attainable by considering convergence of partitions rather than labelings.

\section{CONCLUSION}
We demonstrated how to efficiently couple partition-valued Gibbs samplers using optimal transport --
to take advantage of parallelism for improved estimation. Multiple directions show promise for future work.
E.g., while we have used CPUs in our experiments here, we expect that GPU implementations will improve the applicability of our methodology.
More extensive theory on the trimmed estimator could clarify its guarantees and best practical settings.
Another direction is developing couplings for models with more complicated combinatorial structure -- such as topic modeling \citet{pritchard2000inference,blei2003latent} or feature allocations \citep{griffiths2011indian}, in which data indices can belong to more than one latent group at a time.

\section*{ACKNOWLEDGMENTS}
This work was supported by an NSF CAREER Award and an ONR Early Career Grant. BLT was also supported by NSF GRFP.

\bibliography{references}
\bibliographystyle{plainnat}

\clearpage
\appendix
\onecolumn

\section{RELATED WORK} \label{apd:related_work}
Couplings of Markov chains have a long history in MCMC.
Historically, they have primarily been a theoretical tool for analyzing convergence of Markov chains (see e.g.\ \citet{lindvall2002lectures} and references therein).
Some works prior to \citet{jacob2020unbiased} used coupled Markov chains for computation, but do not provide guarantees of consistency in the limit of many processes or are not generally applicable to Markov chains over partitions.
E.g., \citet{propp1996exact} and follow-up works generate exact, i.i.d.\ samples but require a partial ordering of the state space that is almost surely preserved by applications of an iterated random function representation of the Markov transition kernel \citep[Chapter 4.4]{jacob2020couplings}.
It is unclear what such a partial ordering looks like for the space of partitions.
\citet{neal1992circularly} proposes estimates obtained using circularly coupled chains that can be computed in parallel and aggregated, but these estimates are not unbiased and so aggregated estimates are not asymptotically exact.
Parallel tempering methods \citep{swendsen1986replica} also utilize coupled chains to improve MCMC estimates but, like naive parallelism, provide guarantees asymptotic only in the number of transitions, not in the number of processes. 

Outside of couplings, other lines of work have sought to utilize parallelism to obtain improved MCMC estimates in limited time.
To our best knowledge, that work has focused on challenges introduced by large datasets and has subsequently focused on distributing datasets across processors.
For example, \citet{rabinovich2015variational,scott2016bayes,srivastava2018scalable} explore methods running multiple chains in parallel on small subsets of a large dataset,
and \citet{lao2020tfp} proposes using data parallelism on GPUs to accelerate likelihood computations.
However, these methods offer little help in the current setting as the partition is the quantity of interest in our case; even if distributions over partitions of subsets are found at each processor, these distributions are not trivial to combine across processors.
Also, the operations that avail themselves to GPU acceleration (such as matrix multiplications) are not immediately present in Markov chains on partitions.

\section{FUNCTIONS OF INTEREST} \label{apd:functions}
We express functions of interest, $h$, in partition notation.
Suppose there are $N$ observations, and the partition is $\Pi = \{ \rIdx{A}{1}, \rIdx{A}{2}, \ldots, \rIdx{A}{K}\}$. 
To compute largest component proportion ($\LCP$), we first rank the clusters by decreasing size, $|\rIdx{A}{(1)}| \geq |\rIdx{A}{(2)}| \geq \ldots \geq |\rIdx{A}{(K)}|$, and report the proportion of data in the largest cluster: $|\rIdx{A}{(1)}|/N$.
If we are interested in the co-clustering probability of data points indexed by $j_1$ and $j_2$, then we let $h$ be the co-clustering indicator. 
Namely, if $j_1$ and $j_2$ belong to the same element of $\Pi$ (i.e.\ there exists some $A \in \Pi$ such that $j_1, j_2 \in A$), then $h(\Pi)$ equals 1; otherwise, it equals 0.

In addition to these summary statistics of the partition, we can also estimate cluster-specific parameters, like cluster centers.
For the Gaussian DPMM from \Cref{s-sec:random-cluster}, suppose that we care about the mean of clusters that contain a particular data point, say data point {1}.
This expectation is $\E(\mu_A \text{ s.t. } 1 \in A)$.
This is equivalent to $\E[\theta_i \given x]$ in the notation of \citet{maceachern1994estimating}.
In \Cref{s-sec:random-cluster}, we use $\mu_A$ to denote the cluster center for all elements $i \in A$,
while \citet{maceachern1994estimating} uses individual $\theta_i$'s to denote cluster centers for individual data points, with the possibility that $\theta_i = \theta_j$ if data points $i$ and $j$ belong in the same partition element. 
We can rewrite the expectation as $\E(  \E[\mu_A \text{ s.t. } 1 \in A \given \Pi] )$, using the law of total expectation.
$\E[\mu_A \text{ s.t. } 1 \in A \given \Pi]$ is the posterior mean of the cluster that contains data point {1}, which is a function only of the partition $\Pi$. 

\section{UNBIASEDNESS THEOREM} \label{apd:verify-proof}
\begin{lemma}[Transition kernel is aperiodic and irreducible for Gaussian DPMM] \label{lem:Tkernel}
	Denote by $\allSame$ the partition of $[N]$ where all elements belong to one cluster. For Gaussian DPMM, the transition kernel from \Cref{alg:single-sweep} satisfies
	\begin{description}
		\item[$\bullet$] For any $X \in \calP{N}$, $T(X, X) > 0.$ 
		\item[$\bullet$] For any $X \in \calP{N}$, $T(X, \allSame) > 0.$ 
		\item[$\bullet$] For any $X \in \calP{N}$, $T(\allSame, X) > 0.$ 
	\end{description}
\end{lemma}

\begin{proof}[Proof of \Cref{lem:Tkernel}] 
	For any starting $X \in \calP{N}$, we observe that there is positive probability to stay at the state after the $T(X,\cdot)$ transition i.e.\ $T(X,X) > 0.$ 
	In Gaussian DPMM, because the support of the Gaussian distribution is the whole Euclidean space (see also \Cref{eq:dpmm-cond}), when the $n$th data point is left out (resulting in the conditional $\GibbsCond{n}(\cdot \given X_{-n})$), there is positive probability that $n$th is re-inserted 
	into the same partition element of $X$ i.e.\ $\GibbsCond{n}(X \given X_{-n}) > 0.$  Since $T(X,\cdot)$ is the composition of 
	these $N$ leave--outs and re-inserts, the probability of staying at $X$ is the product of the probabilities for each $\GibbsCond{n}(\cdot \given X_{-n})$), which is overall 
	a positive number. 
	
	One series of updates that transform $X$ into $\allSame$ in one sweep is to a) assign $1$ to its own cluster
	and b) assign $2,3,\ldots,N$ to the same cluster as $1.$ This series of update also has positive probability in Gaussian DPMM. 
	
	On transforming $\allSame$ into $X$, for each component $A$ in $X$, 
	let $c(A)$ be the smallest element in the component. For instance, if $X  =\{ \{1,2\}, \{3,4\} \}$ then $c(\{1,2\}) = 1,c (\{3,4\}) = 3.$ We sort the components $A$ by their $c(A)$, to get a list $c_1 < c_2 < \ldots < c_{|X|}$.
	For each $1 \leq n \leq N$, let $l(n) = c(A)$ for the component $A$ that contains $n$.
	In the previous example, we have $c_1 = 1$ and $c_2 = 3$, while $l(1) = 1, l(2) = 1, l(3) = 3, l(4) = 3$. 
	One series of updates that transform $\allSame$ into $X$ is
	\begin{description}
		\item[$\bullet$] Initialize $j = 1$.
		\item[$\bullet$] for $1 \leq n \leq N$, if $n = c_j$, then make a new cluster with $n$ and increment $j = j + 1.$ Else, assign $n$ to the cluster that currently contains $l(n)$.
	\end{description}
	This series of update also has positive probability in Gaussian DPMM.
\end{proof}

\begin{proof}[{Proof of \Cref{thm:verification}}] \label{proof:verification}

	Because of \citet[Proposition 3]{jacob2020unbiased}, it suffices to check \citet[Assumptions 1--3]{jacob2020unbiased}. 
	
	\paragraph{Checking Assumption 1.} Because the sample space $\calP{N}$ is finite, $\max_{\pi \in \calP{N}} h(\pi)$ is finite. This means the expectation of any moment of $h$ under 
	the Markov chain is also bounded. We show that $\E[h(X^t)] \xrightarrow{t \to \infty} \targetE$ by standard ergodicity arguments.\footnote{\citet[Theorem 1]{maceachern1994estimating} 
	states a geometric ergodicity theorem for the Gibbs sampler like \Cref{alg:single-sweep} but does not provide verification of the aperiodicity, irreducibility or stationarity.} 
	\begin{description}
		
		\item[$\bullet$] Aperiodic. From \Cref{lem:Tkernel}, we know  $T(X,X) > 0$ for any $X$. This means the Markov chain is aperiodic \citep[Section 1.3]{levin2017markov}.
		
		\item[$\bullet$] Irreducible. From \Cref{lem:Tkernel}, for any $X, Y$, we know that $\singleT(X,\allSame) > 0$ and $\singleT(\allSame, Y) > 0$, meaning that $\singleT^2(X,Y) > 0.$ This means the Markov chain is irreducible. 
		
		\item[$\bullet$] Invariant w.r.t.\ $\targetP$. The transition kernel $T(X,\cdot)$ from \Cref{alg:single-sweep} leaves the target $\targetP$ invariant because each 
		leave--out conditional $\GibbsCond{n}$ leaves the target $\targetP$ invariant. If $X \sim \targetP$, then $X_{-n} \sim \lotargetP{n}$. Hence, if 
		$\onestep{X} \given X \sim \GibbsCond{n}(\cdot \given X_{-n})$ then by integrating out $X$, we have $\onestep{X} \sim \targetP.$
	\end{description}

	By \citet[Theorem 4.9]{levin2017markov}, there exists a constant $\alpha \in (0,1)$ and $C > 0$ such that
	\begin{equation*}
		\max_{\pi \in \calP{N}} \TV{T^t(\pi, \cdot)}{\targetP} \leq C \alpha^t.
	\end{equation*}
	Since the sample space is finite, the total variation bound implies that for any $\pi$, expectations under $T^t(\pi, \cdot)$ are close to expectations under $\targetP$,
	\begin{equation*}
		\max_{\pi \in \calP{N}}| \E_{X^t \given X^0 = \pi} h(X^t)  - \targetE | \leq (\max_{\pi \in \calP{N}} h(\pi)) C \alpha^t. 
	\end{equation*}
	Taking expectations over the initial condition $X^0 = \pi$,
	\begin{equation*}
		|\E_{X^t} h(X^t) - \targetE| = | \E_{X^0} [\E_{X^t \given X^0 = \pi} h(X^t)  - \targetE] | \leq  \E_{X^0} |\E_{X^t \given X^0 = \pi} h(X^t)  - \targetE]| \leq (\max_{\pi \in \calP{N}} h(\pi)) C \alpha^t.
	\end{equation*}
	Since the right hand side goes to zero as $t \to \infty$, we have shown that $\E[h(X^t)] \xrightarrow{t \to \infty} \targetE.$
	
	\paragraph{Checking Assumption 2.} To show that the meeting time is geometric, we show that there exists $\overline{\eps}$ such that for any $X$ and $Y$, under one coupled sweep from \Cref{alg:coupled-sweep} ($(\onestep{X}, \onestep{Y}) \sim \coupT(\cdot, (X,Y))$), 
	\begin{equation} \label{eq:meeting-prob}
	\pr(\onestep{X} = \onestep{Y} = \allSame \given X, Y) \geq \overline{\eps}.
	\end{equation}
	If this were true, we have that $\pr(\onestep{X} = \onestep{Y} \given X, Y) \geq \overline{\eps}$, and
	\begin{equation*}
		\pr(\tau > t) = \pr \left(  \cap_{i=0}^t X^{i+1} \neq Y^i \right) = \pr(X^1 \neq Y^0) \prod_{i=1}^{t} \pr(X^{i+1} \neq Y^i \given X^{i} \neq Y^{i-1}),
	\end{equation*}
	where we have used the Markov property to remove conditioning beyond $X^{i} \neq Y^{i-1}.$ 
	Since $\min_{X,Y} \pr(\onestep{X} = \onestep{Y} \given X, Y) \geq \overline{\eps}$, $\pr(X^{i+1} \neq Y^i \given X^{i} \neq Y^{i-1}) \leq 1 - \overline{\eps}$, meaning $\pr(\tau > t)  \leq (1-\overline{\eps})^t.$
	
	To see why \Cref{eq:meeting-prob} is true, because of \Cref{lem:Tkernel}, there exists a series of intermediate partitions $x^1, x^2, \ldots, x^{N-1}$ ($x^0 = X, x^{N} = \allSame$) such that for $1 \leq n \leq N$, $\GibbsCond{n}(x^n \given x^{n-1}_{n}) > 0.$ Likewise, there exists a series $y^1, y^2, \ldots, y^{N-1}$ for $Y.$ Because the coupling function $\coupfnNoArg$ satisfies $u^{ij} > \eps$, for any $n$, there is at least probability $\eps$ of transitioning to $(x^n,y^n)$ from $(x^{n-1}, y^{n-1})$. Overall, there is probability at least $\eps^N$ of transitioning from $(X,Y)$ to $(\allSame, \allSame)$. Since the choice of $X,Y$ has been arbitrary, we have proven \Cref{eq:meeting-prob} with 
	$\overline{\eps} = \eps^N.$
	
	\paragraph{Checking Assumption 3.} By design, the chains remain faithful after coupling.
\end{proof}

\section{TIME COMPLEXITY}\label{apd:gibbs_sweep_proof}
\begin{proposition}\label{prop:OT-solution}
	Given the atom sizes $\tIdx{a}{k},\tIdx{b}{k'}$ and atom locations $\rIdx{\pi}{k}, \rIdx{\nu}{k'}$ in the sense of \Cref{def:coupfn},
	we can compute the coupling matrix $\mu^{k,k'}$ for OT coupling function in $O(\maxK^3\log \maxK)$ time. 
\end{proposition}

\begin{proof}[Proof of \Cref{prop:OT-solution}]
	To find $\mu^{k,k'}$, we need to solve the optimization problem that is \Cref{eq:OT-coupling}. 
	However, given just the marginal distributions ($\tIdx{a}{k},\tIdx{b}{k'}$ and $\rIdx{\pi}{k}, \rIdx{\nu}{k'}$), we do not have enough ``data'' in the optimization problem, since the pairwise distances $\HamDist(\rIdx{\pi}{k}, \rIdx{\nu}{k'})$ for $k \in [K], k' \in [K']$, which define the objective function, are missing.
	We observe that it is not necessary to compute $\HamDist(\rIdx{\pi}{k}, \rIdx{\nu}{k'})$; it suffices to compute $\HamDist(\rIdx{\pi}{k}, \rIdx{\nu}{k'}) - c$ for some constant $c$ in the sense that the solution to the optimization problem in \Cref{eq:OT-coupling} is unchanged when we add a constant value to every distance.
	In particular, because for any coupling $\coup$, $\sum_{k=1}^{K} \sum_{k^\prime=1}^{K^\prime}\rIdx{u}{k,k^\prime} = 1$,
	\begin{equation} \label{eqn:reframed_OT_problem}
			\gamma^* \defined \argmin_{\text{couplings } \coup} \sum_{k=1}^{K} \sum_{k^\prime=1}^{K^\prime} \rIdx{u}{k,k^\prime} \HamDist(\rIdx{\pi}{k}, \rIdx{\nu}{k^\prime})  \\
		= \argmin_{\text{couplings } \coup}  \sum_{k=1}^{K} \sum_{k^\prime=1}^{K^\prime} \rIdx{u}{k,k^\prime} [\HamDist(\rIdx{\pi}{k}, \rIdx{\nu}{k^\prime}) - c].
	\end{equation}
	
	We now show that if we set $c = \HamDist(\loo{\pi}{n}, \loo{\nu}{n})$, then we can compute all $O(\maxK^2)$ values of $\HamDist(\rIdx{\pi}{k}, \rIdx{\nu}{k^\prime}) - c$ in $O(\maxK^2)$ time.
	First, if we use $A^k_n$ and $B^{k'}_n$ to denote the elements of $\pi^k$ and $\nu^{k'}$ respectively, containing data-point $n$, then for any $n$ we may write
	\begin{equation} \label{eqn:additional_distance}
		\begin{aligned}
			d(\rIdx{\pi}{k},\rIdx{\nu}{k'}) = d(\loo{\pi}{n}, \loo{\nu}{n}) + 
			\left[|A_n^k|^2 - (|A^k_n|-1)^2 \right] &+ \left[|B^{k'}_n|^2 - (|B^{k'}_n|-1)^2 \right] + {} \\
			&-2\left[ |A_n^k \cap B_n^{k'}|^2 - ( |A_n^k \cap B_n^{k'}|^2  - 1)^2 \right].
		\end{aligned}
	\end{equation}
	Simplifying some terms, we can also write
	\begin{equation*}
		\begin{aligned}
			d(\rIdx{\pi}{k},\rIdx{\nu}{k'}) &= d(\loo{\pi}{n}, \loo{\nu}{n}) + 
			\left[2|A_n^k| - 1\right] + \left[2|B_n^{k'}| - 1 \right] 
			-2\left[ 2|A_n^k \cap B_n^{k'}| - 1\right] \\
			&=d(\loo{\pi}{n}, \loo{\nu}{n}) + 2 \left[
			|A_n^k| + |B_n^{k'}| - 2|A_n^k\cap B_n^{k'}| \right],
		\end{aligned}
	\end{equation*}
	which means
	\begin{equation*}
			d(\rIdx{\pi}{k},\rIdx{\nu}{k'}) - d(\loo{\pi}{n}, \loo{\nu}{n}) = 2 \left[
			|A_n^k| + |B_n^{k'}| - 2|A_n^k\cap B_n^{k'}| \right].
	\end{equation*}

	At first it may seem that this still does not solve the problem, as directly computing the size of the set intersections is $O(N)$ (if cluster sizes scale as $O(N)$).
	However, \Cref{eqn:reframed_OT_problem} is just our final stepping stone.
	If we additionally keep track of sizes of intersections at every step, updating them as we adapt the partitions, it will take only constant time for each update.
	As such, we are able to form the $K \times K'$ matrix of $\HamDist(\rIdx{\pi}{k}, \rIdx{\nu}{k'}) - c$ in $O(\maxK^2)$ time.
	
	With the array of $\HamDist(\rIdx{\pi}{k}, \rIdx{\nu}{k^\prime}) - \HamDist(\loo{\pi}{n}, \loo{\nu}{n})$, we now have enough ``data'' for the optimization problem that is the optimal transport. 	
	Regardless of $N$, the optimization itself may be computed in $O(\maxK^3 \log \maxK)$ time with Orlin's algorithm \citep{orlin1993faster}.
\end{proof}

The next proposition provides estimates of the time taken to construct the Gibbs conditionals ($\GibbsTime(N,K)$) for Gaussian DPMM. 

\begin{proposition}[Gibbs conditional runtime with dense $\Sigma_0$, $\Sigma_1$] \label{prop:dense-runtime}
	Suppose the covariance matrices $\Sigma_0$ and $\Sigma_1$ are dense i.e.\ the number of non-zero entries is $\Theta(D^2).$
	The standard implementation takes time $\beta(N,K) = O(ND + KD^3)$. 
	By spending $O(D^3)$ time precomputing at beginning of sampling, and using additional data structures, the time can be reduced to $\beta(N,K) = O(KD^2 + D^3)$.
\end{proposition}

\begin{proof}[{Proof of \Cref{prop:dense-runtime}}]
	We first mention the well-known posterior formula of a Gaussian model with known covariances \citep[Chapter 2.3]{bishop2006pattern}. Namely, if $\mu \sim \distNorm{\mu_0}{\Sigma_0}$ and $\data_1, \data_2, \ldots \data_M \given \mu \distind \distNorm{\mu}{\Sigma_1}$ then $\mu \given \data_1, \ldots, \data_M$ is a Gaussian with covariance $\Sigma_c$ and mean $\mu_c$ satisfying
	\begin{equation} \label{eq:Gaussian-posterior}
		\begin{aligned}
			\Sigma_c &= (\Sigma_0^{-1} + M \Sigma_1^{-1} )^{-1} \\
			\mu_c &= \Sigma_c \left(\Sigma_0^{-1}  \mu_0 + \Sigma_1^{-1} \left[\sum_{m=1}^{M}\data_m\right]  \right).
		\end{aligned}
	\end{equation}

	Suppose $|\Pi| = K.$
	Based on the expressions for the Gibbs conditional in \Cref{eq:dpmm-cond}, the computational work involved for a held-out observation $\data_{n}$ can be broken down into three steps 
	\begin{description}
		\item[1.] Evaluating the prior likelihood $\Gaussian(\data_n \given \mu_0, \Sigma_0 + \Sigma_1)$.
		\item[2.] For each cluster $c \in \Pi(-n)$, compute $\mu_c$, $\Sigma_c$, $(\Sigma_c + \Sigma_1)^{-1}$ and the determinant of $(\Sigma_c + \Sigma_1)^{-1}$.
		\item[3.] For each cluster $c \in \Pi(-n)$, evaluate the likelihood $\Gaussian(\data_n \given \mu_c, \Sigma_c + \Sigma_1)$.
	\end{description}

	\paragraph{Standard implementation.} 
	The time to evaluate the prior $\Gaussian(\data_n \given \mu_0, \Sigma_0 + \Sigma_1)$ is $O(D^3)$, as we need to compute the precision matrix $(\Sigma_0 + \Sigma_1)^{-1}$ and its determinant. 
	With time $O(KD^3)$, we can compute the various cluster-specific covariances, precisions and determinants (where $D^3$ is the cost for each cluster).
	To compute the posterior means $\mu_c$, we need to compute the sums $\sum_{j} \data_j$ for all clusters, which takes $O(ND)$, as we need to iterate over all $D$ coordinates of all $N$ observations.
	The time to evaluate $\Gaussian(\data_n \given \mu_c, \Sigma_c + \Sigma_1)$ across clusters is $O(KD^2)$.
	Overall this leads to $O(ND + KD^3)$ runtime.
	
	\paragraph{Optimized implementation.}
	By precomputing $(\Sigma_0 + \Sigma_1)^{-1}$ (and its determinant) once at the beginning of sampling for the cost of $O(D^3)$, we can solve Step 1 in time $O(D^2)$, since that is the time to compute the quadratic form involved in the Gaussian likelihood.
	Once we have the mean and precisions from Step 2, the time to complete Step 3 is $O(KD^2)$: for each cluster, it takes time $O(D^2)$ to evaluate the likelihood, and there are $K$ clusters.
	It remains to show how much time it takes to solve Step 2. 
	We note that quantities like $\Sigma_0^{-1} \mu_0$ and $\Sigma_1^{-1}$ can also be computed once in $O(D^3)$ time at start up.
	
	Regarding the covariance $\Sigma_c$ and the precisions $(\Sigma_c + \Sigma_1)^{-1}$, at all points during sampling, the posterior covariance $\Sigma_c$ only depends on the number of data points in the cluster (\Cref{eq:Gaussian-posterior}), and leaving out data point $n$ only changes the number of points in exactly one cluster. 
	Hence, if we maintain $\Sigma_c$, $(\Sigma_c + \Sigma_1)^{-1}$ (and their determinants) for all clusters $c \in \Pi$, when a data point is left out, we only need to update one such $\Sigma_c$ and $(\Sigma_c + \Sigma_1)^{-1}$.
	Namely, suppose that $\Pi = \{ \rIdx{A}{1}, \rIdx{A}{2}, \ldots, \rIdx{A}{K}\}$. 
	We maintain the precisions are $(\Sigma(\rIdx{A}{1}) + \Sigma_1)^{-1}, (\Sigma(\rIdx{A}{2}) + \Sigma_1)^{-1}, \ldots, (\Sigma(\rIdx{A}{k}) + \Sigma_1)^{-1}$.
	Let $\rIdx{A}{j}$ be the cluster element that originally contained $n$.
	When we leave out data point $n$ to form $\Pi(-n)$, the only precision that needs to be changed is $(\Sigma(\rIdx{A}{j}) + \Sigma_1)^{-1}$.
	Let the new cluster be $\widetilde{\rIdx{A}{j}}$: the time to compute $\Sigma(\widetilde{\rIdx{A}{j}})$, $(\Sigma(\widetilde{\rIdx{A}{j}}) + \Sigma_1)^{-1}$, and its determinant is $O(D^3)$.

	Regarding the means $\mu_c$, the use of data structures similar to the covariances/precisions removes the apparent need to do $O(ND)$ computations.	
	If we keep track of $\sum_{i \in c} \data_i$ for each cluster $c$, then when data point $n$ is left out, we only need to update $\sum_{i \in c}\data_i$ for the cluster $c$ that originally contained $n$, which only takes $O(D)$.
	With the $\sum_{j}\data_j$ in place, to evaluate each of $K$ means $\mu_c$ takes $O(D^2)$; hence the time to compute the means is $O(KD^2).$
	Overall, the time spent in Step 2 is $O(KD^2 + D^3)$, leading to an overall $O(KD^2 + D^3)$ runtime.
\end{proof}

The standard implementation is used, for instance, in \citet{valpine2017programming} (see the CRP\_conjugate\_dmnorm\_dmnorm() function from NIMBLE's source code).
\citet{miller2018mixture} uses the standard implementation in the univariate case (see the Normal.jl function). 

\begin{corollary}[Gibbs conditional runtime with diagonal $\Sigma_0$, $\Sigma_1$] \label{prop:diagonal-runtime}
	Suppose the covariances $\Sigma_0$ and $\Sigma_1$ are diagonal matrices i.e.\ there are only $\Theta(D)$ non-zero entries.
	Then a standard implementation takes time $\beta(N,K) = O(ND)$.
	Using additional data structures, the time can be reduced to $\beta(N,K) = O(KD)$.
\end{corollary}

\begin{proof}[{Proof of \Cref{prop:diagonal-runtime}}]
	When the covariance matrices are diagonal, we do not incur the cubic costs of inverting $D \times D$ matrices.
	The breakdown of computational work is similar to the proof of \Cref{prop:dense-runtime}.
	
	\paragraph{Standard implementation.} 
	The covariances and precision matrices each take only time $O(D)$ to compute: as there are $K$ of them, the time taken is $O(KD)$.
	To compute the posterior means $\mu_c$, we iterate through all coordinates of all observations in forming the sums $\sum_{j} \data_j$, leading to $O(ND)$ runtime.
	Time to evaluate the Gaussian likelihoods are just $O(D)$ because of the diagonal precision matrices.
	Overall the runtime is $O(ND)$.
	
	\paragraph{Optimized implementation.}
	By avoiding the recomputation of $\sum_{j} \data_j$ from scratch, we reduce the time taken to compute the posterior means to $O(KD)$.
	Overall the runtime is $O(KD)$.
\end{proof}

\section{LABEL-SWITCHING} \label{apd:label_switching}
\subsection{Example 1} \label{s-sec:switching-ex1}

Suppose there are 4 data points, indexed by 1,2,3,4.
The labeling of the $X$ chain is $z_1 = [1,2,2,2]$, meaning that the partition is $\{\{1\}, \{2,3,4\}\}$.
The labeling of the $Y$ chain is $z_2 = [2,1,1,2]$, meaning that the partition is $\{\{1,4\}, \{2,3\}\}$.
The Gibbs sampler temporarily removes the data point $4$.
For both chains, the remaining data points is partitioned into $\{\{1\}, \{2,3\}\}$.
We denote $\rIdx{\pi}{1} = \{\{1,4\}, \{2,3\} \}$, $\rIdx{\pi}{2} = \{\{1\}, \{2,3,4\} \}$, $\rIdx{\pi}{3} = \{\{1\}, \{2,3\}, \{4\} \}$: in the first two partitions, the data point is assigned to an existing cluster while in the last partition ,the data point is in its own cluster. 
There exists three positive numbers $\rIdx{a}{1}, \rIdx{a}{2}, \rIdx{a}{3}$, summing to one, such that
\begin{equation*}
		\GibbsCond{4}(\cdot \given \loo{X}{4}) = \GibbsCond{4}(\cdot \given \loo{Y}{4}) = \sum_{k=1}^{3}  \rIdx{a}{k} \delta_{\rIdx{\pi}{k}}(\cdot).
\end{equation*}

Since the two distributions on partitions are the same, couplings based on partitions like $\coupfnOTwN$ will make the chains meet with probability $1$ in the next step.
However, this is not true under labeling--based couplings like maximal or common RNG.
In this example, the same partition is represented with different labels under either chains.
The $X$ chain represents $\rIdx{\pi}{1}, \rIdx{\pi}{2}, \rIdx{\pi}{3}$ with the labels $1$, $2$, $3$, respectively.
Meanwhile, the $Y$ chain represents $\rIdx{\pi}{1}, \rIdx{\pi}{2}, \rIdx{\pi}{3}$ with the labels $2$, $1$, $3$, respectively. 
Let $z_X$ be the label assignment of the data point in question (recall that we have been leaving out $4$) under the $X$ chain. 
Similarly we define $z_Y.$ 
Maximal coupling maximizes the probability that $z_X = z_Y.$
However, the coupling that results in the two chains $X$ and $Y$ meeting is the following
\begin{equation*}
	\Pr(z_X = u, z_Y = v) = \begin{cases}
	\rIdx{a}{3}	&\text{ if } u = v = 3 \\
	\rIdx{a}{1}	&\text{ if } u = 1, v = 2 \\
	\rIdx{a}{2}	&\text{ if } u = 2, v = 1 \\
	0	&\text{ otherwise}.
	\end{cases}
\end{equation*}
In general, $\rIdx{a}{1} \neq \rIdx{a}{2}$, meaning that the maximal coupling is different from this coupling that causes the two chains to achieve the same partition after updating the assignment of $4.$
A similar phenomenon is true for common RNG coupling. 

\subsection{Example 2} \label{s-sec:switching-ex2}
For the situation in \Cref{s-sec:switching-ex1}, the discussion of Ju et al.\ from \citet{tancredi2020unified} proposes a relabeling procedure to better align the clusters in the two partitions before constructing couplings.
Indeed, if $z_2$ were relabeled $[1,2,2,1]$ (the label of each cluster is the smallest data index in that cluster), then upon the removal of data point 4, both the label-based and partition-based couplings would agree. 
However, such a relabeling fix still suffer from label-switching problem in general, since the smallest data index does not convey much information about the cluster.
For concreteness, we demonstrate an example where the best coupling from minimizing label distances is different from the best coupling minimizing partition distances.

Suppose there are 6 data points, indexed from 1 through 6.
The partition of the $X$ chain is $\{ \{1,3,4\}, \{2,5,6\}  \}$.
The partition of the $Y$ chain is $\{ \{1,5,6\}, \{2,3,4\}  \}$.
Using the labeling rule from above, the label vector for $X$ is $z_X = [1,2,1,1,2,2]$ while that for $Y$ is $z_Y = [1,2,2,2,1,1]$. 
The Gibbs sampler temporarily removes the data point $1$.
The three next possible states of the $X$ chain are the partitions $\nu_1, \nu_2, \nu_3$ where $\nu_1 = \{\{1,3,4\}, \{2,5,6\} \}$, $\nu_2 = \{\{3,4\}, \{1,2,5,6\} \}$ and $\nu_3 = \{\{3,4\}, \{2,5,6\}, \{1\} \}$.
The labelings of data points $2$ through $6$ for all three partitions are the same; the only different between the labeling vectors are the label of data point $1$: for $\nu_1$, $z_X(1) = 1$, for $\nu_2$, $z_X(1) = 2$ and for $\nu_1$, $z_X(1) = 3$.
On the $Y$ side, the three next possible states of the $Y$ chain are the partitions $\mu_1, \mu_2, \mu_3$ where $\mu_1 = \{\{1,5,6\}, \{2,3,4\} \}$, $\mu_2 = \{\{5,6\}, \{1,2,3,4\} \}$ and $\mu_3 = \{\{5,6\}, \{2,3,4\}, \{1\}\}$.
As for the labeling of $1$ under $Y$, for $\mu_1$, $z_Y(1) = 1$, for $\mu_2$, $z_Y(1) = 2$ and for $\mu_3$, $z_Y(1) = 3$.
Suppose that the marginal assignment probabilities are the the following:
\begin{description}
	\item[$\bullet$] $\Pr(X = \nu_1) = \Pr(X = \nu_2) = 0.45, \Pr(X = \nu_3) = 0.1$.
	\item[$\bullet$] $\Pr(Y = \mu_1) = \Pr(Y = \mu_2) = 0.45, \Pr(Y = \mu_3) = 0.1$.
\end{description}

Under label-based couplings, since $\Pr(z_X(1) = a) = \Pr(z_Y(1) = a)$ for $a \in [1,2,3]$, the coupling that minimizes the distance between the labels will pick $\Pr(z_X(1) = z_Y(1)) = 1$, which means the following for the induced partitions:
\begin{equation} \label{app-eq:label-coupling}
	\Pr(X = \nu, Y = \mu) = \begin{cases}
		0.45 &\text{ if } \nu = \nu_1, \mu = \mu_1 \\
		0.45 &\text{ if } \nu = \nu_2, \mu = \mu_2 \\
		0.1 &\text{ if } \nu = \nu_3, \mu = \mu_3 \\
	\end{cases}.
\end{equation}

Under the partition-based transport coupling, the distance between partitions (\Cref{eqn:distance_between_partitions}) is the following.
\begin{center}
	\begin{tabular}{c | c c c} 
		 & $\mu_1$ & $\mu_2$ & $\mu_3$ \\ [0.5ex] 
		\hline
		$\nu_1$  & 16 & 10 & 12 \\ 
		\hline
		$\nu_2$ & 10 & 16 & 14 \\
		\hline
		$\nu_3$ & 12 & 14 & 8 \\
	\end{tabular}
\end{center}
Notice that the distances $d(\nu_1,\mu_1)$ and $d(\nu_2,\mu_2)$ are actually larger than $d(\nu_1,\mu_2)$ and $d(\nu_2,\mu_1)$: in other words, the label-based coupling from \Cref{app-eq:label-coupling} proposes a coupling with larger-than-minimal expected distance.
In fact, solving the transport problem, we find that the coupling that minimizes the expected partition distance is actually
\begin{equation} \label{app-eq:partition-coupling}
	\Pr(X = \nu, Y = \mu) = \begin{cases}
		0.45 &\text{ if } \nu = \nu_1, \mu = \mu_2 \\
		0.45 &\text{ if } \nu = \nu_2, \mu = \mu_1 \\
		0.1 &\text{ if } \nu = \nu_3, \mu = \mu_3 \\
	\end{cases}.
\end{equation}

\section{TRIMMING} \label{apd:trimming}
We consider the motivating situation in \Cref{ex:outliers}. 
This is a case where trimming outliers before taking the average yields a more accurate estimator (in terms of mean squared error) than the regular sample mean.
For reference, the RMSE of an estimator $\widehat{\mu}$ of a real-valued unknown quantity $\mu$ is
\begin{equation*}
	\sqrt{\mathbb{E} \|\widehat{\mu} - \mu\|^2}.
\end{equation*} 

\begin{example}[Mixture distribution with large outliers] \label{ex:outliers}
	For $\mu > 0$, $p < 1$, consider the mixture distribution $(0.5-p/2)\distNorm{-\mu}{1} + p\distNorm{0}{1} +  (0.5-p/2)\distNorm{\mu}{1}$. The mean is $0.$ The variance is $1 + (1-p)\mu^2.$ Therefore, the RMSE of the sample mean computed using $\nProc$ iid draws is $\sqrt{1 + (1-p)\mu^2}/\sqrt{\nProc}.$ 
\end{example}
In \Cref{ex:outliers}, increasing $\mu$, which corresponds to larger outlier magnitude, increases the RMSE.

In trimmed means (\Cref{s-sec:var-redux}), the quantity $\alpha$ determines how much trimming is done. 
Intuitively, for \Cref{ex:outliers}, if we trim about $0.5-p/2$ of the top and bottom samples from the mixture distribution in \Cref{ex:outliers}, what remain are roughly samples from $\mathcal{N}(0,1)$.
The mean of these samples should have variance only $1/\nProc$, resulting in an RMSE which does not suffer from large $\mu.$

In \Cref{fig:synthetic}, we illustrate the improvement of trimmed mean over sample mean for problems like \Cref{ex:outliers}. 
We set $p = 0.9$, $\mu = 7,$ and $\alpha = 1.2(0.5-p/2)$. 
Similar to \Cref{fig:key}, RMSE is estimated by adding another level of simulation to capture the variability across aggregates.
The left panel shows that RMSE of trimmed mean is smaller than that of sample mean.
The right panel explains why that is the case. 
Here, we box plot the trimmed mean and sample mean, where the randomness is from the iid Monte Carlo draws from the target mixture for $J = 1000.$
The variance of trimmed mean is smaller than that of sample mean, which matches the motivation for trimming. 
\begin{figure}[t]
	\centering
	\includegraphics[scale=0.8]{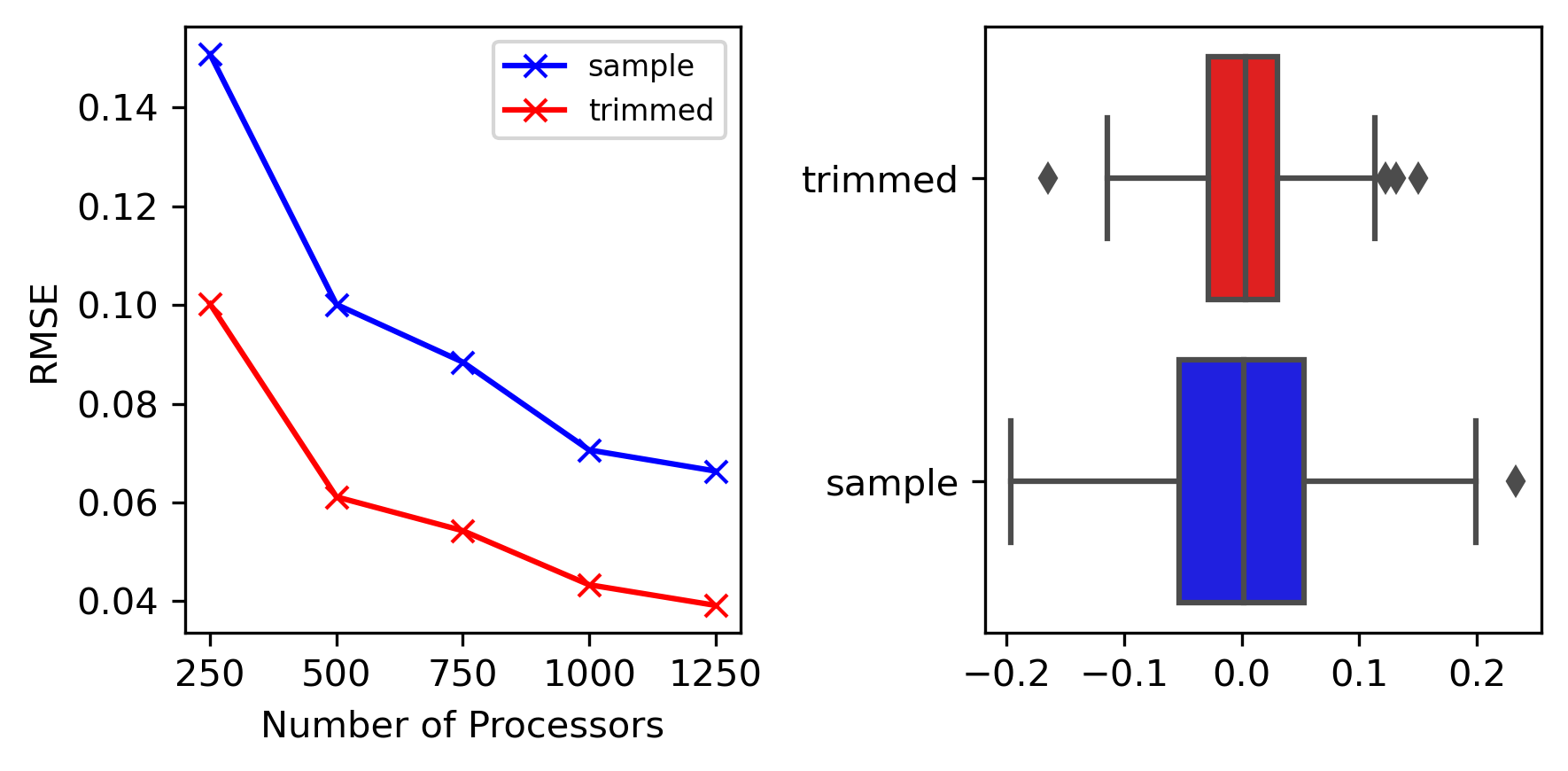}
	\caption{Trimmed mean has better RMSE than sample mean on \Cref{ex:outliers}. Left panel plots RMSE versus $\nProc.$ Right panel gives boxplots $\nProc = 1000.$}
	\label{fig:synthetic}
\end{figure}

For other situations where there exist better estimators than the sample mean, we refer to the literature on Stein's paradox \citep{stigler1990galtonian}.

\section{ADDITIONAL EXPERIMENTAL DETAILS}\label{apd:experimental_setup}
\subsection{Target Distributions And Gibbs Conditionals} \label{apd:targets}
\paragraph{DPMM.}
Denote $\Gaussian(x \given \mu, \Sigma)$ to be the Gaussian density at $x$ for a Gaussian distribution with mean $\mu$ and covariance $\Sigma.$ 
For the Gaussian DPMM from \Cref{s-sec:random-cluster}, the Gibbs conditional have the form
\begin{equation} \label{eq:dpmm-cond}
	\Pr(z_n = c \given \Pi(-n), \data_{1:N}) = \begin{cases}
		\beta\frac{\alpha}{N-1+\alpha} \Gaussian(\data_n \given \mu_0, \Sigma_0 + \Sigma_1)  &\text{ if } c \text{ is new cluster} \\
		\beta \frac{\text{size of cluster c}}{N-1+\alpha} \Gaussian(\data_n \given \mu_c, \Sigma_c + \Sigma_1) &\text{ if } c \text{ is an existing cluster}, \\
	\end{cases}
\end{equation}
where $\beta$ is a normalization constant so that $\sum_{c} \Pr(z_n = c \given \Pi(-n), \data_{1:N}) = 1$,
$c$ is an index into the clusters that comprise $\Pi(-n)$ (or a new cluster),
$\mu_c$ and $\Sigma_c$ are the posterior parameters of the cluster indexed by $c.$ See \citet{neal2000markov} for derivations.

\paragraph{Graph coloring.}
Let $G$ be an undirected graph with vertices $V=[N]$ and edges $E \subset V \otimes V,$ and let $Q=[q]$ be set of $q$ colors.
A graph coloring is an assignment of a color in $Q$ to each vertex satisfying that the endpoints of each edge have different colors.
We here demonstrate an application of our method to a Gibbs sampler which explores the uniform distribution over valid $q-$colorings of $G$,
i.e.\ the distribution which places equal mass on ever proper coloring of $G$.

To employ \Cref{alg:coupled-sweep}, for this problem we need only to characterise the p.m.f.\ on partitions of the vertices implied by the uniform distribution on its colorings.
A partition corresponds to a proper coloring only if no two adjacent vertices are in the element of the partition.
As such, we can write
$$
p_{\Pi_N}(\pi) \propto \ind{|\pi| \le q \text{  and  }  A(\pi)_{i,j}=1 \rightarrow (i, j) \not \in E, \ \forall i \ne j } { q \choose |\pi|} |\pi|!,
$$
where the indicator term checks that $\pi$ can correspond to a proper coloring
and the second term accounts for the number of unique colorings which induce the partition $\pi$.
In particular it is the product of the number of ways to choose $|\pi|$ unique colors from $Q$ ( ${q \choose |\pi|} :=\frac{q!}{|\pi|! (q-|\pi|)!}$) and the number of ways to assign those colors to the groups of vertices in $\pi$.

The Gibbs conditionals have the form
\begin{equation} \label{eq:graph-cond}
	\GibbsCond{n}(\Pi = y \given \loo{\Pi}{n}) =  \frac{ \frac{q!}{(q-|y|)!} }{\sum_{x \text{ consistent with } \loo{\Pi}{n}}  \frac{q!}{(q-|x|)! }} = \frac{ \frac{1}{(q-|y|)!} }{\sum_{x \text{ consistent with } \loo{\Pi}{n}}  \frac{1}{(q-|x|)! }}.
\end{equation}
In \Cref{eq:graph-cond}, $x$ and $y$ are partitions of the whole set of $N$ vertices. 

In implementations, to simulate from the conditional \Cref{eq:graph-cond}, it suffices to represent the partition with a color vector.
Suppose we condition on $\loo{\Pi}{n})$  i.e.\ when the colors for all but the $n$ vertex are fixed, and there are $q'$ unique colors that have been used ($q'$ can be strictly smaller than $q$).
$n$ can either take on a color in $[q']$ (as long as the color is not used by a neighbor), or take on the color $q' + 1$ (if $q' < q$). 
The transition probabilities are computed from the induced partition sizes $|x|$.

\subsection{General Markov Chain Settings} \label{apd:MC_setting}
\paragraph{Ground truth.} For clustering, we run {10} single-chain Gibbs samplers for {10,000} sweeps each; we discard the first {1,000} sweeps. For graph coloring, we also run {10} chains, but each for {100,000} sweeps and discard the first {10,000}. We compute an unthinned MCMC estimate from each chain and use the average across the 10 chains as ground truth. The standard errors across chains are very small. Dividing the errors by the purported ground truth yields values with magnitude smaller than $5 \times 10^{-3}$. In percentage error, this is less than $0.5\%$, which is orders of magnitude smaller than the percentage errors from coupled chains or naive parallel estimates.\footnote{The percentage errors for $\LCP$ are typically $0.01\%$, while percentage errors for co-clustering are typically $0.1\%$.}

\paragraph{Sampler initializations.} 
In clustering, we initialize each chain at the partition where all elements belong to the same element i.e.\ the one-component partition.
In graph coloring, we initialize the Markov chain by greedily coloring the vertices.
Our intuition suggests that coupling should be especially helpful relative to naively parallel chains when samplers require a large burn-in -- since slow mixing induces bias in the uncoupled chains. In general, one cannot know in advance if that bias is present or not, but we can try to encourage suboptimal initialization in our experiments to explore its effects. For completeness, we consider alternative initialization schemes, such as k-means, in \Cref{fig:kmeans=5-init}.

\paragraph{Choice of hyperparameters in aggregate estimates.} 
Recall that \Cref{eqn:unbiased_estimate} involves two free hyperparameters, $\burnin$ and $\minIter$, that we need to set.
A general recommendation from \citet[Section 3.1]{jacob2020unbiased} is to select $\minIter = 10 \burnin$ and $\burnin$ to be a large quantile of the meeting time distribution.
We take heed of these suggestions, but also prioritize $\minIter$'s that are small because we are interested in the time-limited regime. 
Larger $\minIter$ leads to longer compute times across both coupled chains and naively parallel chains, and the bias in naively parallel chains is more apparent for shorter $\minIter$: see \Cref{fig:diffMinIter}. 
In the naive parallel case, we discard the first $10\%$ of sweeps completed in any time budget as burn-in steps. 
In our trimmed estimates, we remove the most extreme $1\%$ of estimates (so $0.5\%$ in either directions).

\paragraph{Simulating many processes.} To quantify the sampling variability of the aggregate estimates (sample or trimmed mean across $\nProc$ processors), we first generate a large number ($\nEst = 180{,}000$) of coupled estimates $\cEst{j}$ (and $\nEst$ naive parallel estimates $\uEst{j}$, where the time to construct $\cEst{j}$ is equal to the time to construct $\uEst{j}$).\footnote{The best computing infrastructure we have access to has only $400$ processors, so we generate these $\nEst$ estimates by sequential running $nEst/400$ batches, each batch constructing $400$ estimates in parallel.}
For each $\nProc$, we batch up the $\nEst$ estimates in a consistent way across coupled chains and naive parallel, making sure that the equality between coupled wall time and naive parallel wall time is maintained.
There are $I = \nEst/\nProc$ batches.
For the $i$th batch, we combine $\cEst{j}$ (or $\uEst{j})$ for indices $j$ in the list $[(i-1)\nProc + 1,i\nProc]$ to form $H_{c,J}^{(i)}$ (or $H_{u,J}^{(i)}$) in the sense of \Cref{s-sec:accurate-estimation}.
By this batching procedure, smaller values of $\nProc$ have more batches $I$.
The largest $\nProc$ we consider for \textsc{gene}, \textsc{k-regular} and \textsc{abalone} is $2{,}750$ while that for \textsc{synthetic} and \textsc{seed} is $1{,}750$.
This mean the largest $\nProc$ has at least $57$ batches. 

To generate the survival functions (last column of \Cref{fig:composite}), we use {600} draws from the (censored) meeting time distribution by simulating {600} coupling experiments.

\subsection{Datasets Preprocessing, Hyperparameters, Dataset-Specific Markov Chain Settings} 
\paragraph{\textsc{gene} i.e.\ single-cell RNAseq.}
We extract $D=50$ genes with the most variation of $N=200$ cells.
We then take the log of the features, and normalize so that each feature has mean $0$ and variance $1$.
We target the posterior of the probabilistic model in \Cref{s-sec:random-cluster} with $\alpha = 1.0$, $\mu_0 = 0_D$, diagonal covariance matrices $\Sigma_0 = 0.5 I_D$, $\Sigma_1 = 1.3 I_D$.
Notably, this is a simplification of the set-up considered by \citet{prabhakaran2016dirichlet}, who work with a larger dataset and additionally perform fully Bayesian inference over these hyperparameters.
That the prior variance is smaller than the noise variance yields a ``challenging'' clustering problem, where the cluster centers themselves are close to each other and observations are noisy realizations of the centers. 
We set $\burnin = 30$ and $\minIter = 300.$

\paragraph{\textsc{seed} i.e.\ wheat seed measurements.}
The original dataset from \cite{charytanowicz2010complete} has 8 features; we first remove the ``target'' feature, which contains label information for supervised learning.
Overall there are $N = 210$ observations and $D = 7$ features.
We normalize each feature to have mean 0 and variance 1.
We target the posterior of the probabilistic model in \Cref{s-sec:random-cluster} with $\alpha = 1.0$, $\mu_0 = 0_D$, diagonal covariance matrices $\Sigma_0 = 1.0 I_D$, $\Sigma_1 = 1.0 I_D$.
We set $\burnin = 10$ and $\minIter = 100.$

\paragraph{\textsc{synthetic}.}
We generate $N = 300$ observations from a $4$-component mixture model in 2 dimensions.
The four cluster centers are $[-0.8,-0.8], [-0.8,0.8], [0.8,-0.8], [0.8,0.8]$
Each data point is equally likely to come from one of four components; the observation noise is isotropic, zero-mean Gaussian with standard deviation $0.5.$
These settings result in a dataset where the observations form clear clusters, but there is substantial overlap at the cluster boundaries -- see \Cref{sub-fig:dataViz-synthetic}.

On this data, we target the posterior of the probabilistic model in \Cref{s-sec:random-cluster} with $\alpha = 0.2$, $\mu_0 = 0_D$, diagonal covariance matrices $\Sigma_0 = 0.75 I_D$, $\Sigma_1 = 0.7 I_D$.
Different from \textsc{gene}, the prior variance is larger than the noise variance for \textsc{synthetic}. 
We set $\burnin = 10$, $\minIter = 100.$

\paragraph{\textsc{abalone} i.e.\ physical measurements of abalone specimens.}
The original dataset from \cite{nash1994population} has 9 features; we first remove the ``Rings'' feature, which contains label information for supervised learning, and the ``Sex'' feature, which contains binary information that is not compatible with the Gaussian DPMM generative model.
Overall there are $N = 4{,}177$ observations and $D = 7$ features.
We normalize each feature to have mean 0 and variance 1.
We target the posterior of the probabilistic model in \Cref{s-sec:random-cluster} with $\alpha = 1.0$, $\mu_0 = 0_D$, diagonal covariance matrices $\Sigma_0 = 2.0 I_D$, $\Sigma_1 = 2.0 I_D$.
We set $\burnin = 10$ and $\minIter = 100.$

\paragraph{\textsc{k-regular}.}
Anticipating that regular graphs are hard to color, we experiment with a 4-regular, 6-node graph -- see \Cref{sub-fig:dataViz-kRegular}.
The target distribution is the distribution over vertex partitions induced by uniform colorings using 4 colors.
We set $\burnin = 1$, $\minIter = 4$. 

\subsection{Visualizing Synthetic Data}
\Cref{fig:dataViz} visualizes the two synthetic datasets.
\begin{figure*}[t]
	\centering
	\begin{subfigure}[b]{0.45\linewidth}
		\centering
		\includegraphics[scale=1.0]{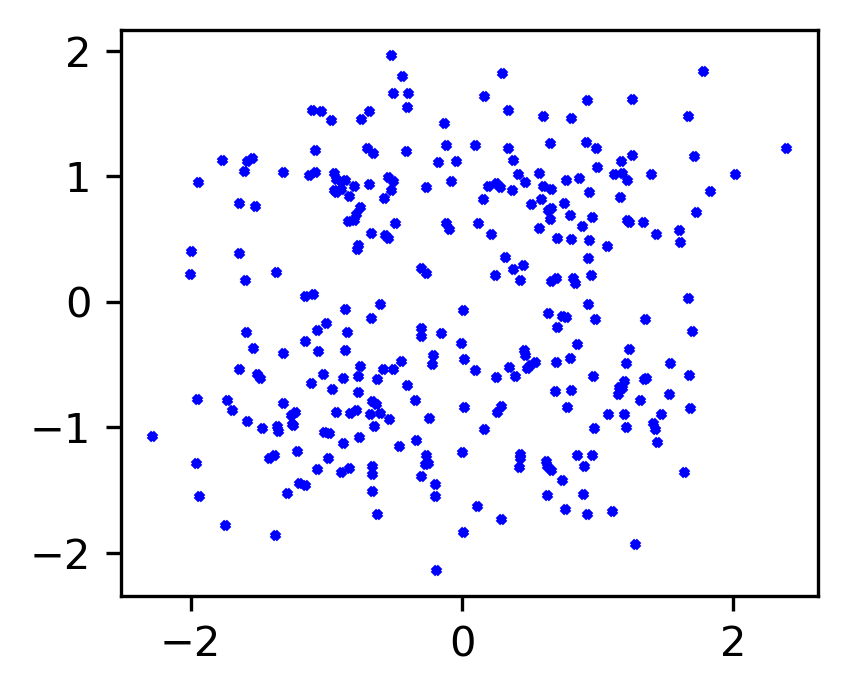}
		\caption{\textsc{synthetic} data}
		\label{sub-fig:dataViz-synthetic}
	\end{subfigure}
	\hfill
	\begin{subfigure}[b]{0.45\linewidth}
		\centering
		\includegraphics[scale=0.4]{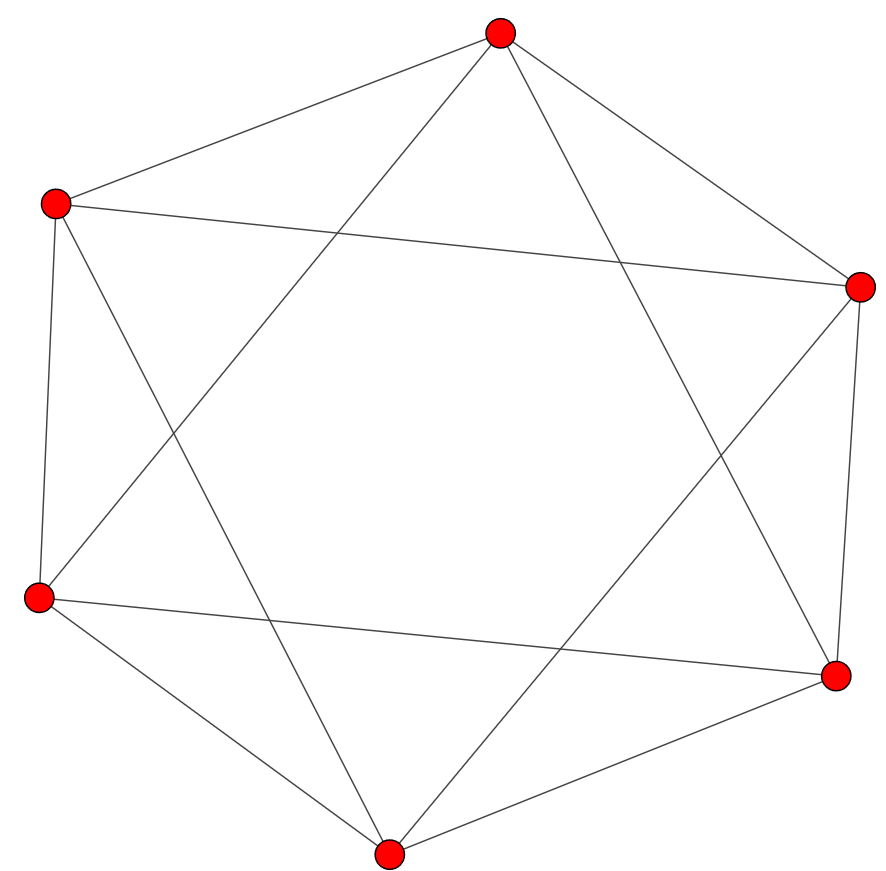}
		\caption{\textsc{k-regular} data}
		\label{sub-fig:dataViz-kRegular}
	\end{subfigure}%
	\caption{Visualizing synthetic data}
	\label{fig:dataViz}
\end{figure*}

\section{ALL FIGURES} \label{apd:all_figures}
\subsection{\textsc{gene}}
\Cref{fig:gene-all} shows results for $\LCP$ estimation on \textsc{gene} -- see \Cref{fig:co-cluster} for results on co-clustering.
The two panels that did not appear in \Cref{fig:composite} are the left panel of \Cref{sub-fig:gene-estimation} and the right panel of \Cref{sub-fig:gene-coupling}.
The left panel of \Cref{sub-fig:gene-estimation} is the same as \Cref{fig:key}: the y-axis plots the RMSE instead of the range of losses. 
As expected from the bias-variance decomposition, the RMSE for coupled estimates decreases with increasing $\nProc$ because of unbiasedness, while the RMSE for naive parallel estimates does not go away because of bias.
The right panel of \Cref{sub-fig:gene-coupling} plots typical $\HamDist$ distances between coupled chains under different couplings as a function of the number of sweeps done.
$\HamDist$ decreases to zero very fast under OT coupling, while it is possible for chains under maximal and common RNG couplings to be far from each other even after many sampling steps.

\begin{figure}[h]
	\centering
	\begin{subfigure}[h]{0.48\linewidth}
		\includegraphics[scale=0.67]{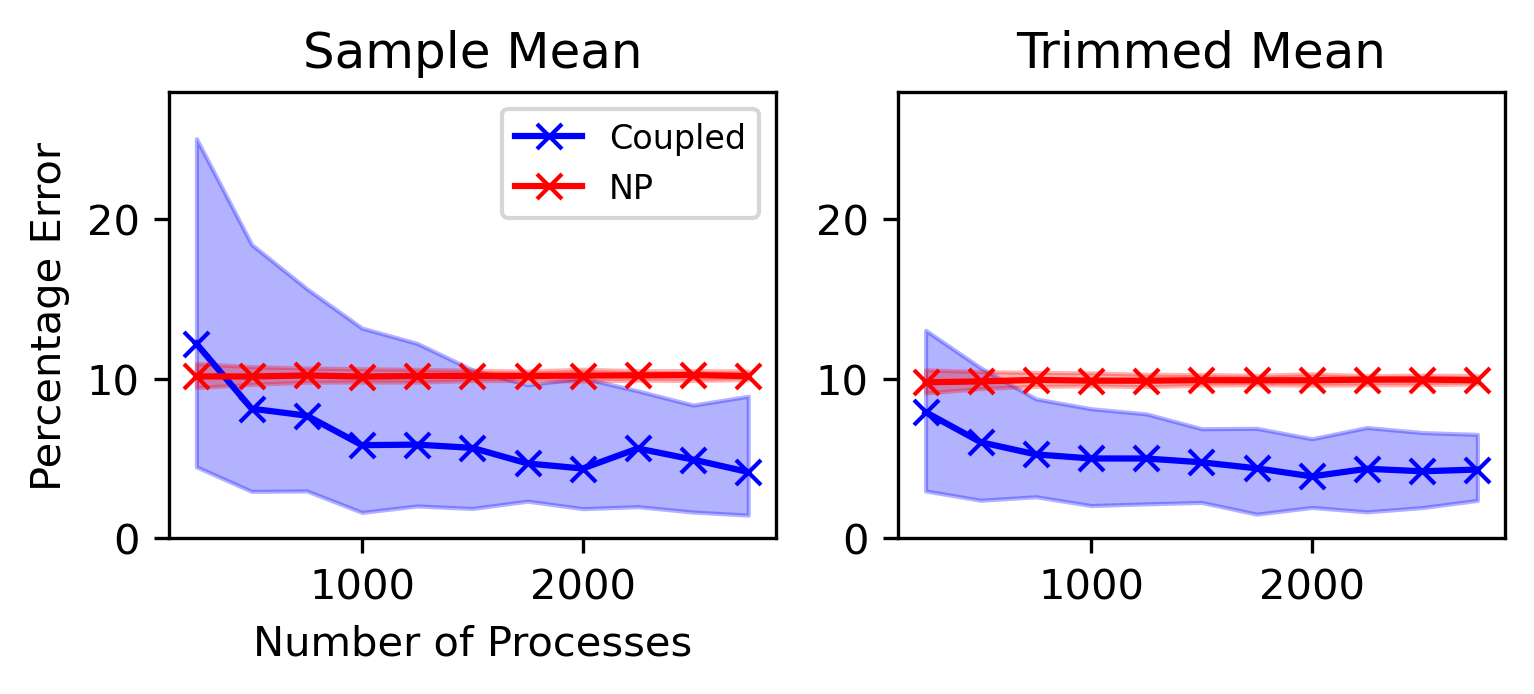}
		\caption{Losses}
		\label{sub-fig:gene-losses}
	\end{subfigure}\hspace{5mm}
	\begin{subfigure}[h]{0.48\linewidth}
		\includegraphics[scale=0.67]{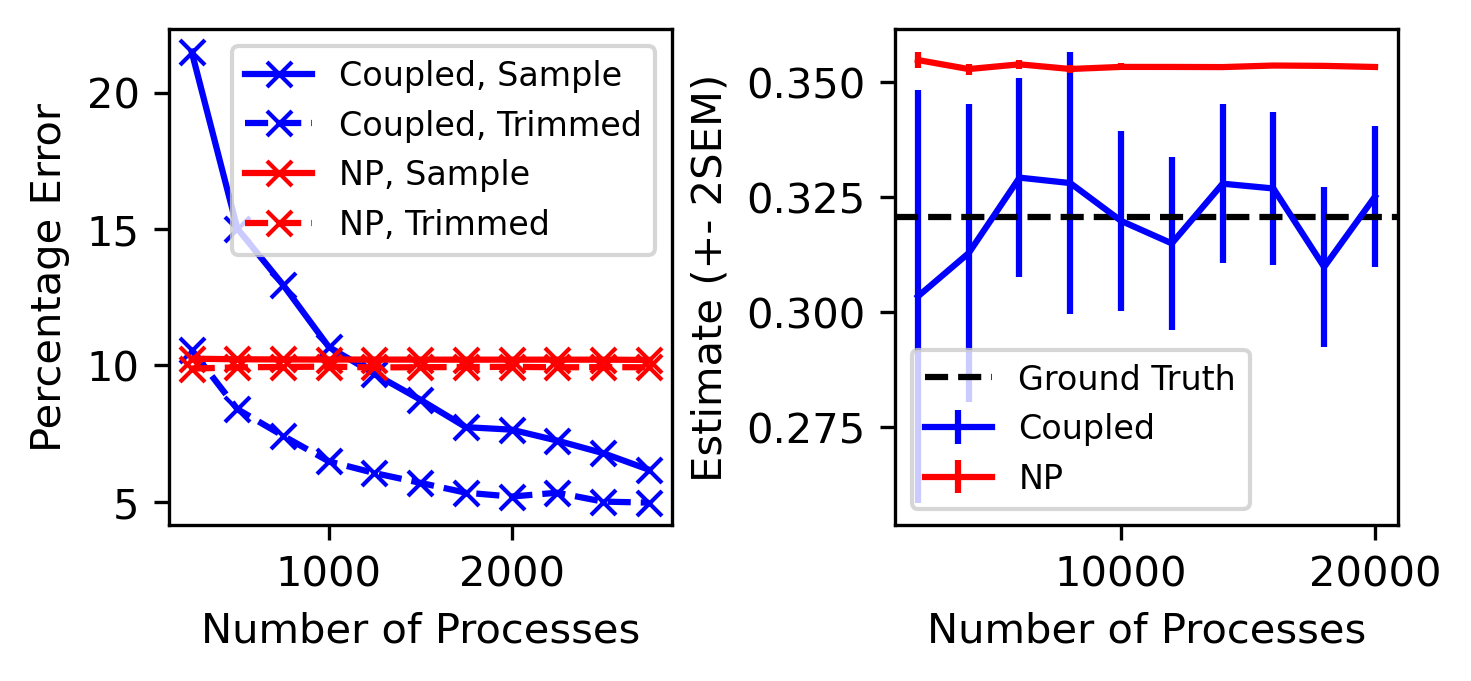}
		\caption{RMSE and intervals}
		\label{sub-fig:gene-estimation}
	\end{subfigure}
	\begin{subfigure}[b]{0.48\linewidth}
		\includegraphics[scale=0.67]{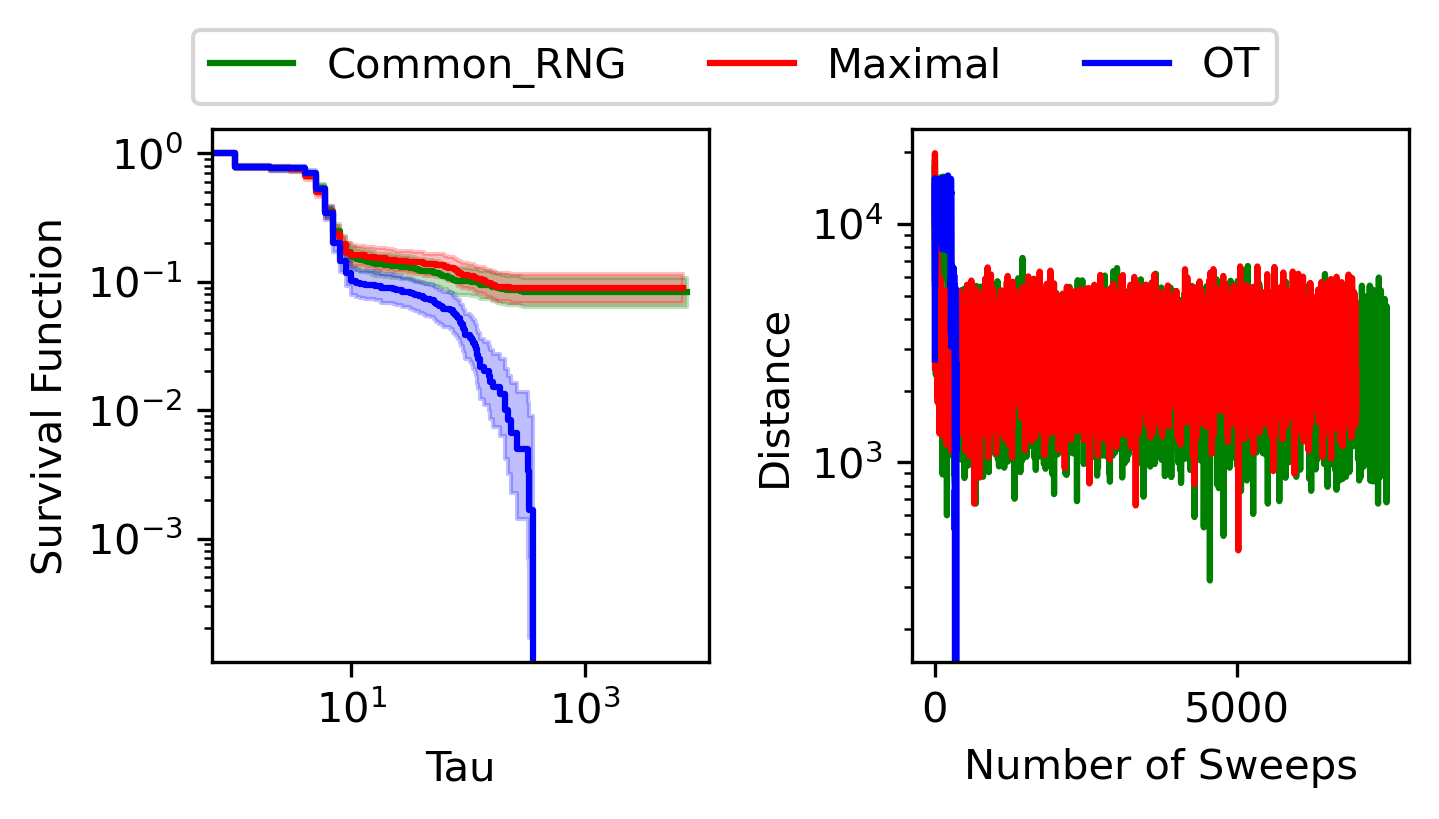}
		\caption{Coupling choice}
		\label{sub-fig:gene-coupling}
	\end{subfigure}
	\caption{Results on \textsc{gene}.}
	\label{fig:gene-all}
\end{figure}

\subsection{\textsc{synthetic}}
\Cref{fig:synthetic-all} shows results for $\LCP$ estimation on \textsc{synthetic} -- see \Cref{fig:co-cluster} for results on co-clustering.

\begin{figure}[h]
	\centering
	\begin{subfigure}[h]{0.48\linewidth}
		\includegraphics[scale=0.67]{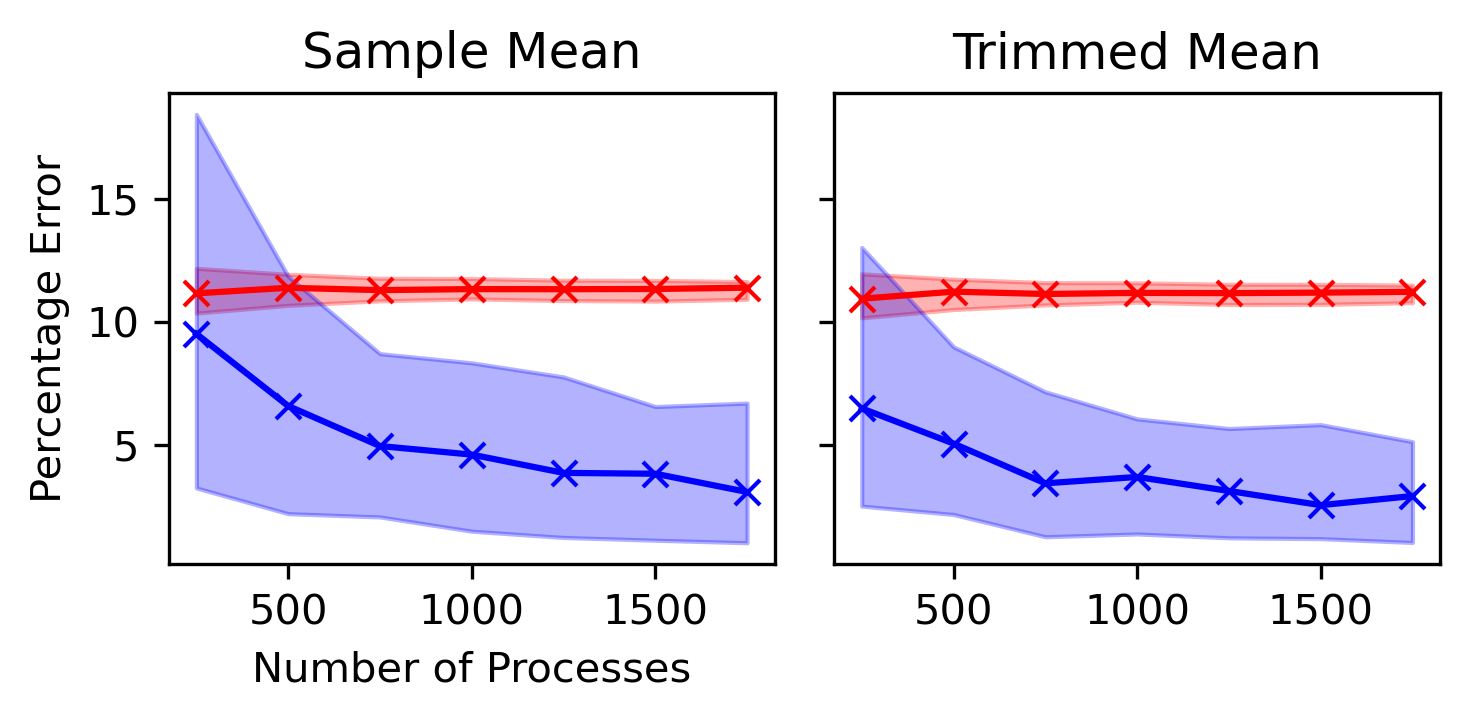}
		\caption{Losses}
		\label{sub-fig:synthetic-losses}
	\end{subfigure}
	\begin{subfigure}[h]{0.48\linewidth}
		\includegraphics[scale=0.67]{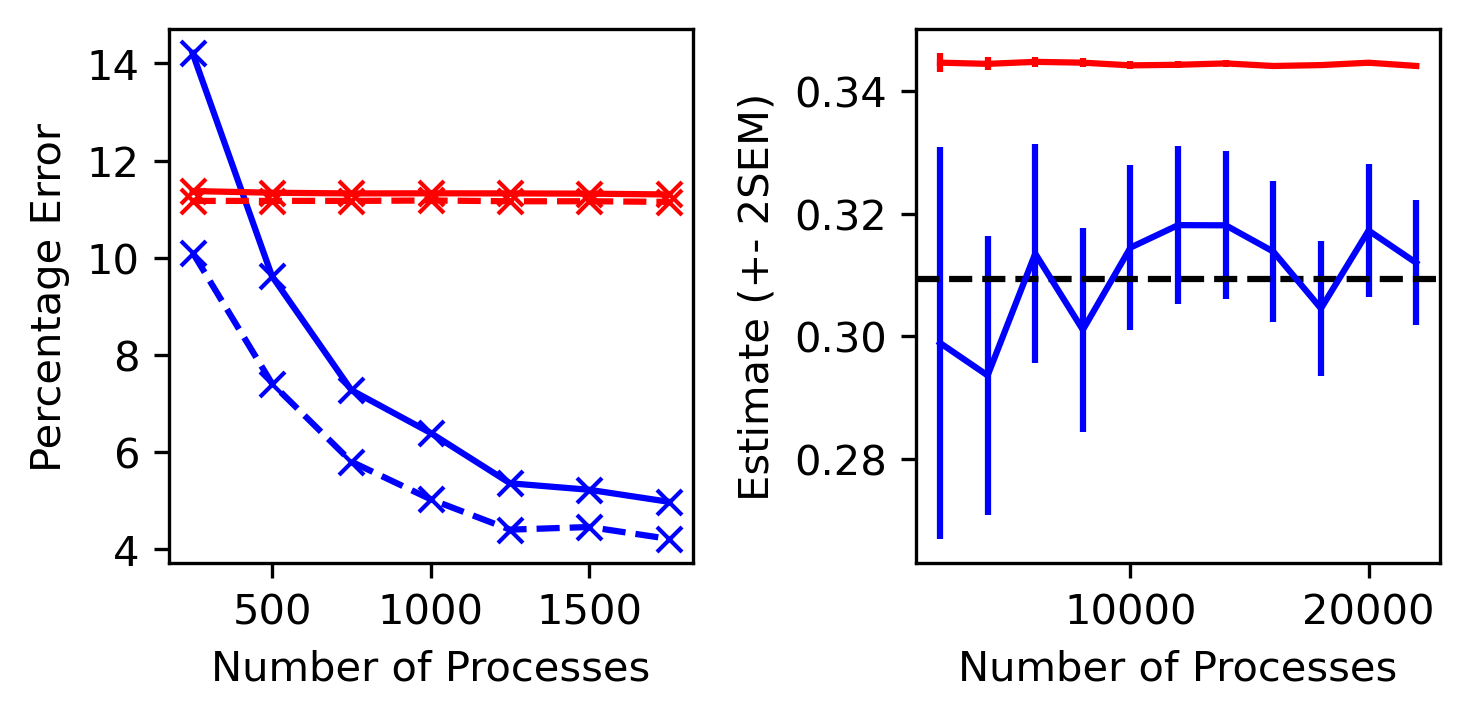}
		\caption{RMSE and intervals}
		\label{sub-fig:synthetic-estimation}
	\end{subfigure}
	\begin{subfigure}[b]{0.48\linewidth}
		\includegraphics[scale=0.67]{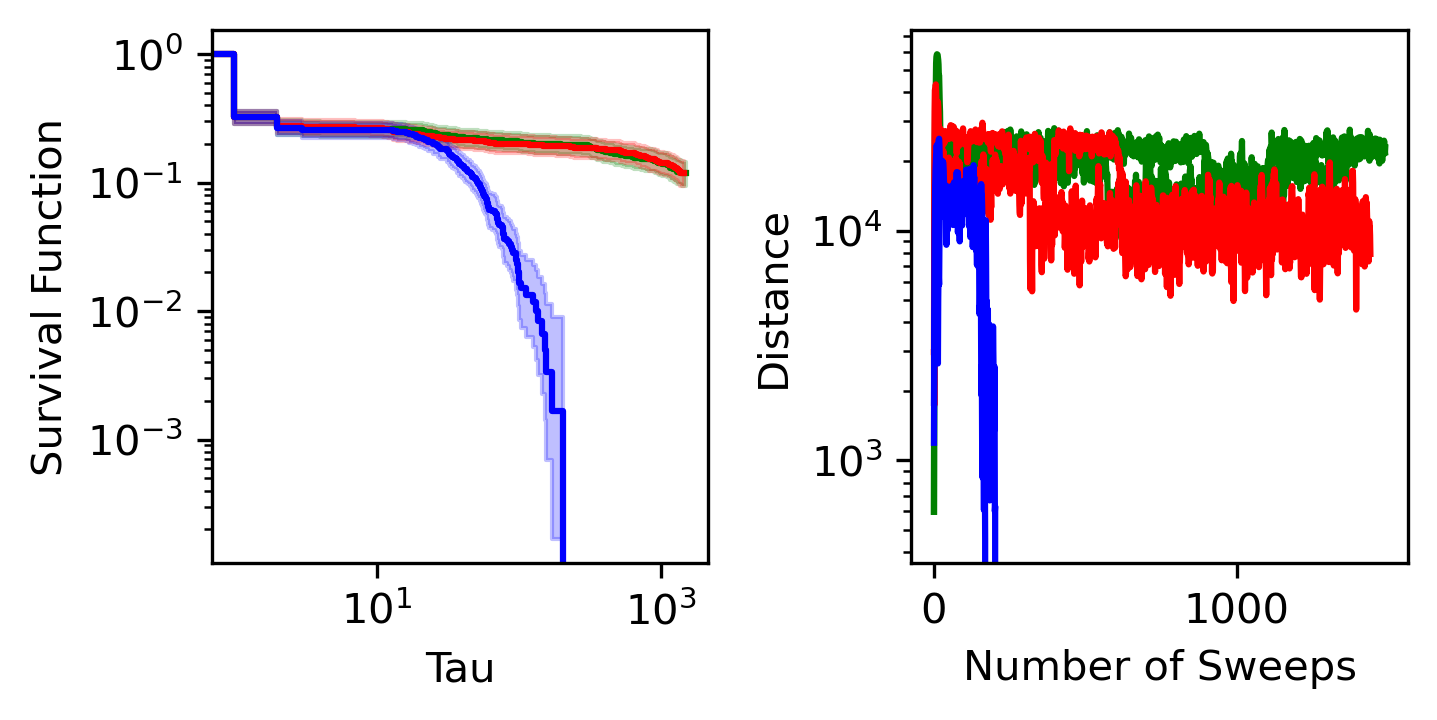}
		\caption{Coupling choice}
		\label{sub-fig:synthetic-coupling}
	\end{subfigure}
	\caption{Results on \textsc{synthetic}. Figure legends are the same as \Cref{fig:gene-all}. The results are consistent with \Cref{fig:composite}.}
	\label{fig:synthetic-all}
\end{figure}

\subsection{\textsc{seed}}
\Cref{fig:seed-all} shows results for $\LCP$ estimation on \textsc{seed} -- see \Cref{fig:co-cluster} for results on co-clustering.

\begin{figure}[h]
	\centering
	\begin{subfigure}[h]{0.48\linewidth}
		\includegraphics[scale=0.67]{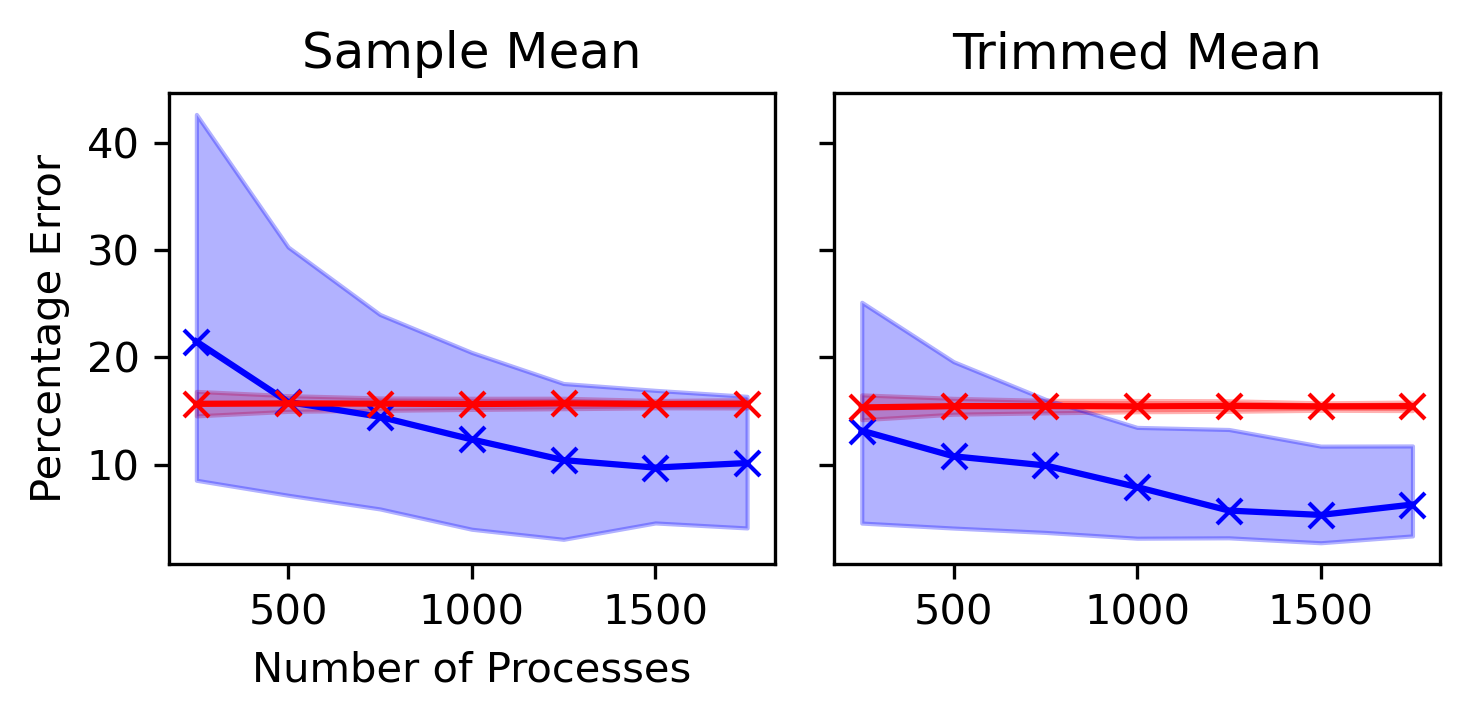}
		\caption{Losses}
		\label{sub-fig:seed-losses}
	\end{subfigure}
	\begin{subfigure}[h]{0.48\linewidth}
		\includegraphics[scale=0.67]{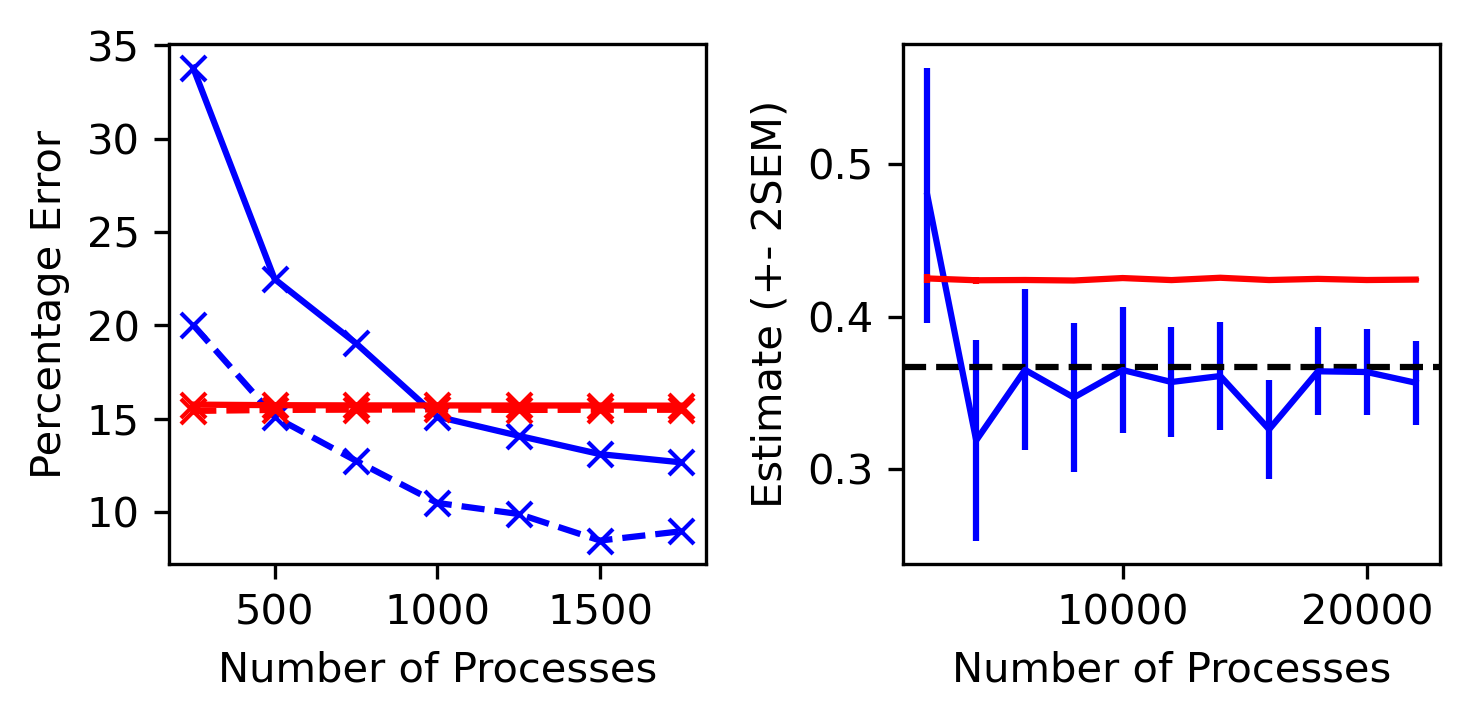}
		\caption{RMSE and intervals}
		\label{sub-fig:seed-estimation}
	\end{subfigure}
	\begin{subfigure}[b]{0.48\linewidth}
		\includegraphics[scale=0.67]{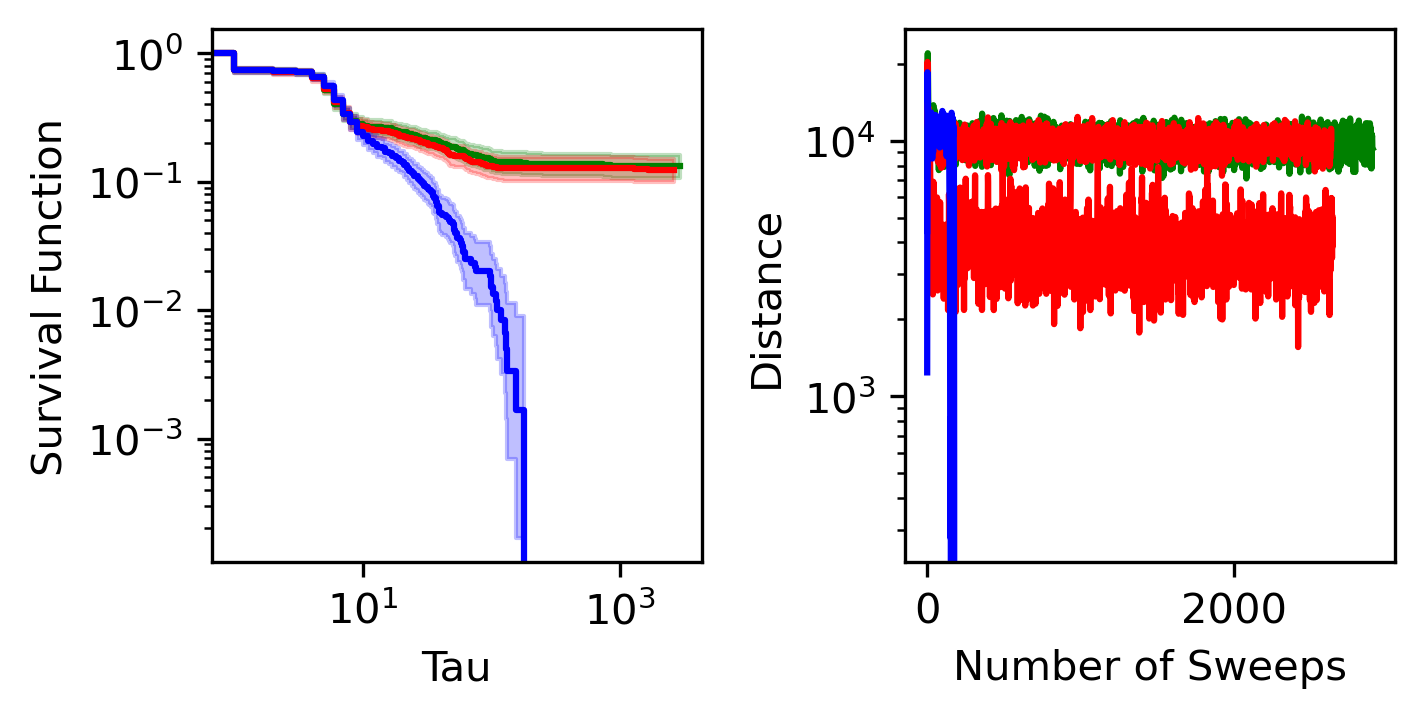}
		\caption{Coupling choice}
		\label{sub-fig:seed-coupling}
	\end{subfigure}
	\caption{Results on \textsc{seed}. Figure legends are the same as \Cref{fig:gene-all}. The results are consistent with \Cref{fig:composite}.}
	\label{fig:seed-all}
\end{figure}

\subsection{\textsc{abalone}}
\Cref{fig:abalone-all} shows results for $\LCP$ estimation on \textsc{abalone}.
In \Cref{sub-fig:abalone-losses} and \Cref{sub-fig:abalone-estimation}, we do not report results for the trimmed estimator with the default trimming amount ($0.01$ i.e.\ $1\%$). 
This trimming amount is too large for the application, and in \Cref{fig:trimming}, we show that trimming the most extreme $0.1\%$ yields much better estimation. 

\begin{figure}[h]
	\centering
	\begin{subfigure}[h]{0.3\linewidth}
		\includegraphics[scale=0.7]{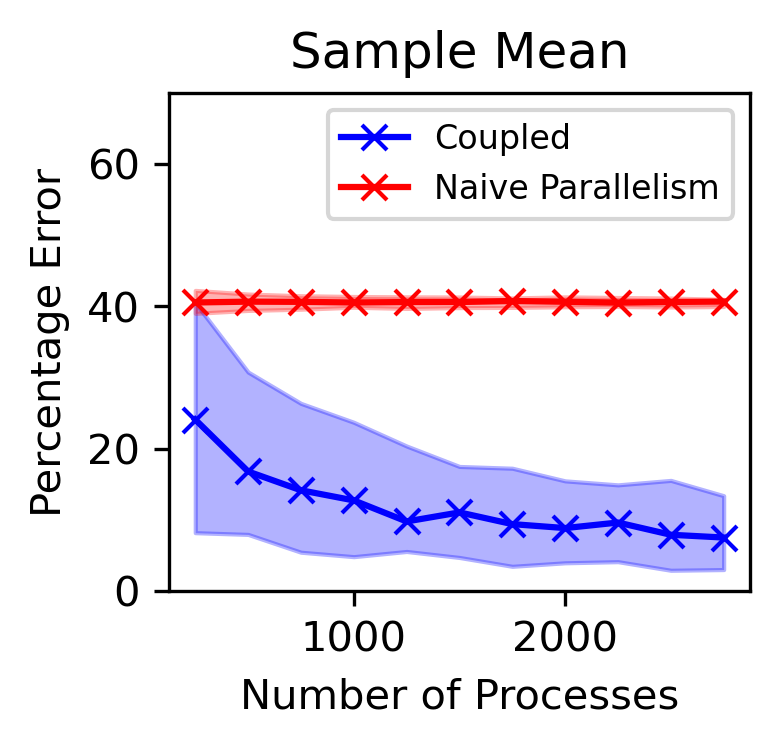}
		\caption{Losses}
		\label{sub-fig:abalone-losses}
	\end{subfigure}
	\begin{subfigure}[h]{0.48\linewidth}
		\includegraphics[scale=0.75]{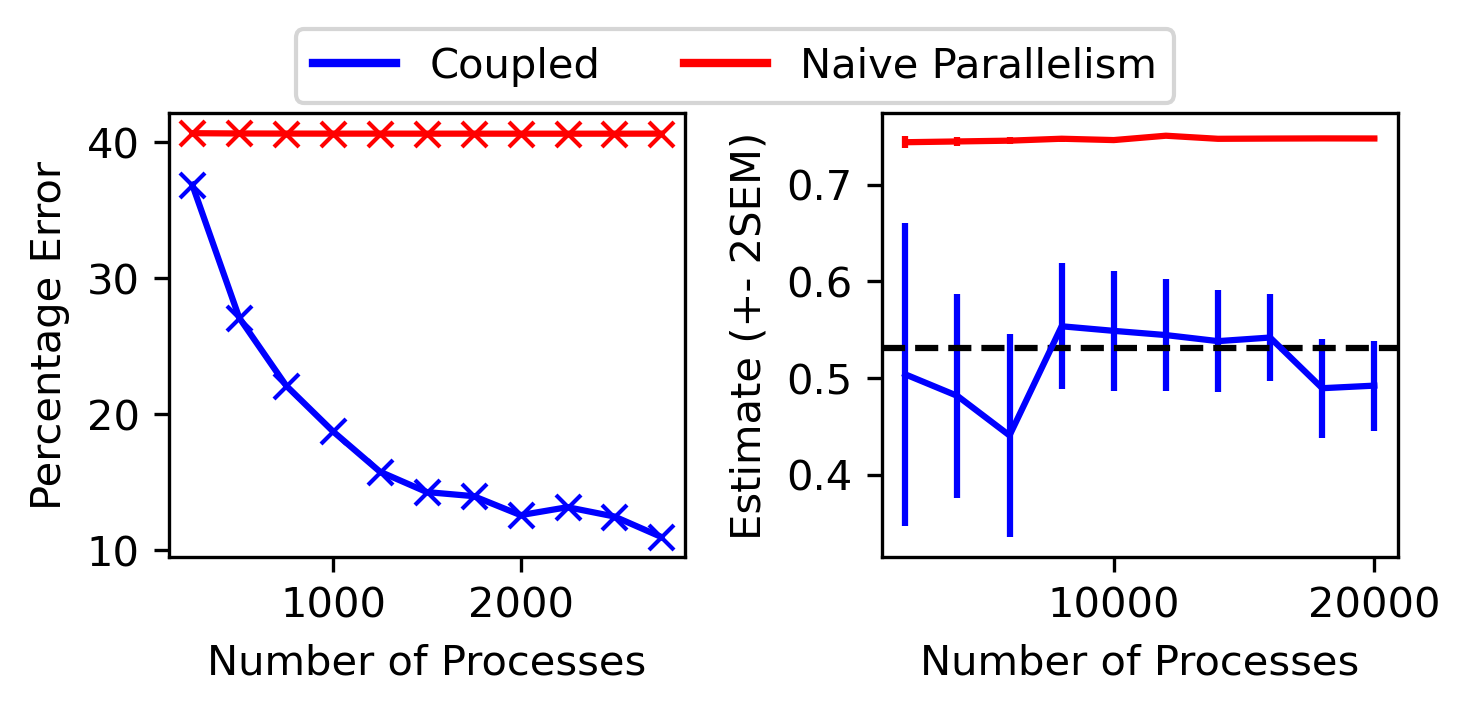}
		\caption{RMSE and intervals}
		\label{sub-fig:abalone-estimation}
	\end{subfigure}
	\begin{subfigure}[b]{0.48\linewidth}
		\includegraphics[scale=0.7]{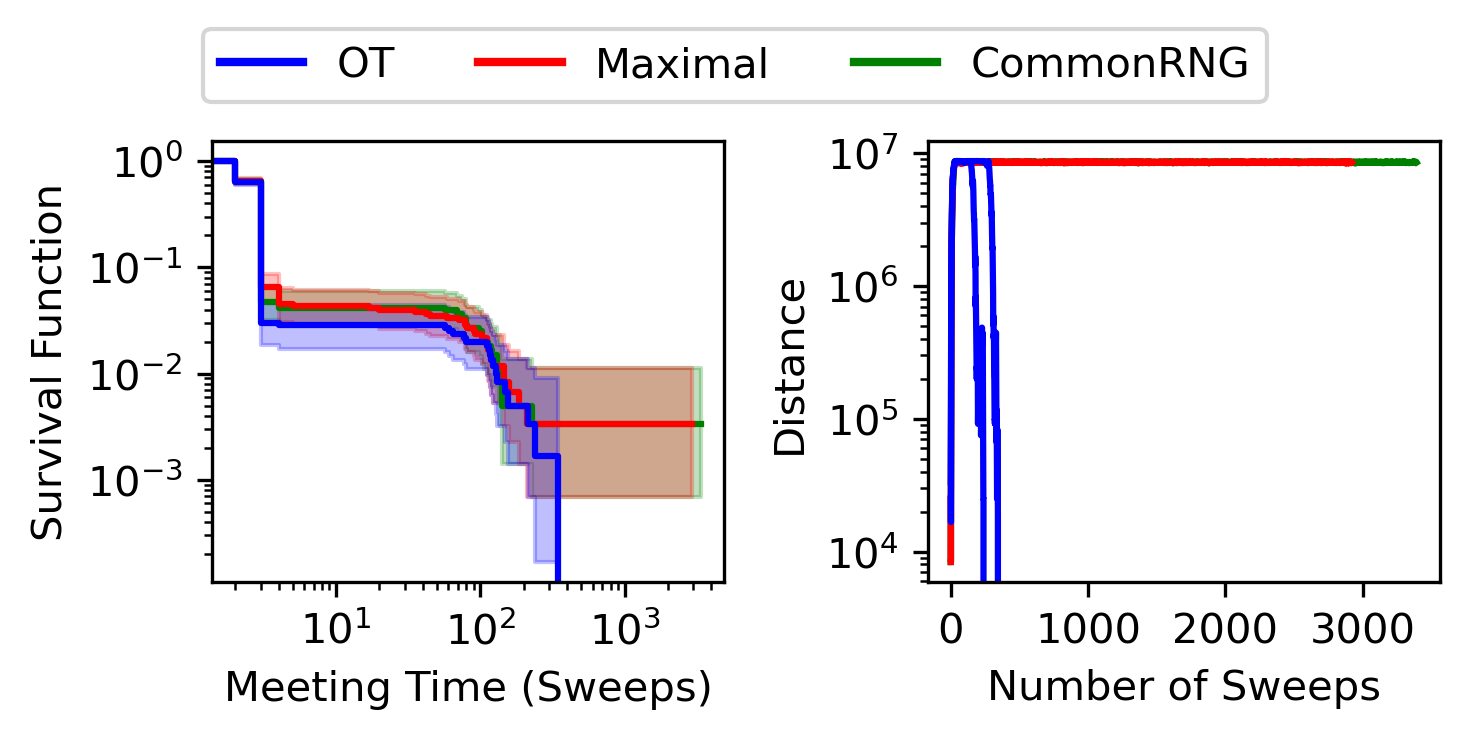}
		\caption{Coupling choice}
		\label{sub-fig:abalone-coupling}
	\end{subfigure}
	\caption{Results on \textsc{abalone}. Similar to \Cref{fig:composite}, coupled chains perform better than naive parallelism with more processes, and our coupling yields smaller meeting times than label-based couplings. See \Cref{fig:trimming} for the performance of trimmed estimators.}
	\label{fig:abalone-all}
\end{figure}

In \Cref{fig:trimming}, the first panel (from the left) plots the errors incurred using the trimmed mean with the default $\alpha = 1\%$.
Trimming of coupled chains is still better than naive parallelism, but worse than sample mean of coupled chains.
In the second panel, we use $\alpha = 0.1\%$, and the trimming of coupled chains performs much better. 
In the third panel, we fix the number of processes to be {2000} and quantify the RMSE as a function of the trimming amount (expressed in percentages). 
We see a gradual decrease in the RMSE as the trimming amount is reduced, indicating that this is a situation in which smaller trimming amounts is prefered.
\begin{figure}[h]
	\centering
	\includegraphics[scale=0.8]{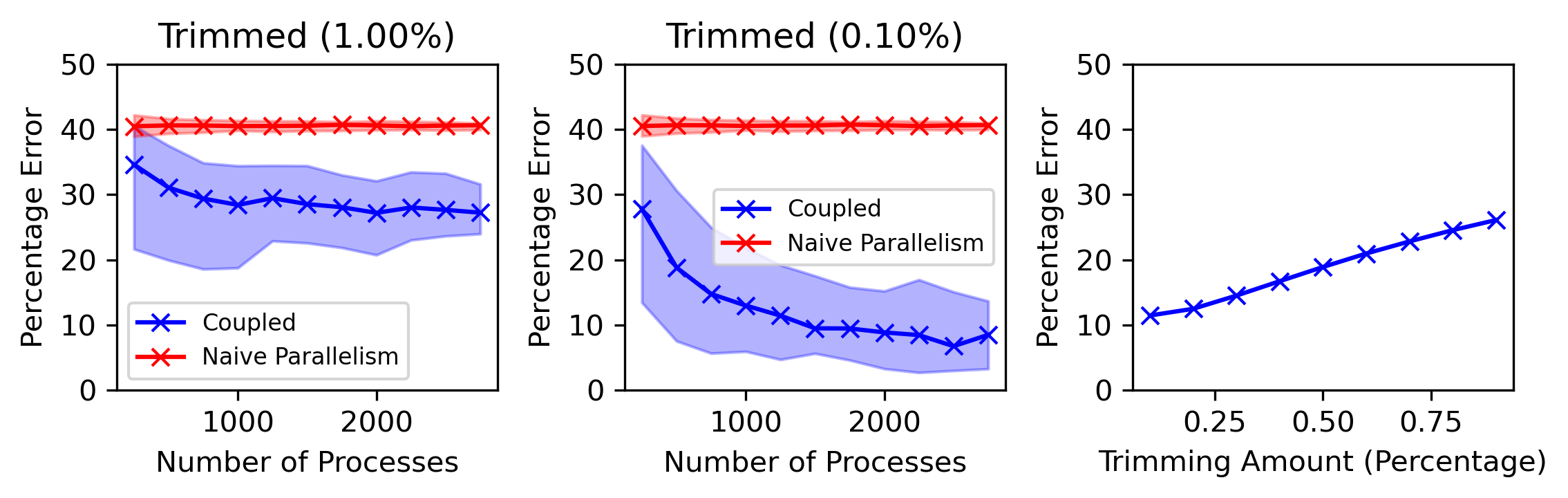}
	\caption{Effect of trimming amount on \textsc{abalone}.}
	\label{fig:trimming}
\end{figure}

\subsection{\textsc{k-regular}}
\Cref{fig:kRegular-all} shows results for $\CC{2}{4}$ estimation on \textsc{k-regular}.

\begin{figure}[h]
	\centering
	\begin{subfigure}[h]{0.48\linewidth}
		\includegraphics[scale=0.67]{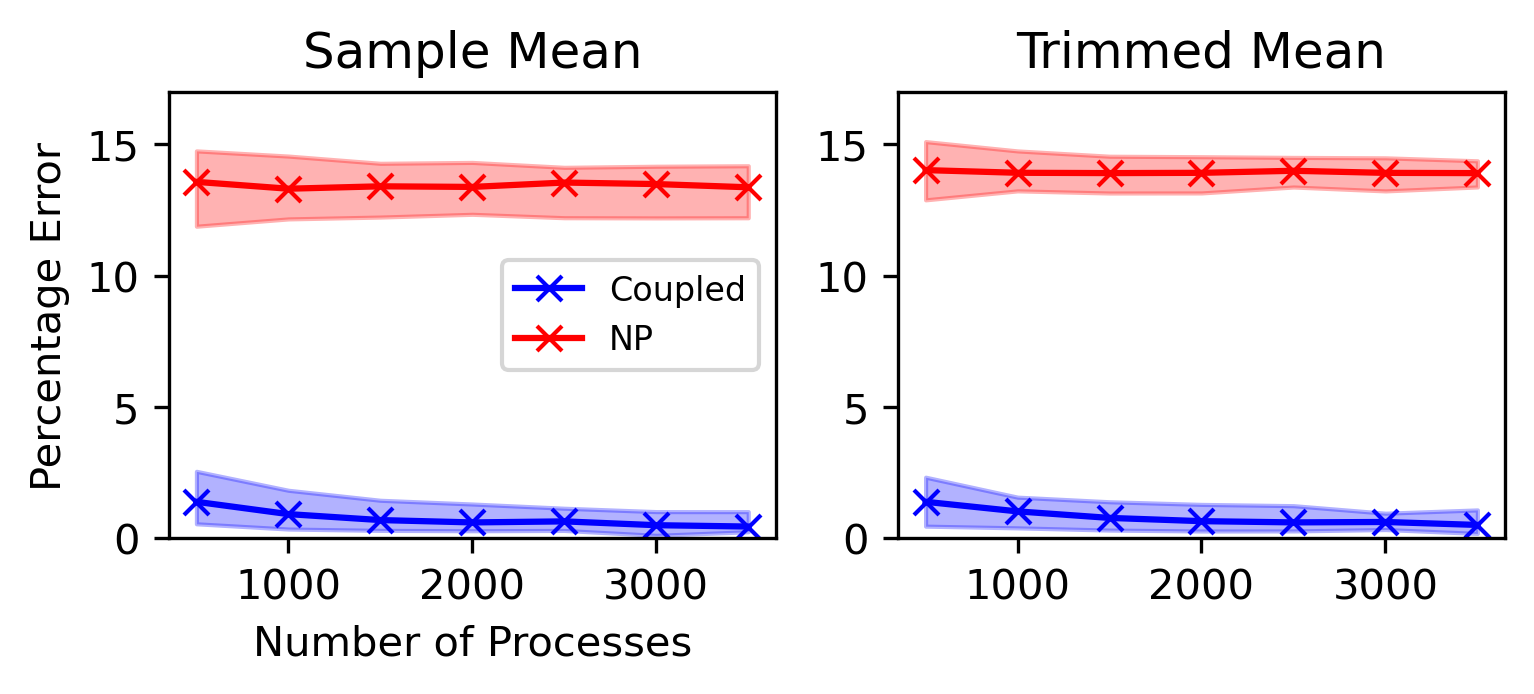}
		\caption{Losses}
		\label{sub-fig:kRegular-losses}
	\end{subfigure}\hspace{5mm}
	\begin{subfigure}[h]{0.48\linewidth}
		\includegraphics[scale=0.67]{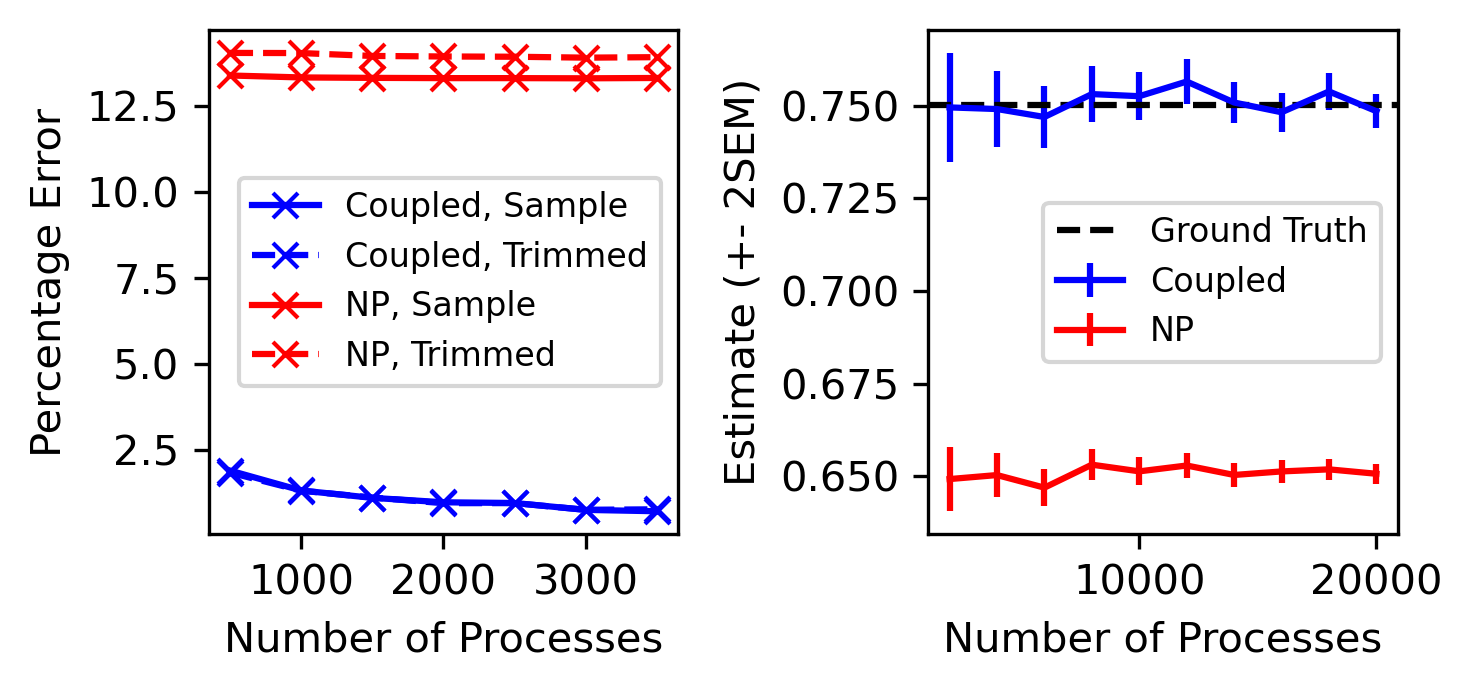}
		\caption{RMSE and intervals}
		\label{sub-fig:kRegular-estimation}
	\end{subfigure}
	\begin{subfigure}[b]{0.48\linewidth}
		\includegraphics[scale=0.67]{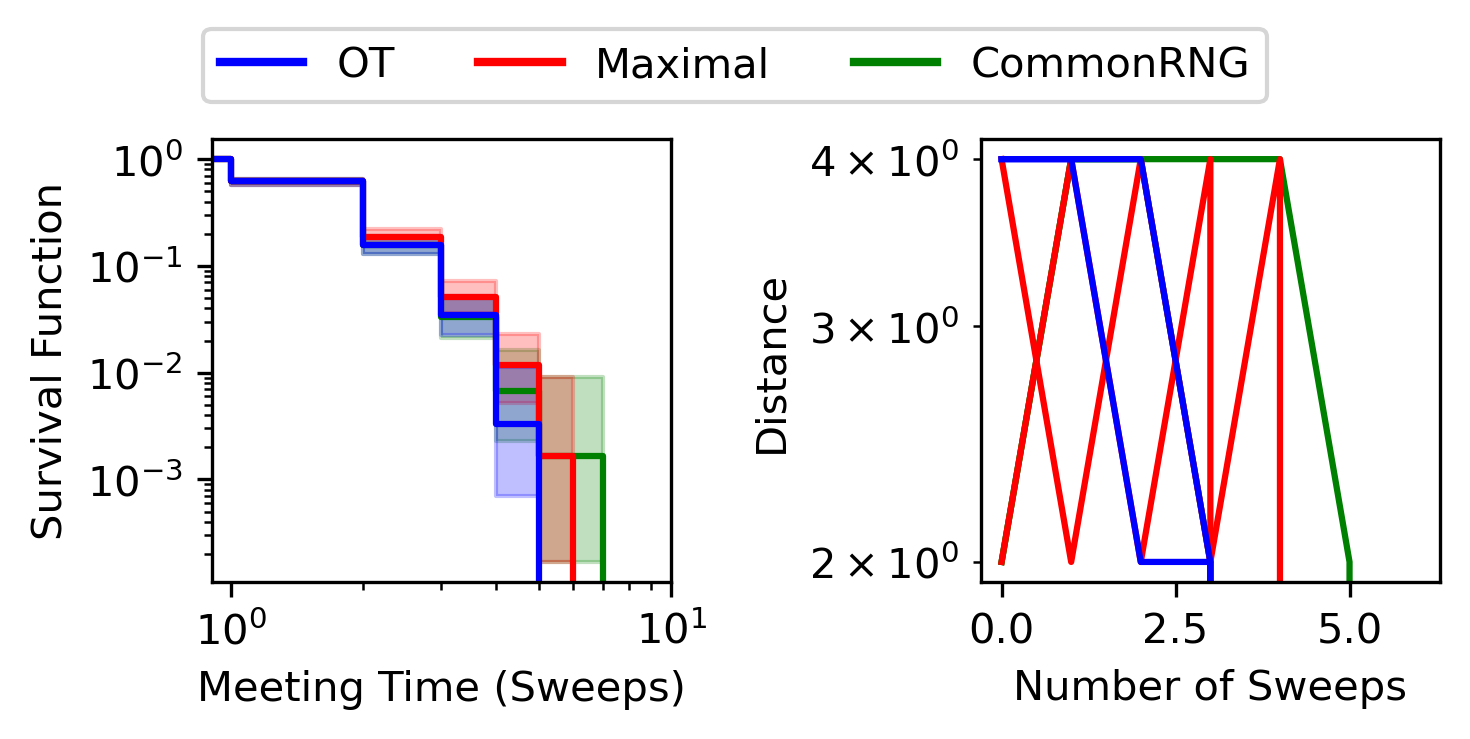}
		\caption{Coupling choice}
		\label{sub-fig:kRegular-coupling}
	\end{subfigure}
	\caption{Results on \textsc{k-regular}. Figure legends are the same as \Cref{fig:gene-all}.}
	\label{fig:kRegular-all}
\end{figure}

\section{METRIC IMPACT} \label{apd:metric}
\subsection{Definition Of Variation Of Information Metric}
Variation of information, or VI, is defined in \citet[Equation 16]{meilua2007comparing}.
We replicate the definition in what follows.
Let $\pi$ and $\nu$ be two partitions of $[N]$.
Denote the clusters in $\pi$ by $\{\rIdx{A}{1},\rIdx{A}{2},\ldots,\rIdx{A}{K}\}$
and the clusters in $\nu$ by $\{\rIdx{B}{1},\rIdx{B}{2},\ldots,\rIdx{B}{K'}\}$.
For each $k \in [K]$ and $k' \in K'$, define the number $P(k,k')$ to be
\begin{equation*}
	P(k,k') \defined \frac{|\rIdx{A}{k} \cap \rIdx{B}{k'} |}{N}.
\end{equation*}
$|\rIdx{A}{k} \cap \rIdx{B}{k'}|$ is the size of the overlap between $\rIdx{A}{k}$ and $\rIdx{B}{k'}$. 
Because of the normalization by $N$, the $P(k,k')$'s are non-negative and sum to $1$, hence can be interpreted as probability masses. 
Summing across all $k$ (or $k'$) has a marginalization effect, and we define
\begin{equation*}
	P(k) \defined \sum_{k'=1}^{K'} P(k,k'). 
\end{equation*}
Similarly we define $P'(k') \defined \sum_{k=1}^{K} P(k,k')$.
The VI metric is then
\begin{equation} \label{eq:VI}
	d_{I}(\pi, \nu) = \sum_{k=1}^{K} \sum_{k'=1}^{K'} P(k,k') \log \frac{P(k,k')}{P(k)P(k')}.
\end{equation}
In terms of theoretical properties, \citet[Property 1]{meilua2007comparing} shows that $d_I$ is a metric for the space of partitions.

\subsection{Impact Of Metric On Meeting Time} 
In \Cref{sub-fig:gene-metric,sub-fig:seed-metric,sub-fig:synthetic-metric,sub-fig:kRegular-metric}, we examine the effect of metric on the meeting time for coupled chains. 
In place of the Hamming metric in \Cref{eqn:distance_between_partitions}, we can use the variation of information (VI) metric from \Cref{eq:VI}
in defining the OT problem (\Cref{eq:OT-coupling}). 
Based on the survival functions, the meeting time under VI metric is similar to meeting time under the default Hamming metric: in all cases, the survival functions lie mostly right on top of each other. 
Time is measured in number of sweeps taken, rather than processor time, because under Hamming metric we have a fast implementation (\Cref{s-sec:runtime}) while we are not aware of fast implementations for the VI metric. 
Hence, our recommended metric choice is Hamming (\Cref{eqn:distance_between_partitions}).

\begin{figure}[h]
	\centering
	\begin{subfigure}{0.48\linewidth}
		\includegraphics[scale=0.64]{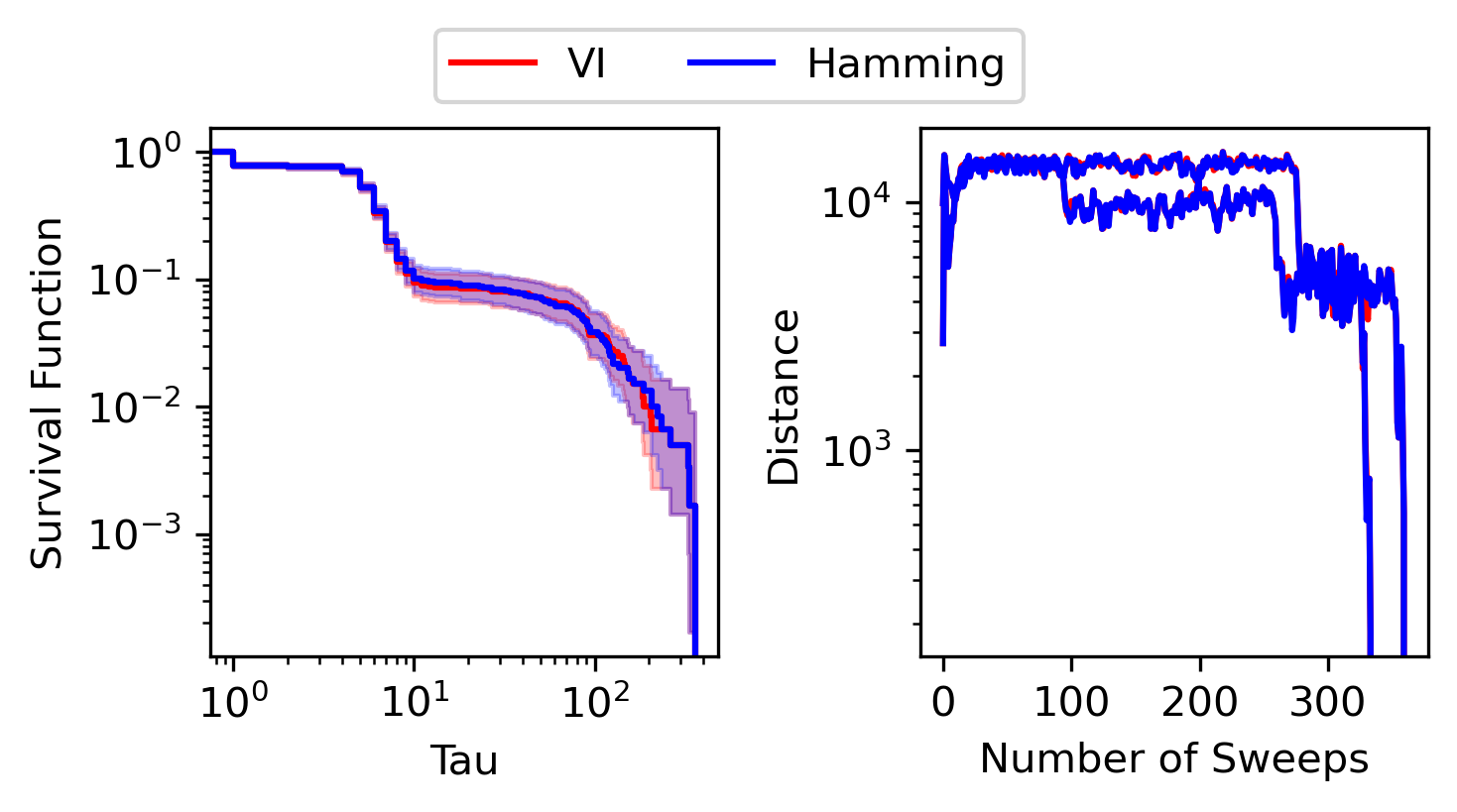}
		\caption{\textsc{gene}}
		\label{sub-fig:gene-metric}
	\end{subfigure}
	\begin{subfigure}{0.48\linewidth}
		\includegraphics[scale=0.67]{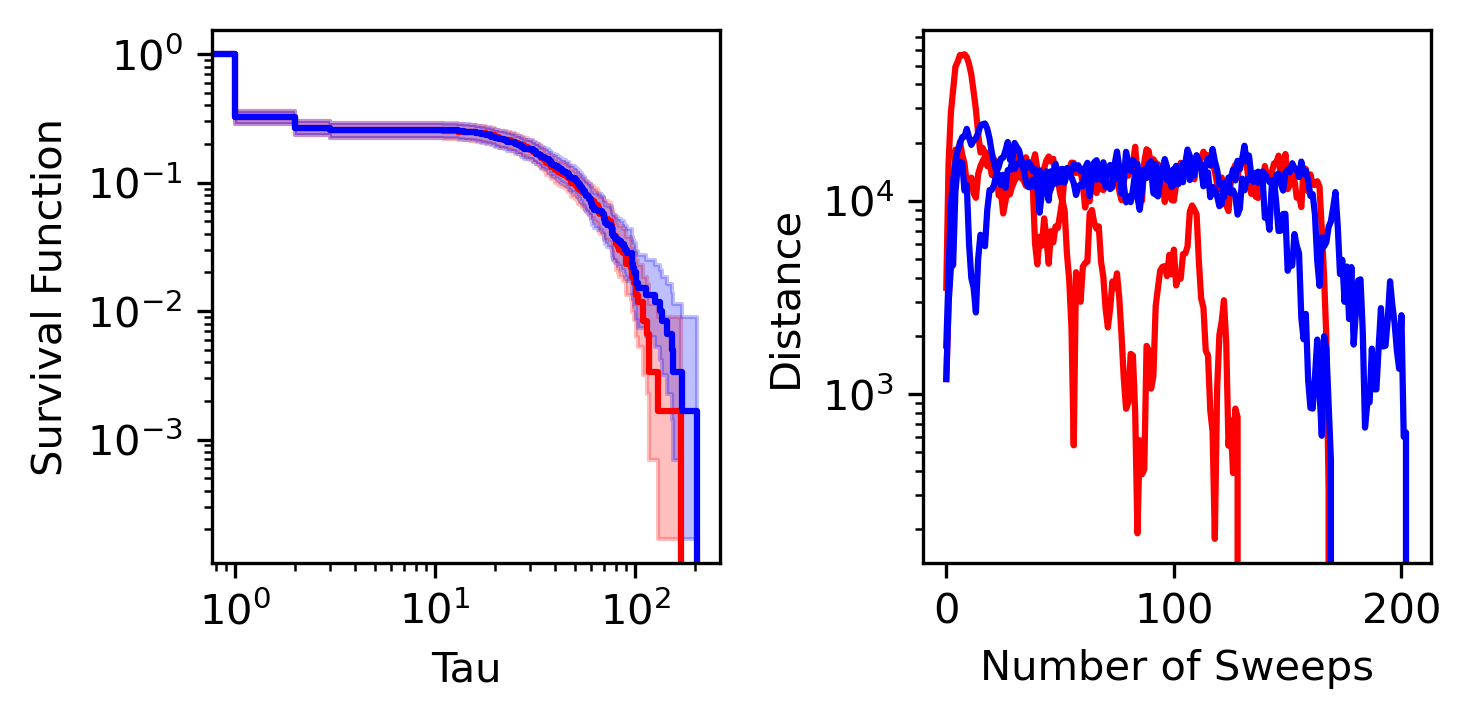}
		\caption{\textsc{synthetic}}
		\label{sub-fig:synthetic-metric}
	\end{subfigure}\\
	\begin{subfigure}{0.48\linewidth}
		\includegraphics[scale=0.67]{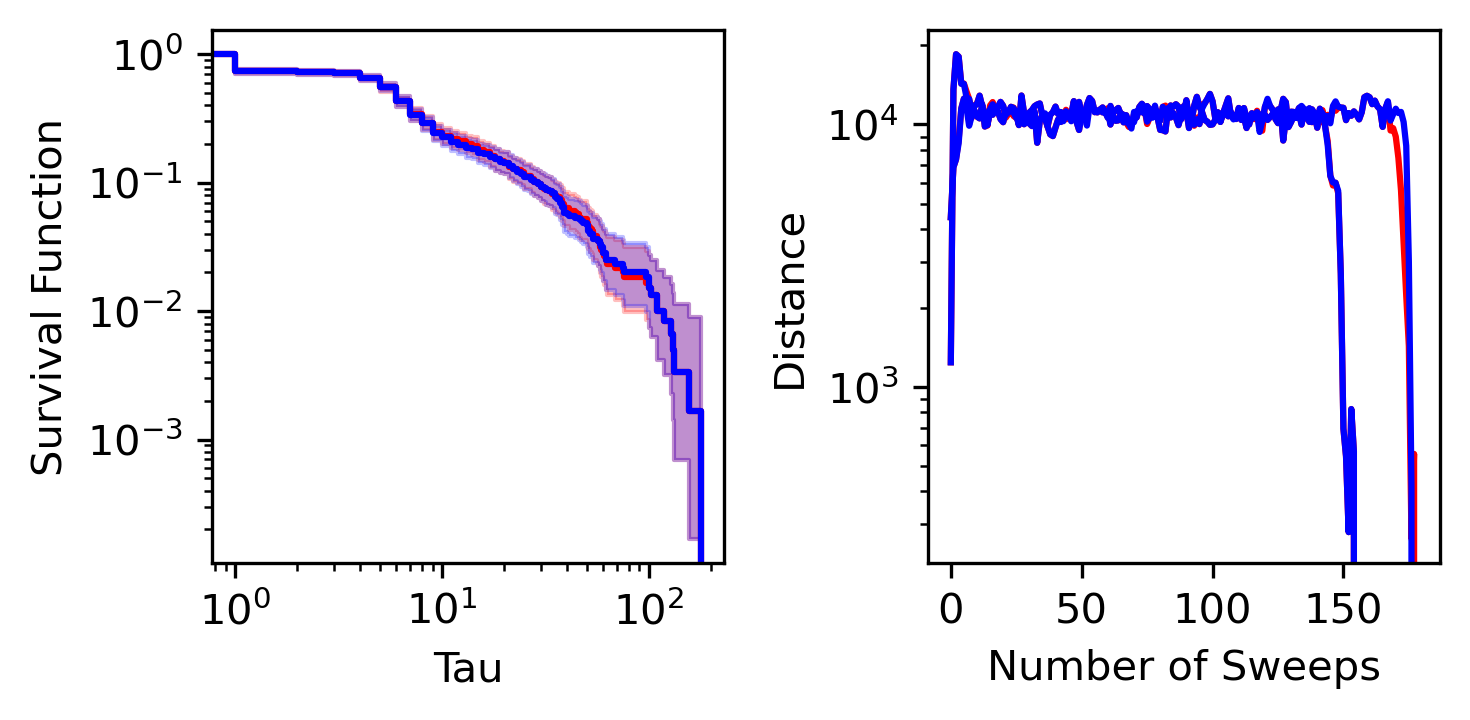}
		\caption{\textsc{seed}}
		\label{sub-fig:seed-metric}
	\end{subfigure}
	\begin{subfigure}{0.48\linewidth}
		\includegraphics[scale=0.67]{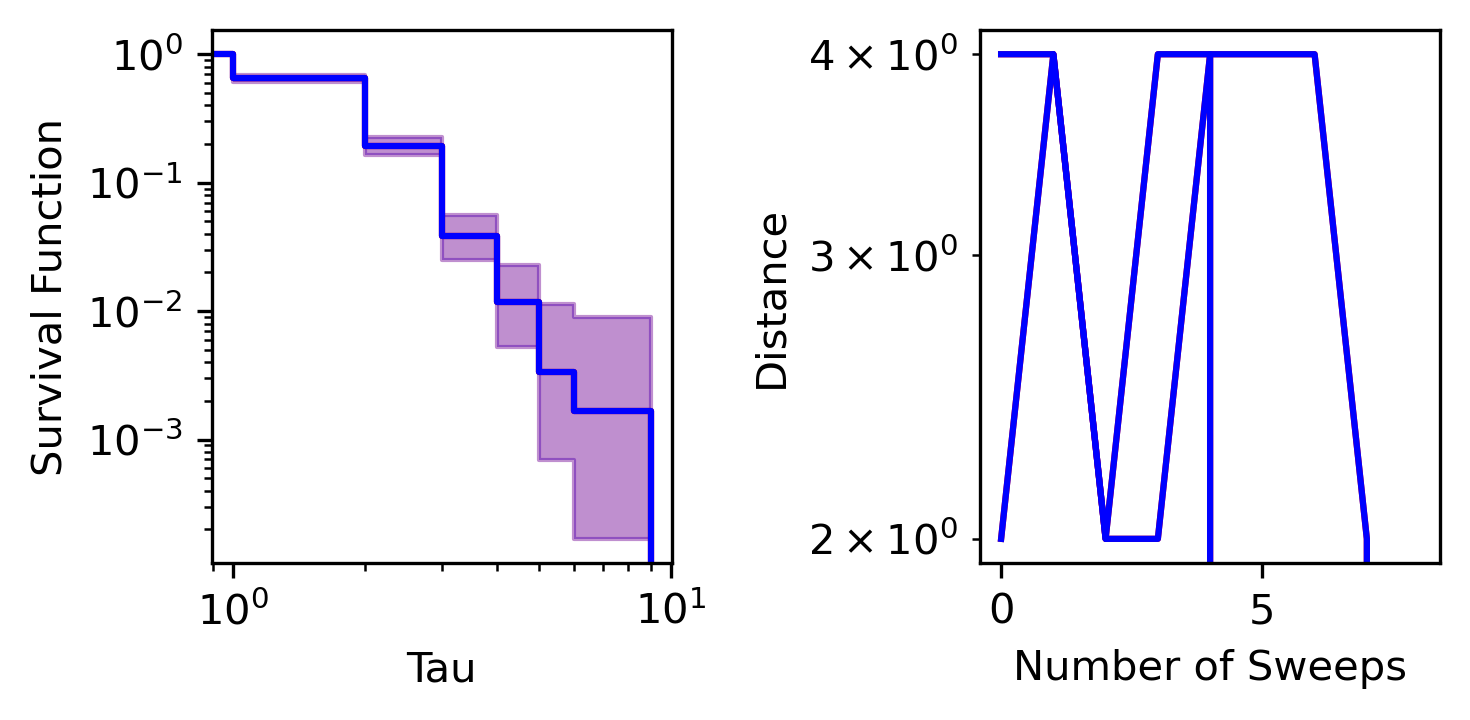}
		\caption{\textsc{k-regular}}
		\label{sub-fig:kRegular-metric}
	\end{subfigure}
	\caption{Hamming and VI metric induce similar meeting time}
	\label{fig:metrics}
\end{figure}

\section{EXTENSION TO SPLIT-MERGE SAMPLER} \label{apd:extensions}
$\SpMe(i,j,X)$ is the ``Restricted Gibbs Sampling Split--Merge Procedure'' from \citet{jain2004split}, where our implementation proposes $1$ split--merge move and uses $5$ intermediate Gibbs scan to compute the proposed split (or merge) states.

We refer to \Cref{apd:experimental_setup} for comprehensive experimental setup.
The $\LCP$ estimation results for \textsc{gene} are given in \Cref{fig:gene-splitMerge}. 
Instead of the one-component initialization, we use a k-means clustering with 5 components as initialization.
$\minIter$ is set to be 100, while $\burnin$ is 10.
Switching from pure Gibbs sampler to split-merge samplers can reduce the bias caused by a bad initialization.
But there is still bias that does not go away with replication, and the results are consistent with \Cref{fig:composite}.

\begin{figure}[h]
	\centering
	\begin{subfigure}[h]{0.48\linewidth}
		\includegraphics[scale=0.67]{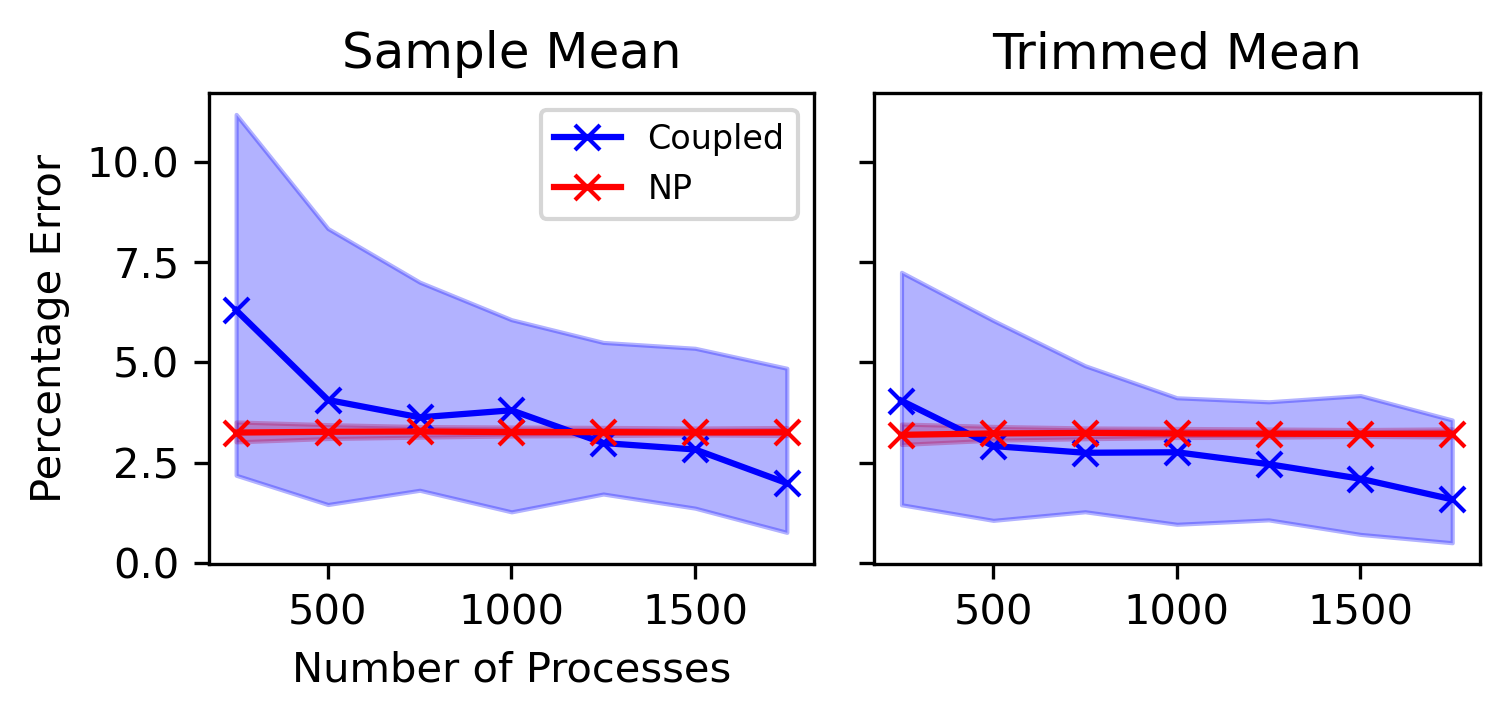}
		\caption{Losses}
	\end{subfigure}\hspace{5mm}
	\begin{subfigure}[h]{0.48\linewidth}
		\includegraphics[scale=0.67]{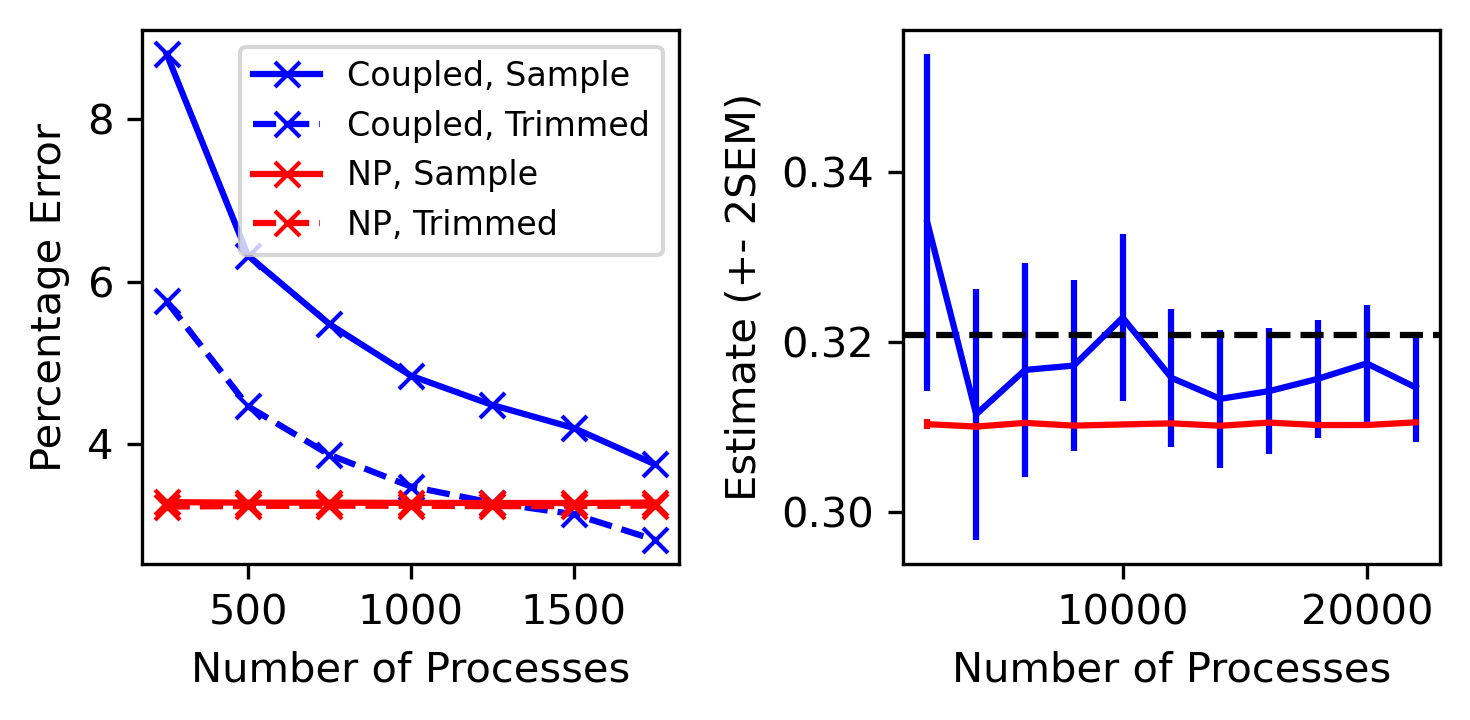}
		\caption{RMSE and intervals}
	\end{subfigure}
	\begin{subfigure}[b]{0.48\linewidth}
		\includegraphics[scale=0.67]{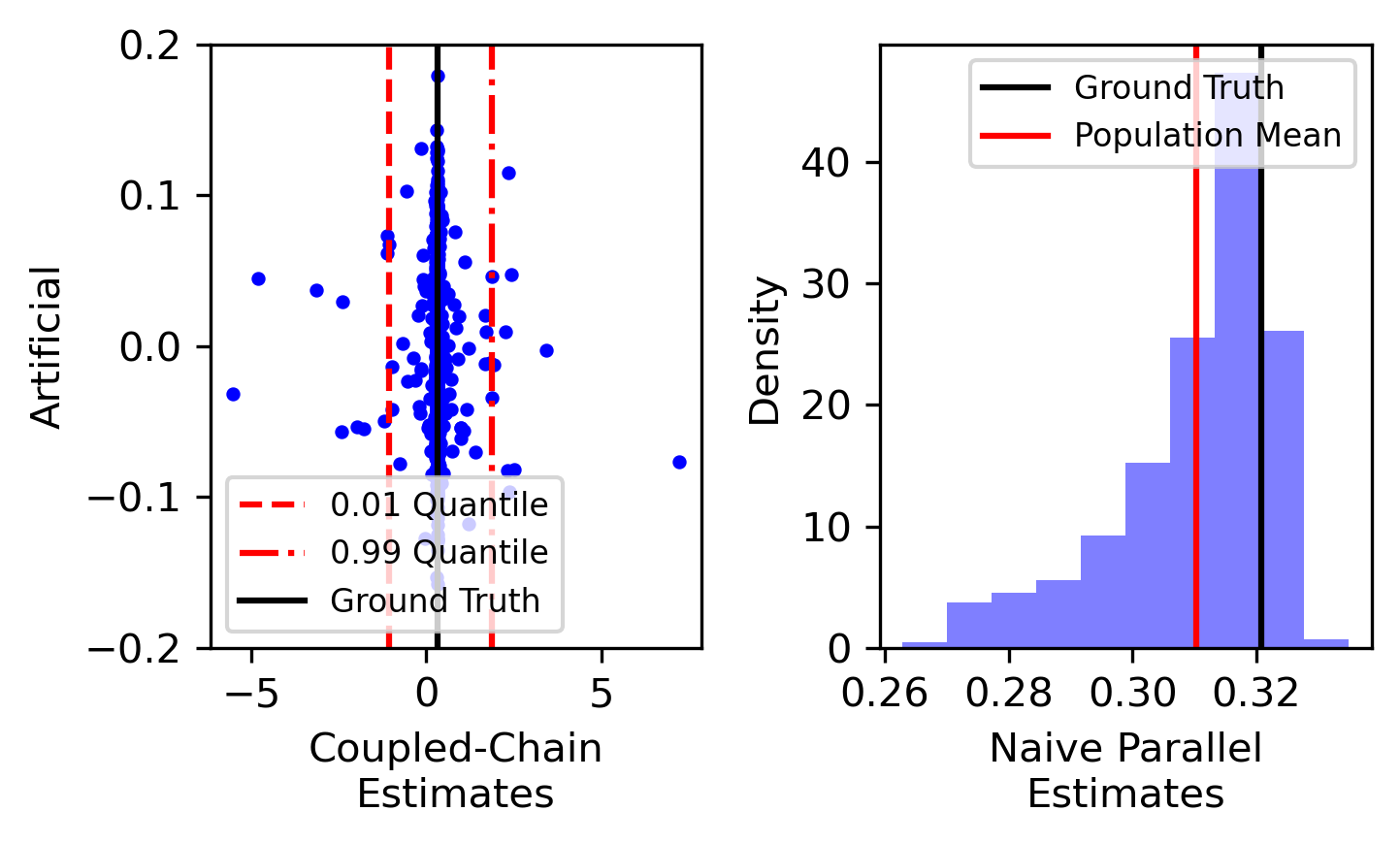}
		\caption{Estimates}
	\end{subfigure}
	\caption{Split-merge results on \textsc{gene}}
	\label{fig:gene-splitMerge}
\end{figure}

We also have split-merge results for estimation of $\CC{0}{1}$ on \textsc{synthetic} in \Cref{fig:synthetic-splitMerge}.
$\minIter$ is set to be 50, while $\burnin$ is 5.

\begin{figure}[h]
	\centering
	\begin{subfigure}[h]{0.48\linewidth}
		\includegraphics[scale=0.67]{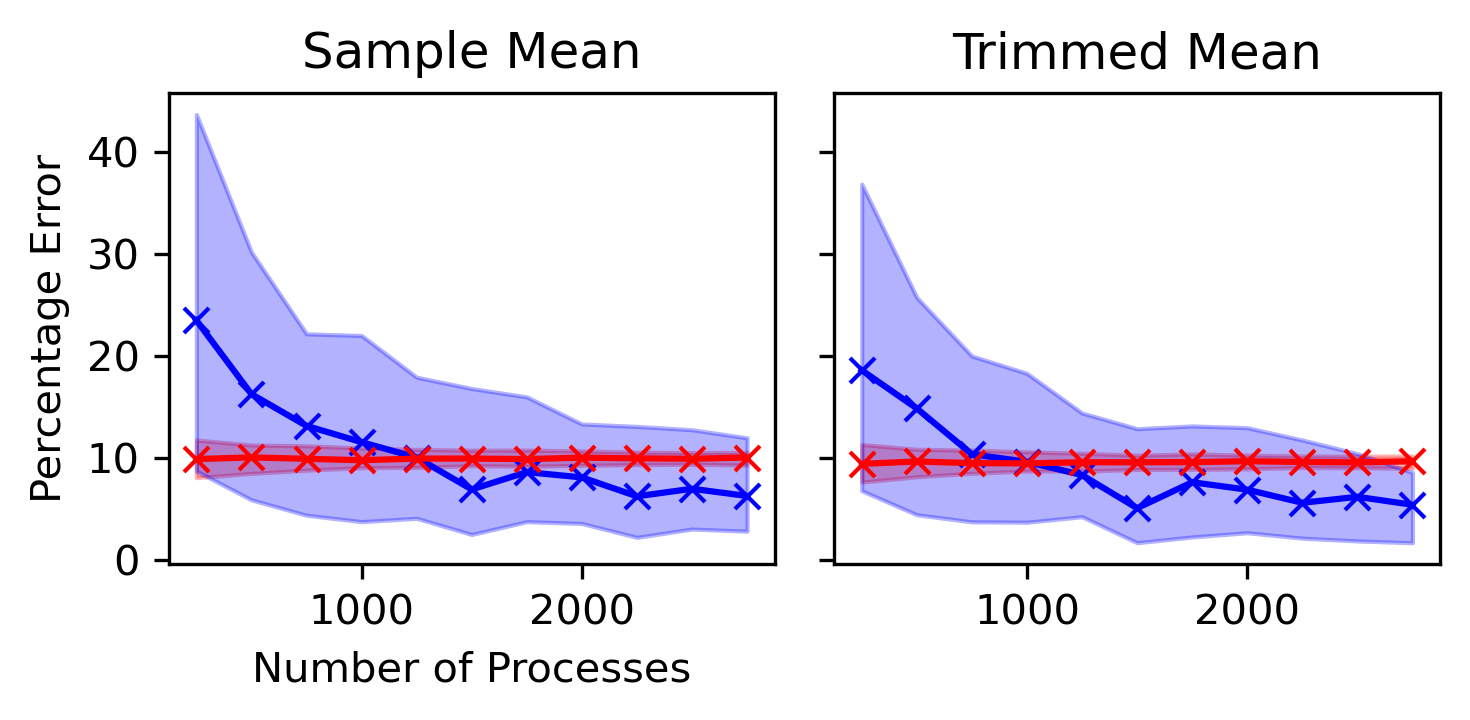}
		\caption{Losses}
	\end{subfigure}
	\begin{subfigure}[h]{0.48\linewidth}
		\includegraphics[scale=0.67]{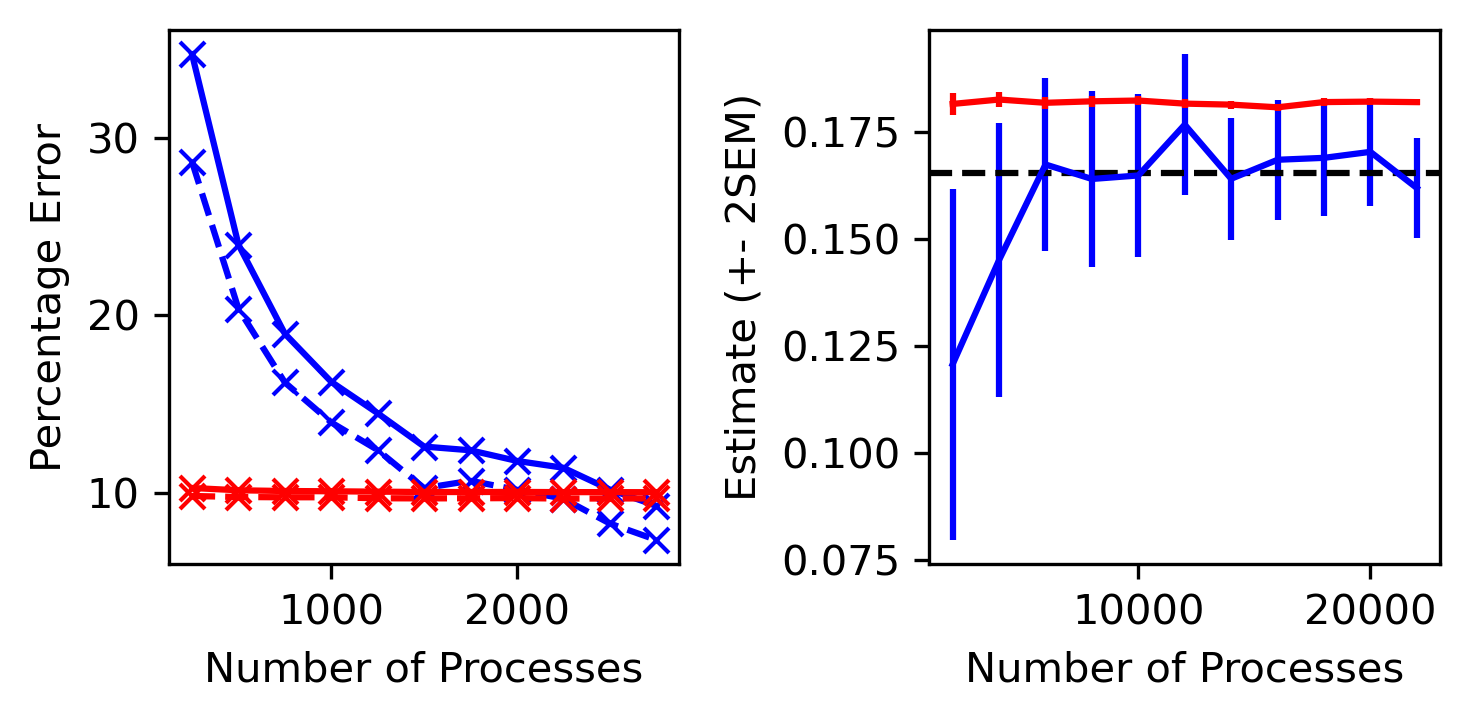}
		\caption{RMSE and intervals}
	\end{subfigure}
	\begin{subfigure}[b]{0.48\linewidth}
		\includegraphics[scale=0.67]{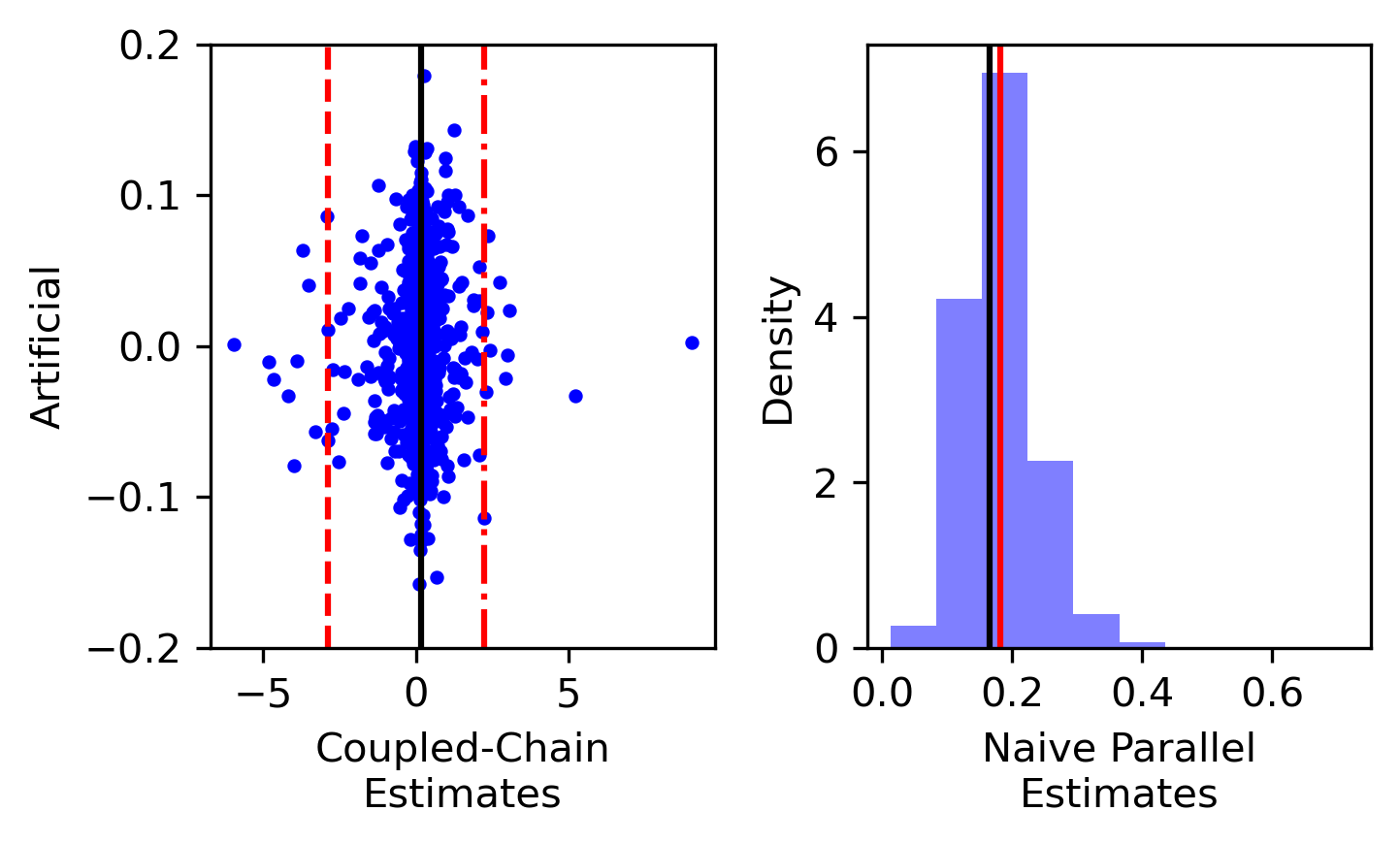}
		\caption{Estimates}
	\end{subfigure}
	\caption{Split-merge results on \textsc{synthetic}}
	\label{fig:synthetic-splitMerge}
\end{figure}

\section{MORE RMSE PLOTS} \label{apd:more-rmse}
\subsection{Different Functions Of Interest}
\Cref{fig:co-cluster} displays co-clustering results for clustering data sets.
The results are consistent with those for $\LCP$ estimation.
Co-clustering appears to be a more challenging estimation problem than $\LCP$, indicated by the higher percentage errors for the same $\minIter$.

\begin{figure}[h]
	\centering
	\begin{subfigure}{0.48\linewidth}
		\includegraphics[scale=0.67]{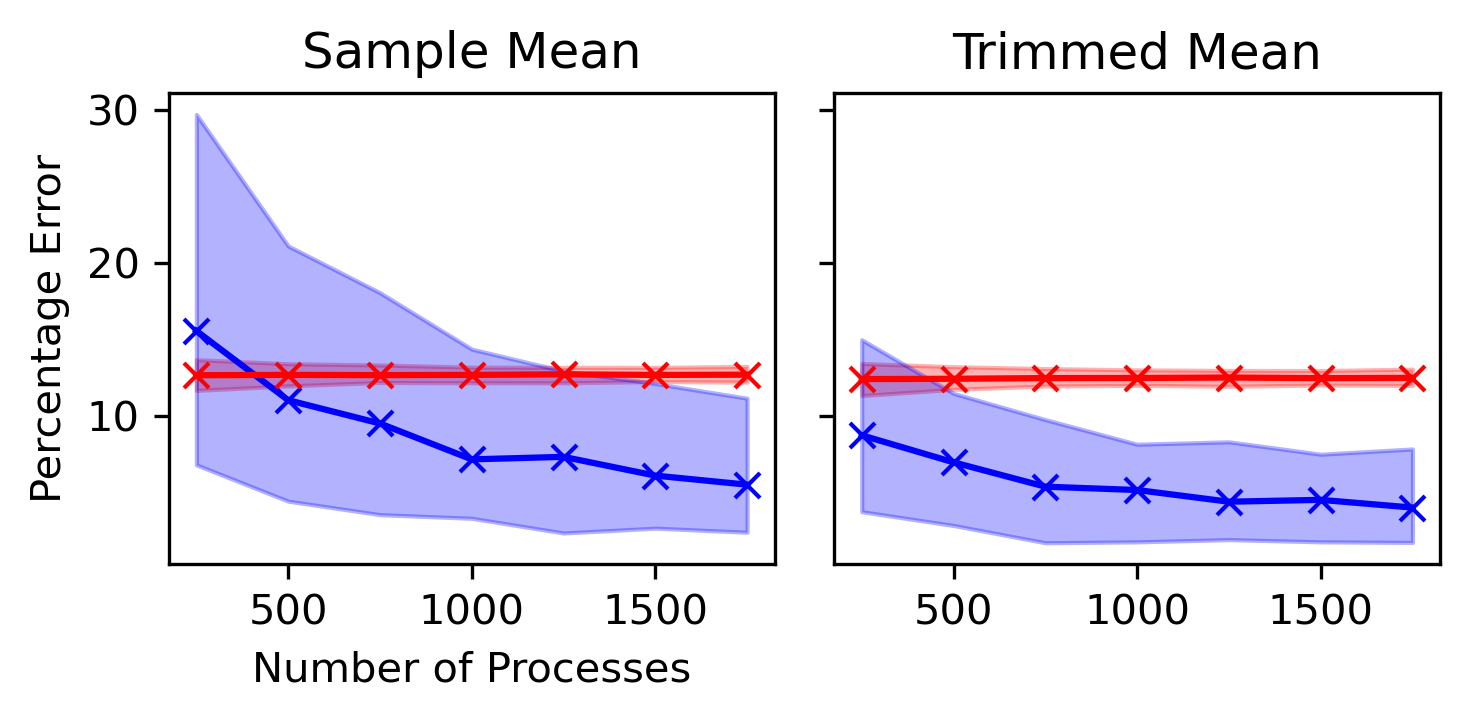}
		\caption{Losses, $\CC{0}{21}$ estimation on \textsc{gene}}
	\end{subfigure}
	\begin{subfigure}{0.48\linewidth}
		\includegraphics[scale=0.67]{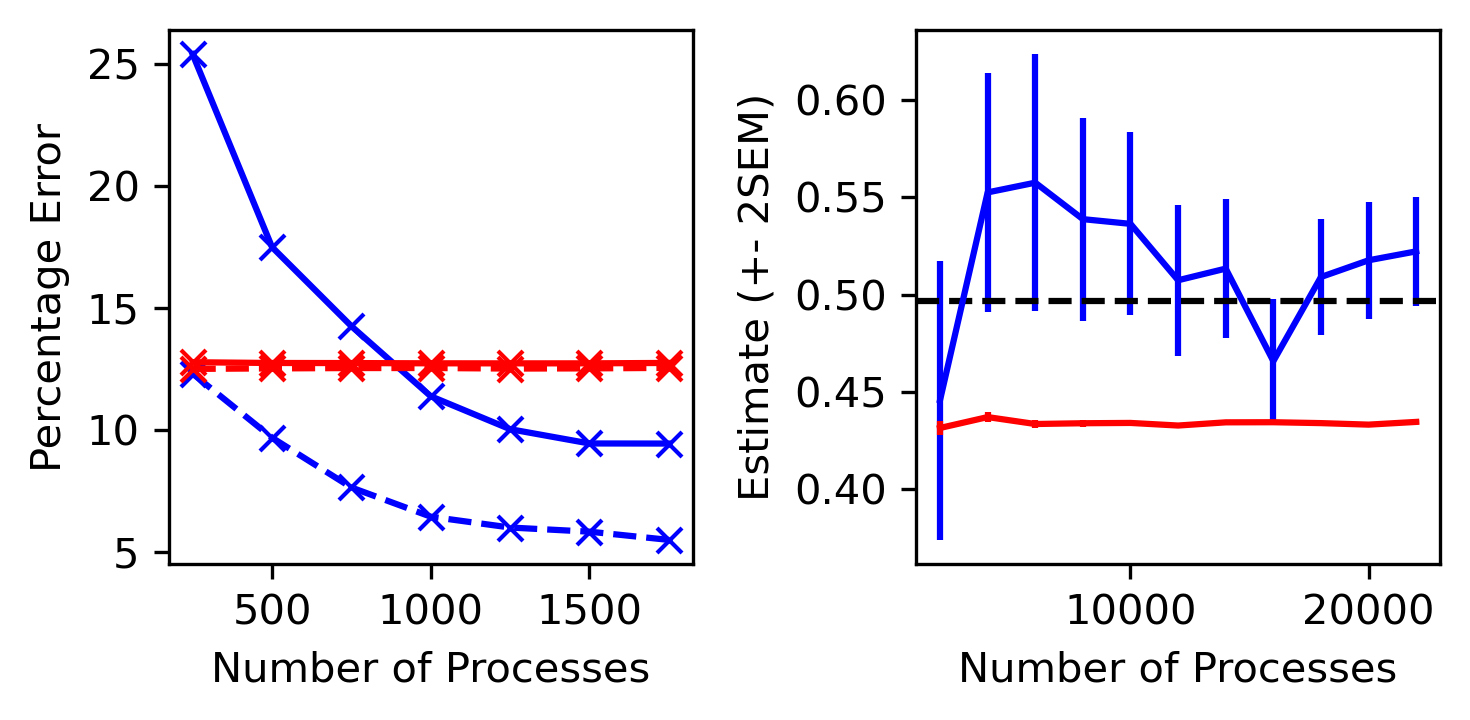}
		\caption{RMSE and intervals $\CC{0}{21}$ estimation on \textsc{gene}}
	\end{subfigure}
	\begin{subfigure}{0.48\linewidth}
		\includegraphics[scale=0.67]{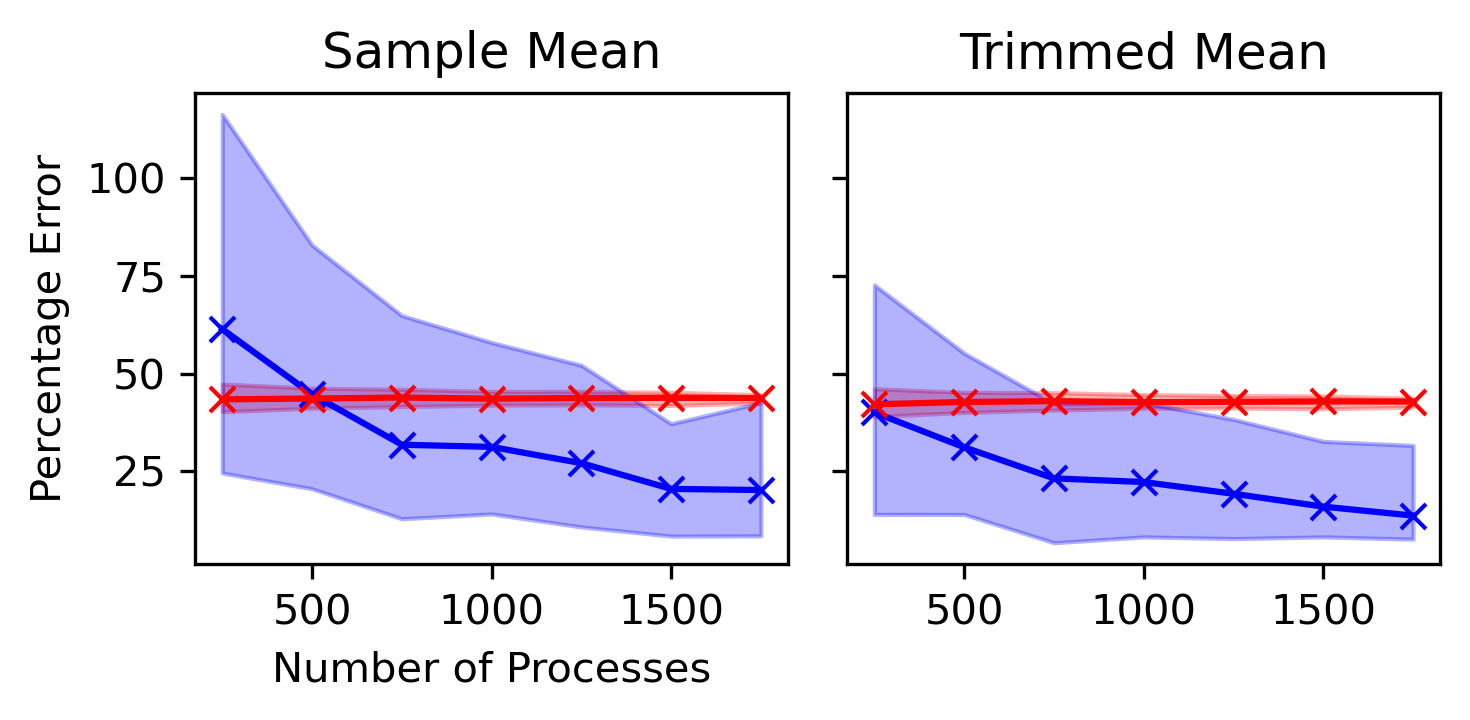}
		\caption{Losses, $\CC{0}{1}$ estimation on \textsc{synthetic}}
	\end{subfigure}
	\begin{subfigure}{0.48\linewidth}
		\includegraphics[scale=0.67]{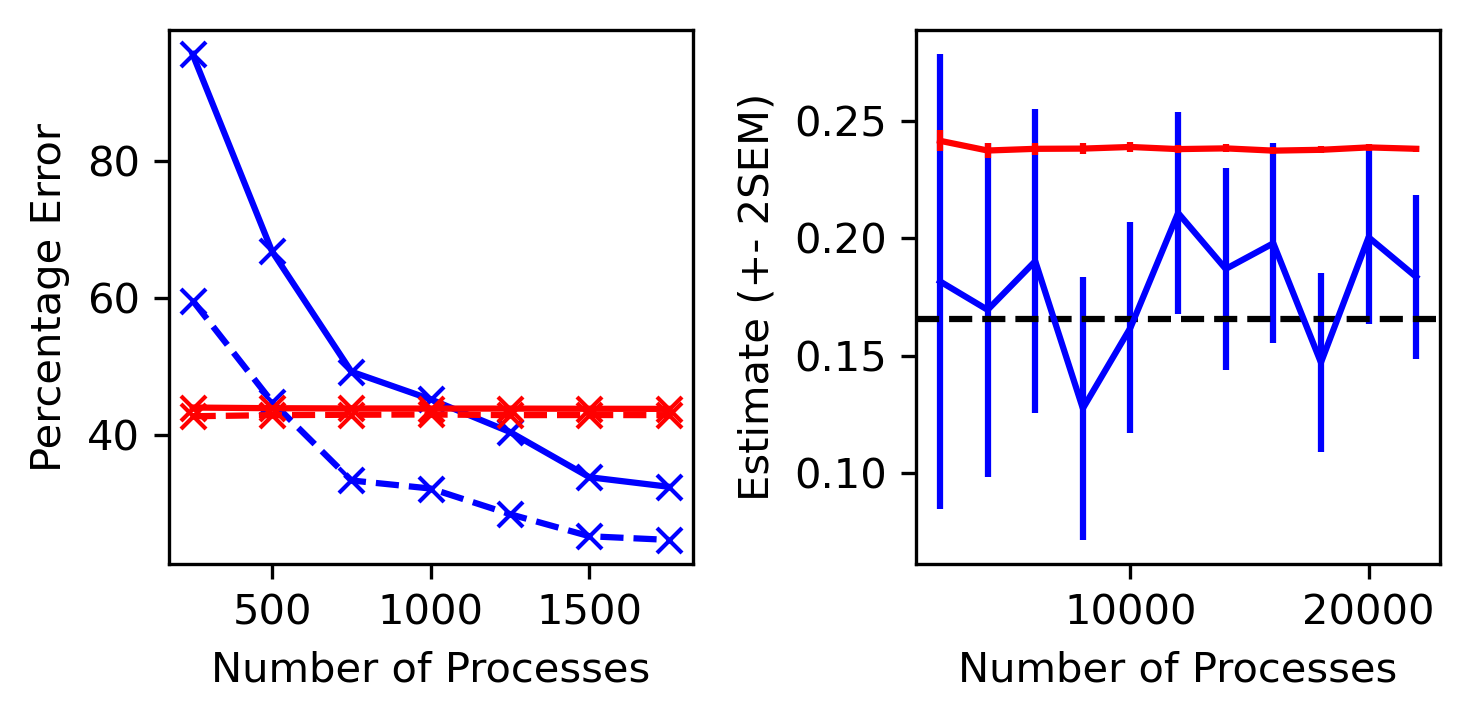}
		\caption{RMSE and intervals, $\CC{0}{1}$ estimation on \textsc{synthetic}}
	\end{subfigure}
	\begin{subfigure}{0.48\linewidth}
		\includegraphics[scale=0.67]{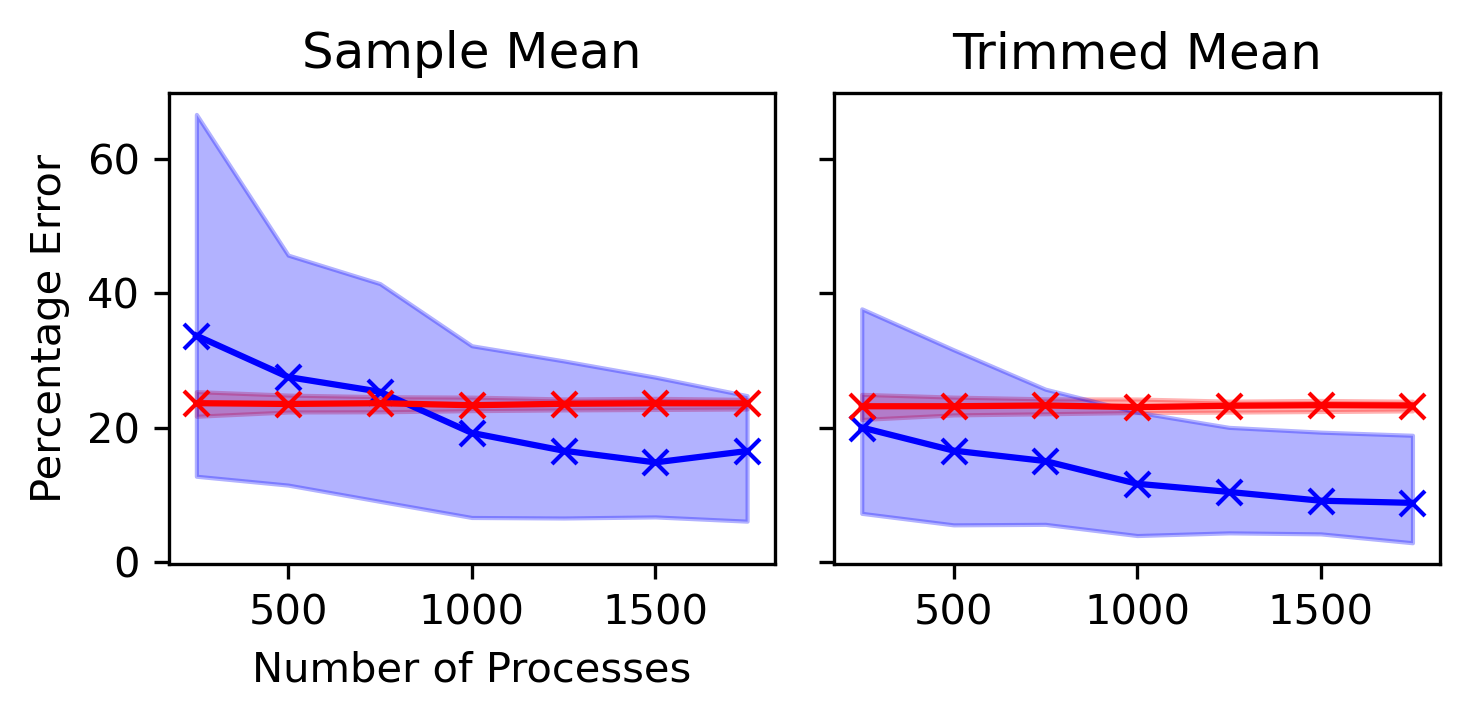}
		\caption{Losses, $\CC{0}{19}$ estimation on \textsc{seed}}
	\end{subfigure}
	\begin{subfigure}{0.48\linewidth}
		\includegraphics[scale=0.67]{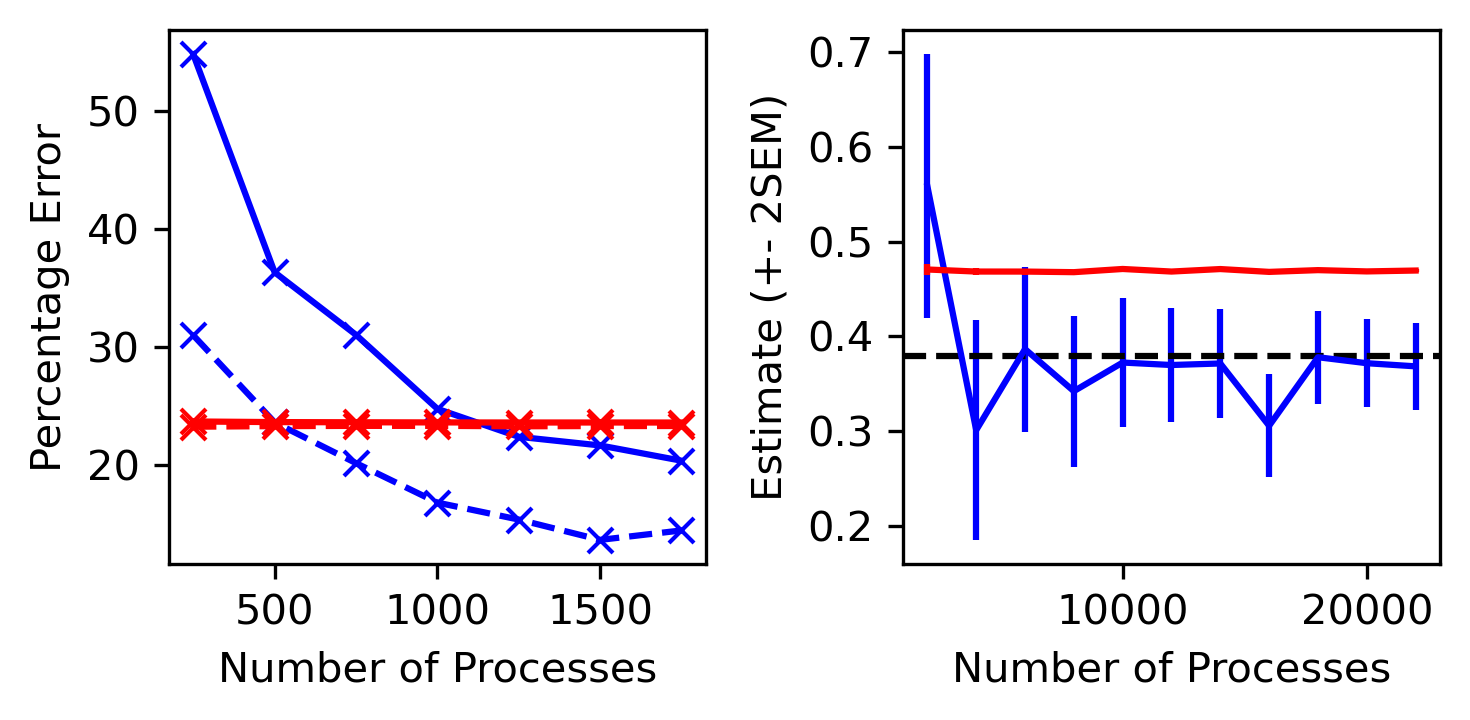}
		\caption{RMSE and intervals, $\CC{0}{19}$ estimation on \textsc{seed}}
	\end{subfigure}
	\caption{Co-clustering results for clustering data sets.}
	\label{fig:co-cluster}
\end{figure}

\subsection{Different Minimum Iteration ($\minIter$) Settings} \label{apd:diffMinIter}
In \Cref{fig:diffMinIter}, with an increase in $\minIter$ (from the default {100} to {150}), the bias in the naive parallel approach reduces (percentage error goes from $15\%$ to $10\%$, for instance), 
and the variance of coupled chains' estimates also reduce.

\begin{figure}[h]
	\centering
	\begin{subfigure}{0.23\linewidth}
		\includegraphics[scale=0.67]{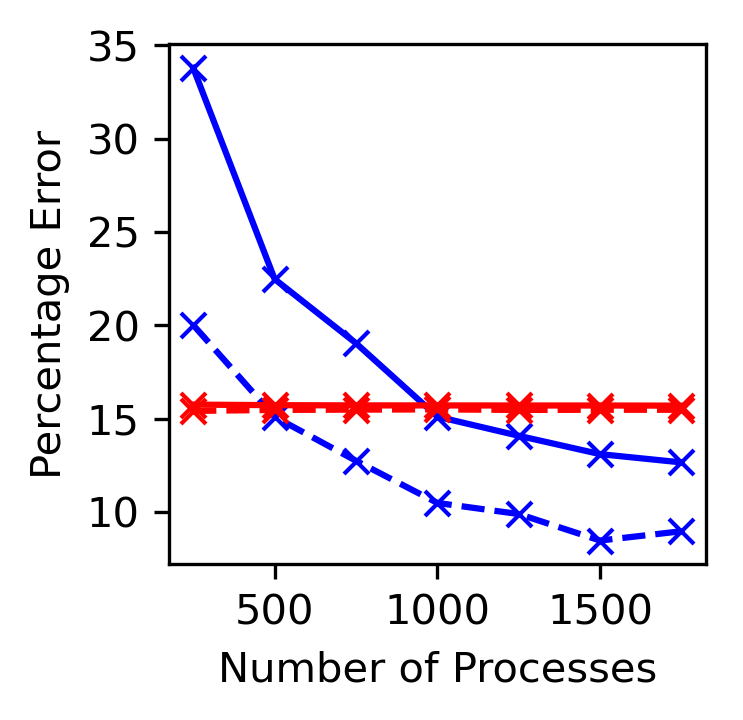}
		\caption{$\minIter = 100$, \textsc{seed}}
	\end{subfigure}
	\begin{subfigure}{0.23\linewidth}
		\includegraphics[scale=0.67]{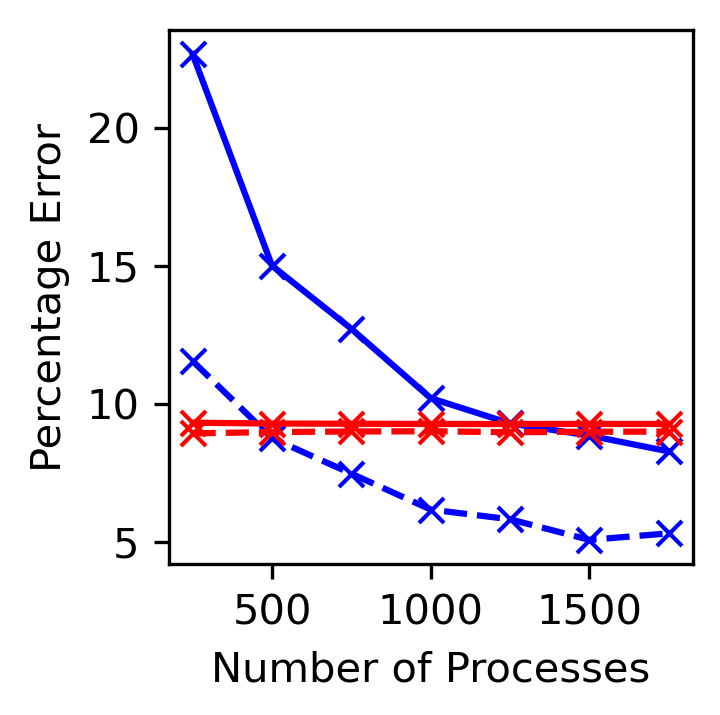}
		\caption{$\minIter = 150$, \textsc{seed}}
	\end{subfigure}
	\begin{subfigure}{0.23\linewidth}
		\includegraphics[scale=0.67]{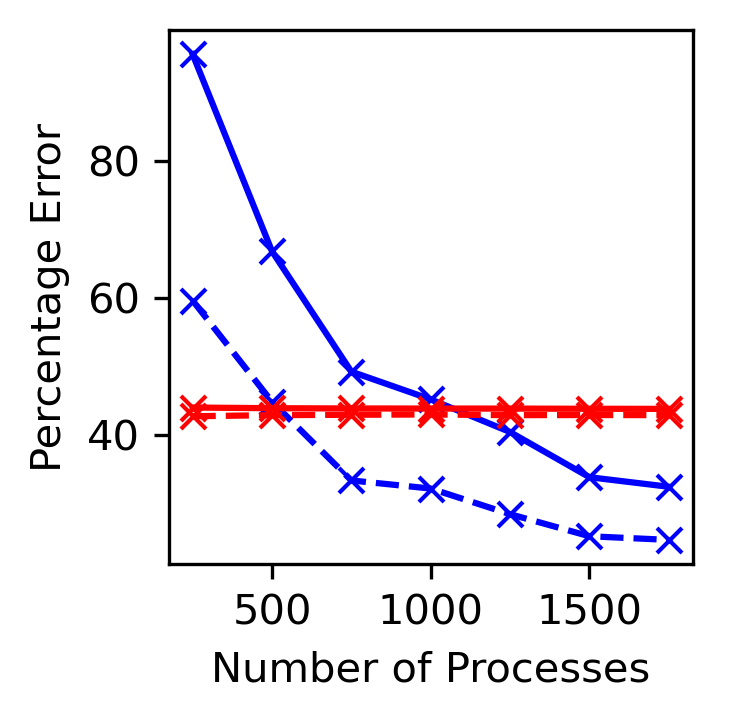}
		\caption{$\minIter = 100$, \textsc{synthetic}}
	\end{subfigure}
	\begin{subfigure}{0.23\linewidth}
		\includegraphics[scale=0.67]{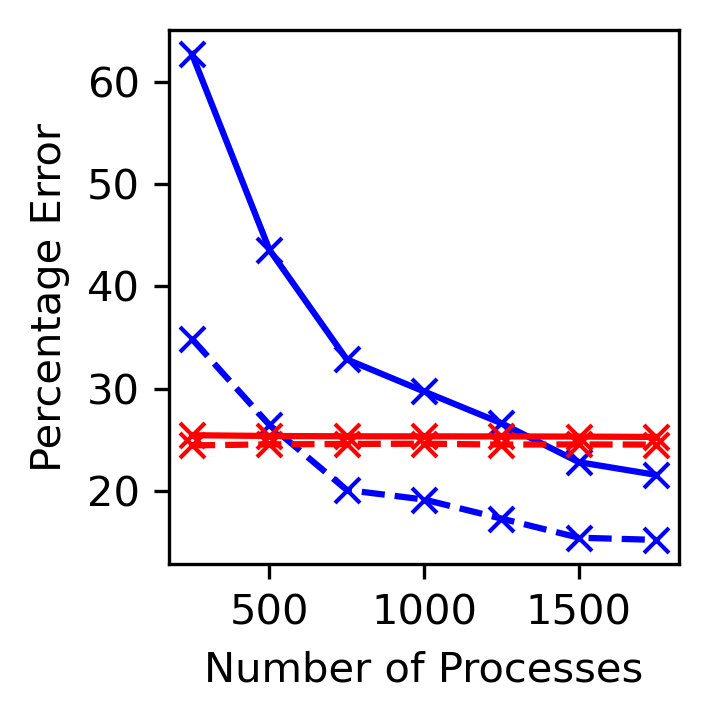}
		\caption{$\minIter = 150$, \textsc{synthetic}}
	\end{subfigure}
	\caption{Impact of different $\minIter$ on the RMSE. The first two panels are $\LCP$ estimation for \textsc{seed}. The last two panels are $\CC{0}{1}$ estimation for \textsc{synthetic}.}
	\label{fig:diffMinIter}
\end{figure}

\subsection{Different Initialization}
In \Cref{fig:kmeans=5-init}, we initialize the Markov chains with the clustering from a k-means clustering with $5$ clusters, instead of the one-component initialization.
Also see \Cref{fig:gene-splitMerge} for more kmeans initialization results. 
The bias from naive parallel is smaller than when initialized from the one-component initialization (RMSE in \Cref{fig:kmeans=5-init} is around $5\%$ while RMSE in \Cref{fig:gene-all} is about $10\%$). 
However, the bias is still significant enough that even with a lot of processors, naive parallel estimates are still inadequate.
\begin{figure}[t]
	\centering
	\includegraphics[scale=0.67]{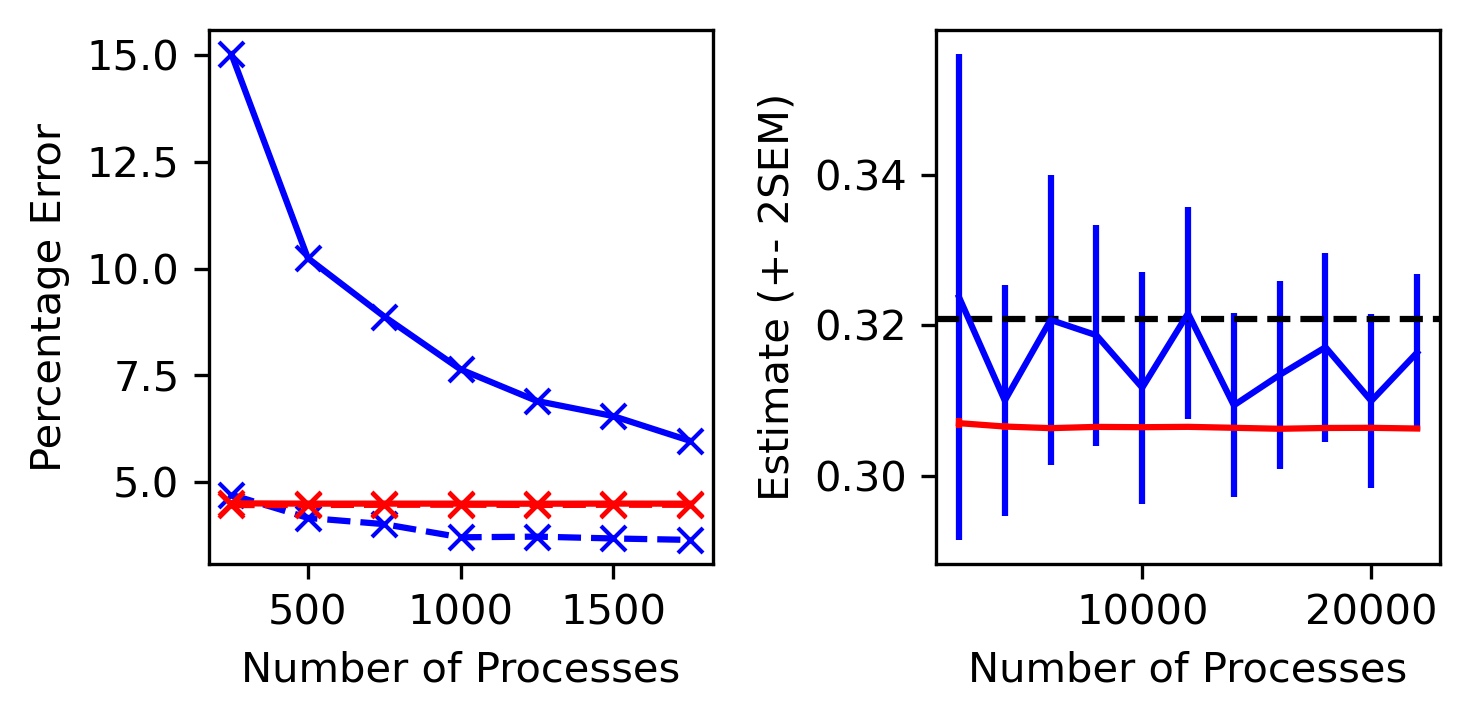}
	\caption{RMSE and intervals for \textsc{gene} on k-means initialization.}
	\label{fig:kmeans=5-init}
\end{figure}

\subsection{Different DPMM Hyperparameters}
For convenience, throughout our experiments, we use diagonal covariance matrices $\Sigma_0 = s_0 I_D$ and $\Sigma_1 = s_1 I_D$, where the variances in different dimensions are the same.
We find that the bias of standard MCMC is influenced by $s_0$ and $s_1$: some settings cause naive parallel chains to have meaningfully large bias, while others do not.
\Cref{fig:diffDPMM} illustrates on \textsc{synthetic} that when $s_1$ is small compared to $s_0$, standard MCMC actually has small bias even when run for a short amount of time.
For values of $s_1$ that are closer to (or larger than $s_0$), the bias in standard MCMC is much larger. 
$\minIter$ and $\burnin$ are set to be 100 and 10 across these settings of $s_1$.

\begin{figure}[h]
	\centering
	\begin{subfigure}{0.31\linewidth}
		\includegraphics[scale=0.67]{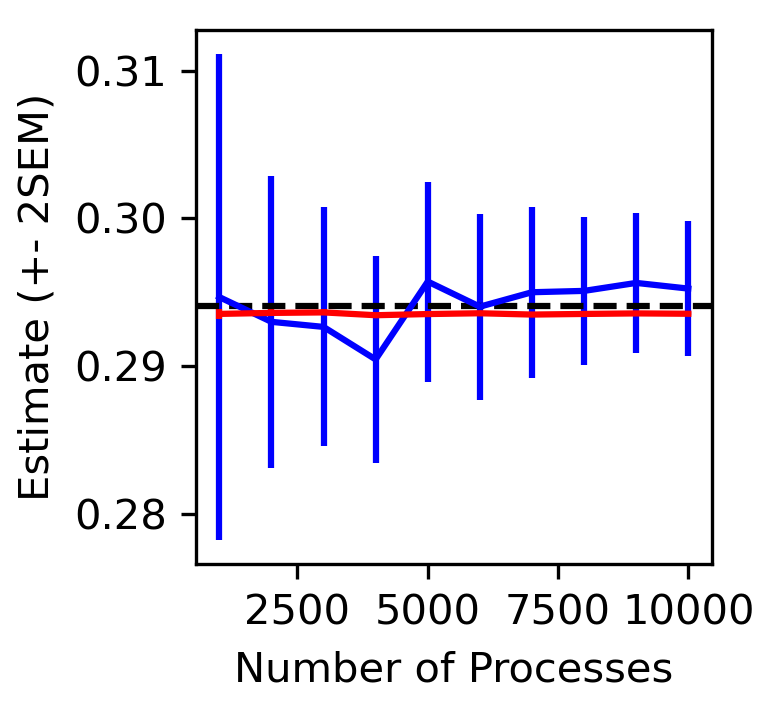}
		\caption{$s_1 = 0.5, s_0 = 0.75$}
	\end{subfigure}
	\begin{subfigure}{0.31\linewidth}
		\includegraphics[scale=0.67]{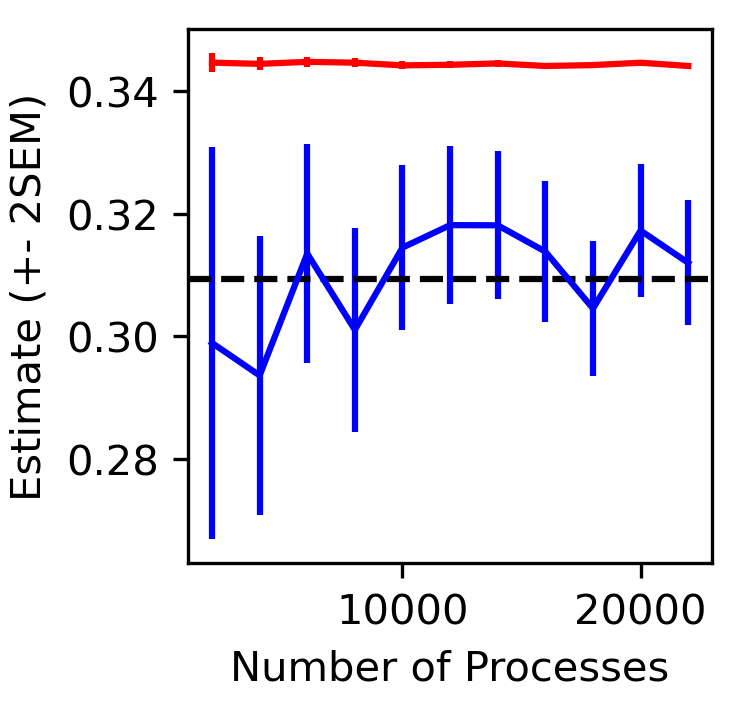}
		\caption{$s_1 = 0.7, s_0 = 0.75$}
	\end{subfigure}
	\begin{subfigure}{0.31\linewidth}
		\includegraphics[scale=0.67]{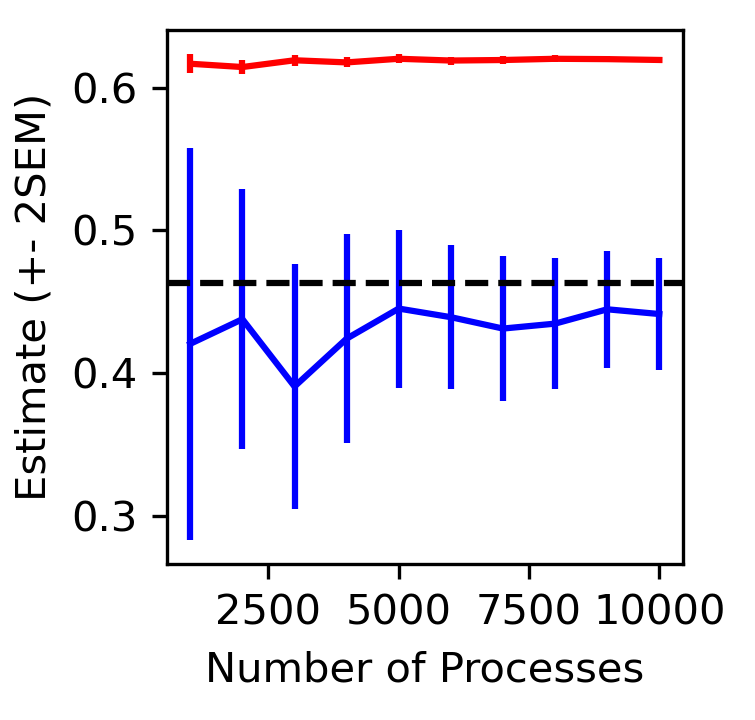}
		\caption{$s_1 = 0.85, s_0 = 0.75$}
	\end{subfigure}
	\caption{The bias in naive parallel estimates is a function of the DPMM hyperparameters.}
	\label{fig:diffDPMM}
\end{figure}

\section{MORE MEETING TIME PLOTS} \label{apd:more-meeting}
In \Cref{fig:meeting-ER}, we generate Erd\H{o}s-R\'{e}nyi random graphs, including each possible edge with probability $0.2$.
The graph in the first two panels has $N = 25$ vertices, while the one in the latter two panels has $N = 30.$
We determine a sufficient number of colors by first greedily coloring the vertices.
It turns out that 6 colors is sufficient to properly color the vertices in either set of panels.

\begin{figure}[h]
	\centering
	\begin{subfigure}{0.48\linewidth}
		\includegraphics[scale=0.8]{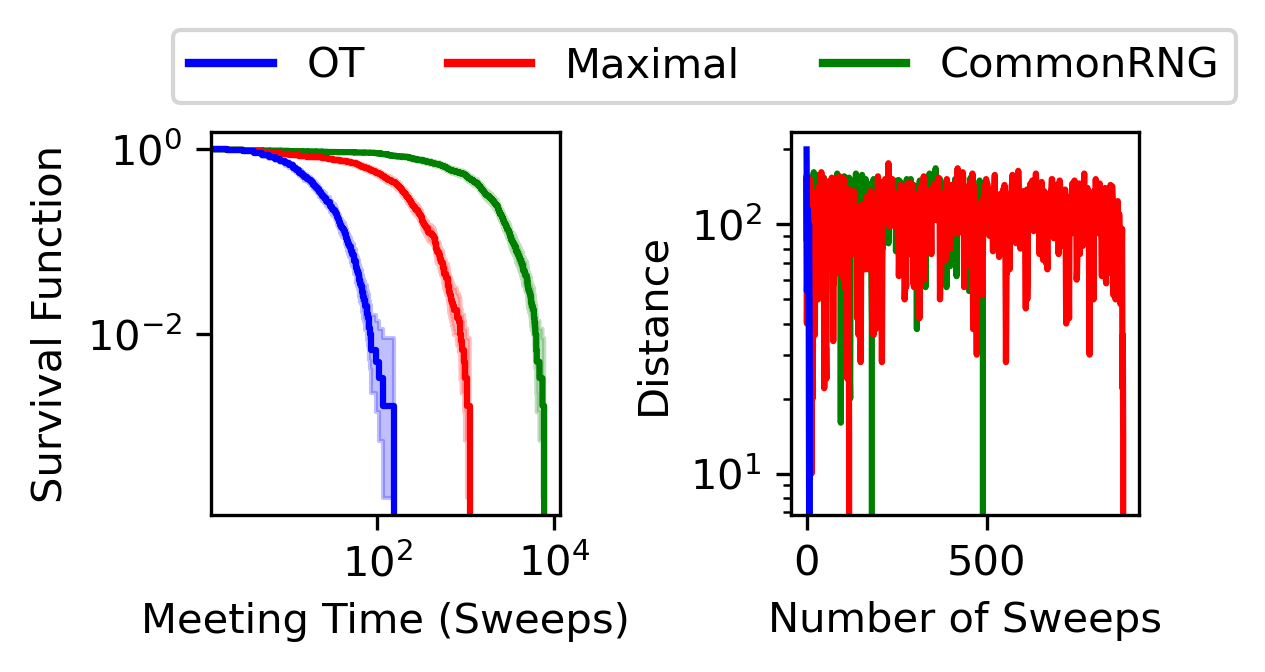}
		\caption{$N = 25$}
	\end{subfigure}
	\begin{subfigure}{0.48\linewidth}
		\includegraphics[scale=0.8]{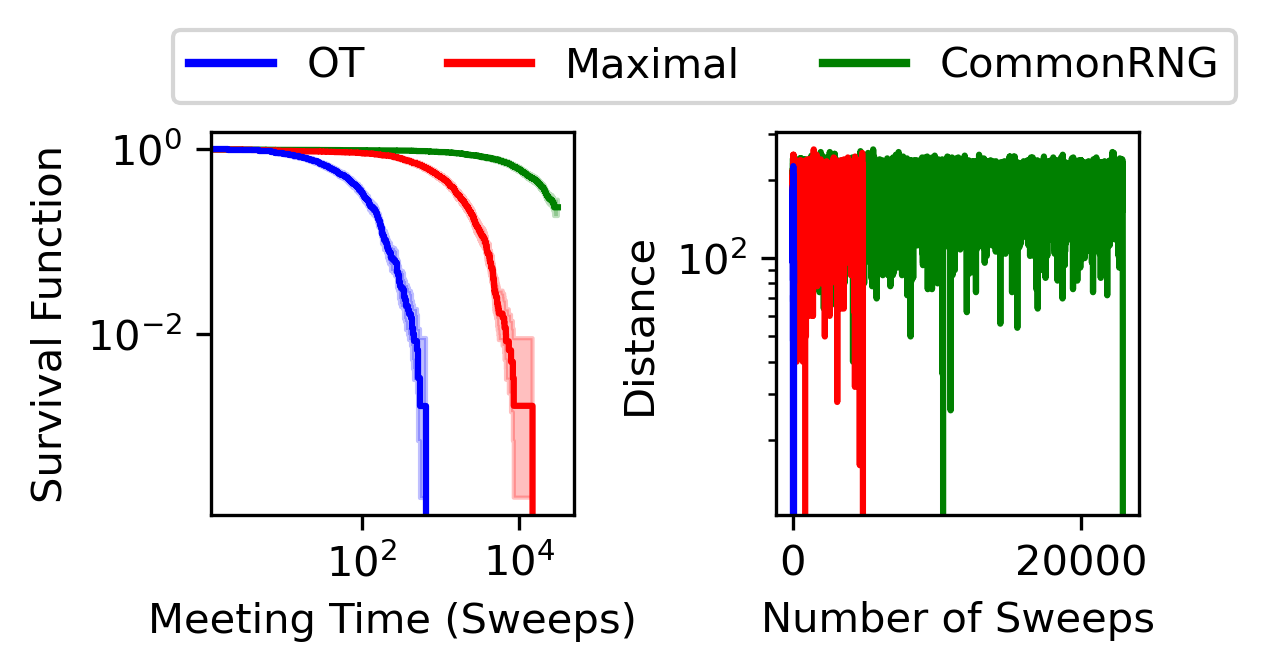}
		\caption{$N = 30$}
	\end{subfigure}
	\caption{Meeting time under OT coupling is better than alternative couplings on Erdos--Renyi graphs, indicated by the fast decrease of the survival functions.}
	\label{fig:meeting-ER}
\end{figure}

\section{ESTIMATES OF PREDICTIVE DENSITY} \label{apd:ppd}
\subsection{Data, Target Model, And Definition Of Posterior Predictive}
As the posterior predictive is easiest to visualize in one dimension, we draw artificial data from a univariate, $10$-component Gaussian mixture model with known observational noise standard deviation $\sigma = 2.0$, and use a DPMM to analyze this data.
The cluster proportions were generated from a symmetric Dirichlet distribution with mass $1$ for all $10$-coordinates. 
The cluster means were randomly generated from $\distNorm{0}{10^2}$.
Since this is an artificial dataset, we can control the number of observations: we denote \textsc{gmm-100} to be the dataset of 100 observations, for instance.

The target DPMM has $\mu_0 = 0, \alpha = 1, \Sigma_0 = 3.0$, and $\Sigma_1 = 2.0$

The function of interest is the posterior predictive density
\begin{equation} \label{eq:pred-def}
	\Pr(\data_{N+1} \in dx \given \data_{1:N}) = \sum_{\Pi_{N+1}} \Pr(\data_{N+1} \in dx \given \Pi_{N+1}, \data_{1:N}) \Pr(\Pi_{N+1} \given \data_{1:N}).
\end{equation}
In \Cref{eq:pred-def}, $\Pi_{N+1}$ denotes the partition of the data $\data_{1:(N+1)}$. To translate \Cref{eq:pred-def} into an integral over just the posterior over $\Pi_{N}$ (the partition of $\data_{1:N}$) we break up $\Pi_{N+1}$ into $(\Pi_{N}, Z)$ where $Z$ is the cluster indicator specifying the cluster of $\Pi_{N}$ (or a new cluster) to which $\data_{N+1}$ belongs. Then
\begin{equation*}
	\Pr(\data_{N+1} \in dx \given \data_{1:N}) = \sum_{\Pi_N} \left[ \sum_{Z}  \Pr(\data_{N+1} \in dx, Z \given \Pi_{N}, \data_{1:N}) \right] \Pr(\Pi_N \given \data_{1:N})
\end{equation*}
Each $\Pr(\data_{N+1} \in dx, Z \given \Pi_{N}, \data_{1:N})$ is computed using the prediction rule for the CRP and Gaussian conditioning. Namely
\begin{equation*}
	\Pr(\data_{N+1} \in dx, Z \given \Pi_{N}, \data_{1:N}) = \underbrace{\Pr(\data_{N+1} \in dx \given Z, \Pi_{N}, \data_{1:N})}_{\text{Posterior predictive of Gaussian}} \times \underbrace{\Pr(Z \given \Pi_{N})}_{\text{CRP prediction rule}}.
\end{equation*}
The first term is computed with the function used during Gibbs sampling to reassign data points to clusters. In the second term, we ignore the conditioning on $\data_{1:N}$, since $Z$ and $\data_{1:N}$ are conditionally independent given $\Pi_{N}.$

\subsection{Estimates Of Posterior Predictive Density}
We first discretize the domain using $150$ evenly-spaced points in the interval $[-20, 30]$: these are the locations at which to evaluate the posterior predictive.
We set $\minIter = 100$ and $\burnin = 10$ in constructing the estimate from \Cref{eqn:unbiased_estimate}. 
We average the results from $400$ coupled chain estimates.
In each panel of \Cref{fig:predictives}, the solid blue curve is an unbiased estimate of the posterior predictive density: the error across replicates is very small and we do not plot uncertainty bands.
The black dashed curve is the true density of the population i.e.\ the $10$-component Gaussian mixture model density.
The grey histogram bins the observed data. 
\begin{figure}[h]
	\centering
	\begin{subfigure}{0.31\linewidth}
		\includegraphics[scale=0.6]{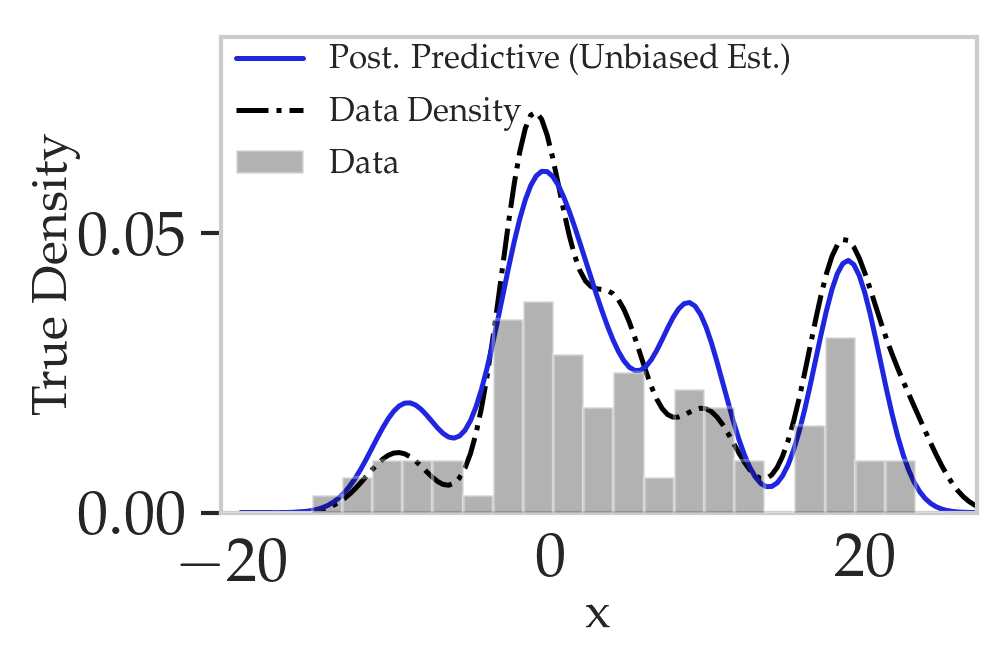}
		\caption{\textsc{gmm-100}}
	\end{subfigure}
	\begin{subfigure}{0.31\linewidth}
		\includegraphics[scale=0.6]{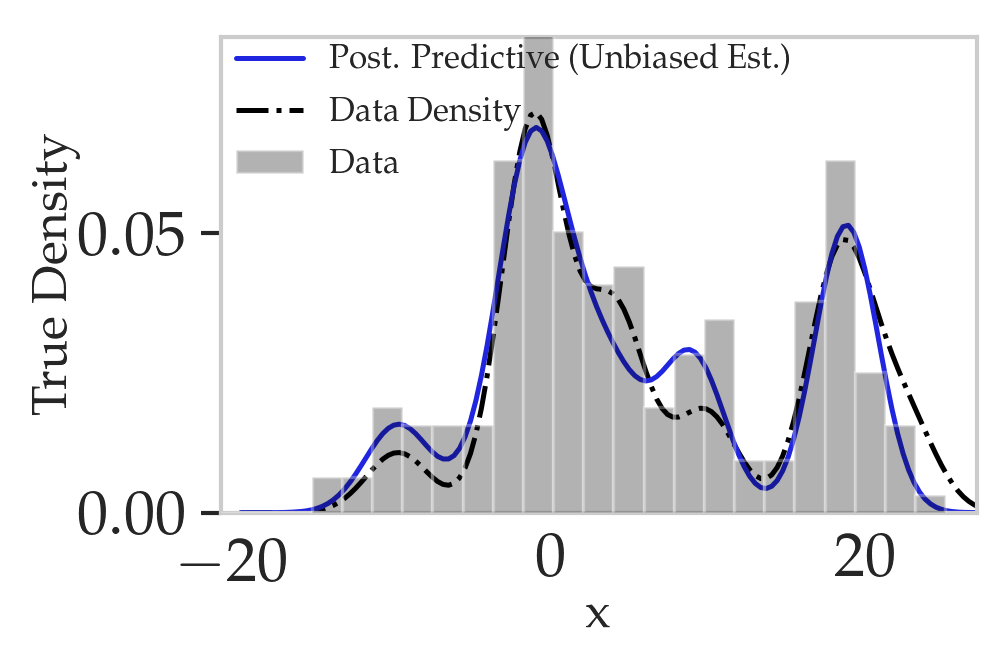}
		\caption{\textsc{gmm-200}}
	\end{subfigure}
	\begin{subfigure}{0.31\linewidth}
		\includegraphics[scale=0.6]{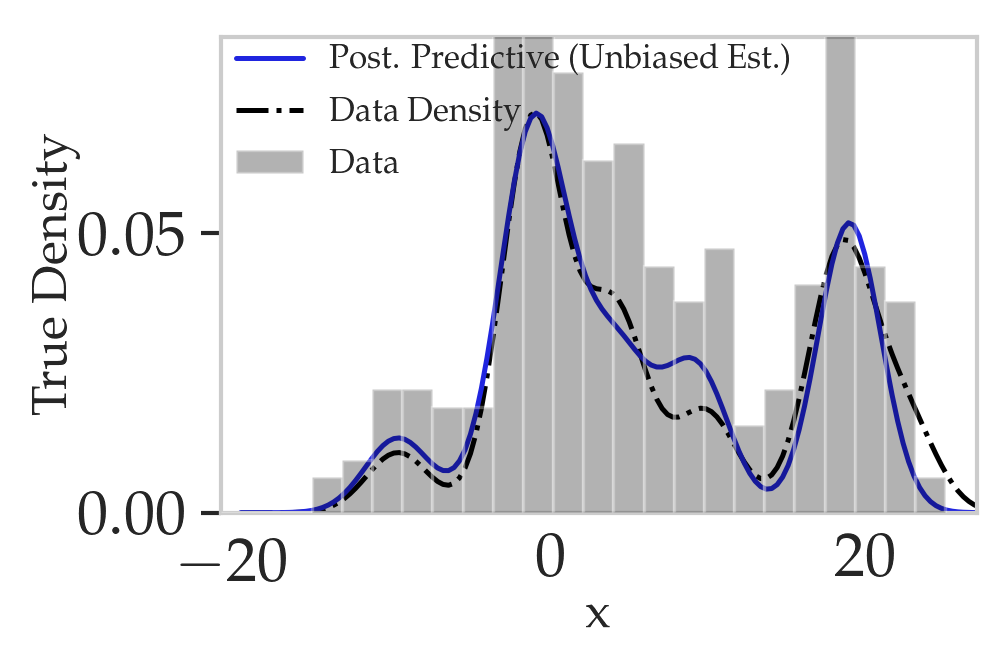}
		\caption{\textsc{gmm-300}}
	\end{subfigure}
	\caption{Posterior predictive density for different number of observations $N$.}
	\label{fig:predictives}
\end{figure}

\subsection{Posterior Predictives Become More Alike True Data Generating Density}
In \Cref{fig:predictives}, by visual inspection, the distance between the posterior predictive density and the underlying density decreases as $N$ increases. 
This is related to the phenomenon of posterior concentration, where with more observations gathered, the Bayesian posterior concentrates more and more on the true data generating process.
We refer to \citet{ghosal1999posterior,lijoi2005consistency} for more thorough discussions of posterior concentration. 
In what follows, we justify the concentration behavior for Gaussian DPMM, when the observation noise is correctly specified.

\begin{theorem}[{DP mixtures prior is consistent for finite mixture models}] \label{thm:DPMM-consistent-finite-mixture}
	Let $f_0(x) \defined \sum_{i=1}^{m} p_i \mathcal{N}(x \given \theta_i,\sigma_1^2)$ be a finite mixture model. Suppose we observe iid data $X_1,\ldots,X_n$ from $f_0$. Consider the following probabilistic model
	\begin{align*}
		\widehat{P} &\sim \distDP{\alpha}{\mathcal{N}(0,\sigma_0^2)} \\
		\theta_i \mid \widehat{P} &\stackrel{iid}{\sim} \widehat{P} & & i = 1,2,\ldots,n \\
		X_i \mid \theta_i &\stackrel{indep}{\sim} \mathcal{N}(\theta_i,\sigma_1^2) & & i = 1,2,\ldots,n
	\end{align*}
	Let $\widehat{P}_n$ be the posterior predictive distribution of this generative process. Then with a.s.\ $P_{f_0}$
	\begin{equation*}
		d_{TV}\left( \widehat{P}_n, P_{f_0}\right) \xrightarrow{n \to \infty} 0.
	\end{equation*}
\end{theorem}

To prove \Cref{thm:DPMM-consistent-finite-mixture}, we first need some definitions and auxiliary results.

\begin{definition}[Strongly consistent priors]
	Suppose iid data $X_1,X_2,\ldots,X_n$ is generated from some probability measure measure that is absolutely continuous with respect to Lebesgue measure. 
	Denote the density of this data generating measure by $f_0$.
	Let $\mathcal{F}$ be the set of all densities on $\mathbb{R}$.
	Consider the probabilistic model where we put a prior $\Pi$ over densities $f$, and observations $X_i$ are conditionally iid given $f$. 
	We use $P_{f}$ to denote the probability measure with density $f$. 
	For any measurable subset $A$ of $\mathcal{F}$, the posterior of $A$ given the observations $X_i$ is denoted $\Pi(A \given X_{1:N})$. 
	A strong neighborhood around $f_0$ is any subset of $\mathcal{F}$ containing a set of the form $V = \{f \in \mathcal{F}: \int | f - f_0| < \epsilon \}$ according to \citet{ghosal1999posterior}.
	The prior $\Pi$ is strongly consistent at $f_0$ if for any strong neighborhood $U$,
	\begin{equation} \label{eq:consistency}
		\lim_{n \to \infty} \Pi(U|X_{1:n})  = 1,
	\end{equation}
	holds almost surely for $X_{1:\infty}$ distributed according to $P_{f_0}^{\infty}$. 
\end{definition}

\begin{proposition} [{\citet[Proposition 4.2.1]{ghosh2003bayesian}}] \label{thm:consistent-predictive}
	If a prior $\Pi$ is strongly consistent at $f_0$ then the predictive distribution, defined as
	\begin{equation} \label{eq:pred-as-density}
		\widehat{P}_n(A \mid X_{1:n}) \coloneqq \int_f P_f(A)  \Pi(f \mid X_{1:n})
	\end{equation}
	also converges to $f_0$ in total variation in a.s. $P_{f_0}^{\infty}$
	\begin{equation*}
		d_{TV}\left( \widehat{P}_n, P_{f_0}\right) \xrightarrow{} 0.
	\end{equation*}
\end{proposition}

The definition of posterior predictive density in \Cref{eq:pred-as-density} can equivalently be rewritten as
\begin{equation*}
	\widehat{P}_n(A \mid X_{1:n}) = \Pr(X_{n+1} \in A \given X_{1:n}),
\end{equation*}
since $P_f(A) = P_f(X_{n+1} \in A)$ and all the $X$'s are conditionally iid given $f$.

We are ready to prove \Cref{thm:DPMM-consistent-finite-mixture}. 
\begin{proof}[{Proof of \Cref{thm:DPMM-consistent-finite-mixture}}]
	First, we can rewrite the DP mixture model as a generative model over continuous densities $f$ 
	\begin{equation} \label{eq:dpmm-as-density}
		\begin{aligned}
			\widehat{P} &\sim \distDP{\alpha}{\mathcal{N}(0,\sigma_0^2)} \\
			f &= \mathcal{N}(0,\sigma_1^2) \ast \widehat{P} 
			\\
			X_i \mid f &\stackrel{iid}{\sim} f & & i = 1,2,\ldots,n
		\end{aligned}
	\end{equation}
	where $\mathcal{N}(0,\sigma_1^2) \ast \widehat{P}$ is a convolution, with density $f(x) \defined \int_{\theta} \mathcal{N}(x - \theta | 0,\sigma_1^2) d\widehat{P}(\theta)$. 
	
	The main idea is showing that the posterior $\Pi(f|X_{1:n})$ is strongly consistent and then leveraging \Cref{thm:consistent-predictive}. For the former, we verify the conditions of \citet[Theorem 1]{lijoi2005consistency}. 
	
	The first condition of \citet[Theorem 1]{lijoi2005consistency} is that $f_0$ is in the K-L support of the prior over $f$ in \Cref{eq:dpmm-as-density}. We use \citet[Theorem 3]{ghosal1999posterior}. Clearly $f_0$ is the convolution of the normal density $\mathcal{N}(0,\sigma_1^2)$ with the distribution $P(.) = \sum_{i=1}^m p_i \delta_{\theta_i}$. $P(.)$ is compactly supported since $m$ is finite. Since the support of $P(.)$ is the set $\{\theta_i\}_{i=1}^{m}$ which belongs in $\mathbb{R}$, the support of $\mathcal{N}(0,\sigma_0^2)$, by \citet[Theorem 3.2.4]{ghosh2003bayesian}, the conditions on $P$ are satisfied. The condition that the prior over bandwidths cover the true bandwidth is trivially satisfied since we perfectly specified $\sigma_1$. 
	
	The second condition of \citet[Theorem 1]{lijoi2005consistency} is simple: because the prior over $\widehat{P}$ is a DP, it reduces to checking that
	\begin{equation*}
		\int_{\mathbb{R}} |\theta| \mathcal{N}(\theta \mid 0, \sigma_0^2) < \infty
	\end{equation*}
	which is true. 
	
	The final condition trivial holds because we have perfectly specified $\sigma_1$: there is actually zero probability that $\sigma_1$ becomes too small, and we never need to worry about setting $\gamma$ or the sequence $\sigma_k$.  
\end{proof}

\end{document}